\documentclass[journal,onecolumn]{IEEEtran}

\usepackage{amsmath,amsfonts,amssymb}
\usepackage{amsthm}
\usepackage{graphicx}
\usepackage{algorithm,algorithmicx,algpseudocode}
\usepackage{comment}
\usepackage{color}

\newtheorem{theorem}{Theorem}
\newtheorem{lemma}{Lemma}

\newtheorem{proposition}{Proposition}

\newtheorem{definition}{Definition}

\newcommand{\floor}[1]{\left\lfloor{#1}\right\rfloor}
\newcommand{\ceil}[1]{\left\lceil{#1}\right\rceil}

\newcommand{\remove}[1]{}
\renewcommand{\epsilon}{\varepsilon}

\graphicspath{ {./figures/} }
\newcommand{\myfigswidth}{.35\textwidth}
\newcommand{\myplotwidth}{.4\textwidth}
\newcounter{MYtempeqncnt}

\begin{document}

\title{Coded Caching for Multi-level Popularity and Access}
\author{
Jad~Hachem,~\IEEEmembership{Student~Member,~IEEE,}
Nikhil~Karamchandani,~\IEEEmembership{Member,~IEEE,}
and~Suhas~Diggavi,~\IEEEmembership{Fellow,~IEEE}%
\thanks{J. Hachem and S. Diggavi are with the Department of Electrical Engineering, University of California, Los Angeles.}%
\thanks{N. Karamchandani is with the Department of Electrical Engineering, Indian Insitute of Technology, Bombay.}%
\thanks{This work was supported in part by NSF grant \#1423271 and a gift by Qualcomm Inc.}}

\maketitle

\begin{abstract}
To address the exponentially rising demand for wireless content, use
of caching is emerging as a potential solution.  It has been recently
established that joint design of content delivery and storage (coded
caching) can significantly improve performance over conventional
caching.  Coded caching is well suited to emerging heterogeneous
wireless architectures which consist of a dense deployment of
local-coverage wireless access points (APs) with high data rates,
along with sparsely-distributed, large-coverage macro-cell base
stations (BS). This enables design of coded caching-and-delivery
schemes that equip APs with storage, and place content in them
in a way that creates
coded-multicast opportunities for combining with macro-cell broadcast
to satisfy users even with different demands.  Such coded-caching
schemes have been shown to be order-optimal with respect to the BS
transmission rate, for a system with
single-level content, i.e., one where all content is uniformly
popular. In this work, we consider a system with non-uniform
popularity content which is divided into multiple levels, based on
varying degrees of popularity. The main contribution of this work is
the derivation of an order-optimal scheme which judiciously shares
cache memory among files with different popularities.  To show order-optimality
we derive new information-theoretic lower bounds, which
use a sliding-window entropy inequality, effectively creating a
non-cutset bound. We also extend the ideas to when users can access
multiple caches along with the broadcast.  Finally we consider two
extreme cases of user distribution across caches for the multi-level
popularity model: a single user per cache (single-user setup) versus a
large number of users per cache (multi-user setup), and demonstrate a
dichotomy in the order-optimal strategies for these two extreme cases.

\end{abstract}

\section{Introduction}
\label{sec:intro}
Broadband data consumption has witnessed a tremendous growth over the
past few years, due in large part to multi-media applications such as
Video-on-Demand. This increased demand has been managed in the wired
internet via Content Distribution Networks (CDNs), by mirroring data
in various locations and in effect pushing the content closer to the
end users. Wireless data consumption, driven by the increased demand
for high-definition content on mobile devices, has also grown at a
significant rate \cite{CiscoReport} and is testing the limits of our
underlying wireless communication
systems \cite{QualcommSmallCells}. However, simply borrowing the CDN
solution from wired networks and applying it to wireless systems is
insufficient to solve the wireless content delivery problem.  In the
wired Internet, CDNs remove the bottleneck at the content distribution
server by utilizing repeated demand of particular content.  It has the
most gains when the local communication link is not the
bottleneck \cite{korupolu1999}. In wireless cellular usage, this is
typically not true as the (cellular) wireless hop is a bottleneck
link.

The \emph{broadcast} nature of wireless can be used as an advantage to
alleviate this problem. This, along with the
emerging \emph{heterogeneous} wireless network can be used to provide
an architecture for wireless content distribution. The heterogeneous
wireless network (HetNet) architecture emerging for 5G consists of
a dense deployment of wireless access points (APs) with small coverage
and relatively large data rates, in combination with cellular
base-stations (BS) with large coverage and smaller data rates. For
example, the access points could be WiFi or emerging small-cells (or
femto-cells), which provide high data rate for short ranges.  The
consequence of this emerging architecture is that a user could
potentially receive broadcast from the BS as well as connect to
(several) wireless APs. Therefore, we could place caches at local APs
and complement them with macro-cellular (BS) broadcast.

In this paper we study a problem based on an architecture where
content is stored at multiple APs without \emph{a priori} knowing the user
requests, and the base-station broadcast is used judiciously to
complement the local caching, \emph{after} the user requests are
known. This is motivated by the new approach initiated
in \cite{maddah-ali2012, maddah-ali2013} where it has been shown that joint
design of storage and delivery (a.k.a. ``coded caching'') can
significantly improve content delivery rate requirements. This was
enabled by content placement that creates (network-coded) multicast
opportunities among users with access to different storage units, even
when they have different (and \emph{a priori} unknown) requests. This enables
an examination of the optimal trade-off between cache memory size and
broadcast delivery rate.

The setup studied in \cite{maddah-ali2012,maddah-ali2013} consisted of
single-level content, i.e., every file in the system is
uniformly demanded.  However, it is well understood that content
demand is non-uniform in practice, with some files being more popular
than others.  Motivated by
this, \cite{niesen2013,Zcodedcaching,ZhangArbitrary,MLcodedcaching,
HKDinfocom15, MLCCusers} considered such
non-uniform content demand, following different models. 
In \cite{niesen2013,Zcodedcaching,ZhangArbitrary}, the setup
considered a single user per cache requesting a file independently and
randomly according to some (arbitrary) probability distribution that
represents content popularity.  These works studied the trade-off
between the average rate and the cache memory.  A memory-sharing
scheme was proposed in \cite{niesen2013}, and its achievable rate was
characterized.  However, from our understanding, this scheme was not
shown to be order-optimal in general.%
\footnote{We refer to an ``order-optimal'' result as
one that is within a constant multiplicative factor from the
information-theoretic optimum. The constant is to be independent of
the number of users, caches, memory size and number of popularity
levels.}
In \cite{Zcodedcaching}, a different scheme was proposed,
based on a clustering of the most popular files into a single content
level, which was shown to
be order-optimal for Zipf-distributed content and, more recently, for
arbitrary distributions in \cite{ZhangArbitrary}.

By contrast, in \cite{MLcodedcaching}, a deterministic multi-level
popularity model was introduced (simultaneous to the aforementioned
other non-uniform popularity models), where content is divided into
discrete levels based on popularity. In this paper, we focus on this
model and study it mostly in the context where a large number of
users connect to each cache (``multi-user model''), and, for each
level, a fixed and \emph{a priori} known fraction of the users per
cache request files from said level.  It is easy to see that, when the
number of users per cache is large enough, this deterministic model
will closely approximate an equivalent stochastic-demands model
similar to \cite{niesen2013,Zcodedcaching,ZhangArbitrary}. We will
also study the scenario where users could connect to multiple access
points (caches) as well as listen to the broadcast to get the desired
content. In short, the setup considered has a different popularity
model and user population as well as cache access than considered in
earlier literature. We will also compare the results, for this
popularity model, between setups with many users per cache (multi-user
setup) and a single user per cache (single-user setup).

The main contribution of this work is, for any given multi-level
content popularity profile, to approximately solve the trade-off of
the transmission cost at the BS with the storage cost at the APs. In
addition, we also approximately solve the case where users have access
to multiple APs. Finally, we study the effect of number of users per
cache in the multi-level content popularity model. In particular, the
following are the core technical contributions of the
work:%
\footnote{Shorter versions of these results have been published
in \cite{MLcodedcaching, HKDinfocom15, MLCCusers}.}.
\begin{itemize}
\item Non-cut-set-based information-theoretic lower bounds, for the multi-user
model with users accessing multiple caches, which are used to evaluate the performance of the proposed schemes.
\item A memory-sharing scheme, which divides the available storage at each cache (AP)
among the various popularity levels. A striking aspect of this
(order-optimal, for the multi-user setup) solution is that,
in some regimes, it is better to
store some less popular content without completely storing the more
popular content, even when cache memory is available.
\item We demonstrate order-optimality of the scheme with respect to the
information-theoretic lower bound that is independent of the number of
popularity levels, number of users, files, and caches.
\item We demonstrate that drastically different strategies are order-optimal
for the multi-user and single-user setup.  In the
single-user case, we show that clustering the most popular levels and
giving them all the memory, leaving none for the rest, is
order-optimal; a strategy proposed
in \cite{Zcodedcaching,ZhangArbitrary}. In contrast, the multi-user
case requires a complete separation of the different levels and a
division of the memory between them.

\end{itemize}

The paper is organized as follows.  Section~\ref{sec:setup} formulates
the problem, describing precisely the multi-user and single-user
setups.  We establish some background in
Section~\ref{sec:preliminaries}, which enables us to state the main
results in Section~\ref{sec:results}.  The caching and delivery
strategy for the multi-user set up, as well as corresponding lower bounds,
are given in Section~\ref{sec:multi-user},
while the single-user setup is studied in Section~\ref{sec:single-user}.  A brief
discussion about the dichotomy in the two setups is given in
Section~\ref{sec:comparison}. The paper concludes in
Section \ref{sec:discussion} with a discussion and some numerical
evaluations to interpret the results.  Many of the detailed proofs are
given in the appendices.

\subsection{Related work}

Content caching has a rich history and has been studied extensively,
see for example \cite{Wessels:2001} and references therein. More
recently, it has been studied in the context of Video-on-Demand
systems where efficient content placement and delivery schemes have
been proposed in \cite{korupolu1999, Borst:2010, Tan:2013,
llorca2013}. The impact of content popularity distributions on caching
schemes has also been widely investigated, see for
example \cite{Wolman99, breslau1999web, Applegate:2010}. 

Most of the literature has focused on wired networks and, as
argued before, the solutions there do not carry directly to wireless
networks. Recently, \cite{FemtoCaching} proposed a caching
architecture for heterogeneous wireless networks, with the small-cell
or WiFi access points acting as helpers by storing part of the
content. A content placement scheme is formulated and posed as a
linear program. However, the (information-theoretic) optimality of
such schemes was not examined in that work. 
%The idea proposed in this paper differs
%from \cite{FemtoCaching} in several aspects: utilizing the macro-cell
%base station broadcast to assist in content delivery, and allowing
%different access structures for different popularity classes, which
%help improve the system performance significantly. Moreover, our
%scheme is compared against the best possible through
%information-theoretic impossibility results, which do not have any
%restrictions on the structure of the placement and delivery schemes.
Another aspect (also common to most of the papers in the content
caching literature) is that the delivery phase used independent
unicasts to serve the different users. The important observation to
utilize broadcast to improve system performance by serving multiple
users simultaneously was made in \cite{ maddah-ali2013,
maddah-ali2012}. They initiated the study of coded caching where
joint design of storage and delivery was considered for the case with
a single level of files and single cache access during delivery by
proposing an order-optimal coded caching scheme.  These results have
been extended to online caching systems in~\cite{PMN13},
heterogeneous cache sizes \cite{wang2015fundamental}, unequal file
sizes \cite{zhang2015filesize}, and improved converse
arguments \cite{ghasemi2015improved, sengupta2015improved}.
Efficient coded caching schemes have been devised in \cite{vettigli2015efficient},
and the effect of finite file sizes has been investigated in \cite{shanmugam2015finite}.
Content caching and delivery has also been studied for hierarchical
tree topologies \cite{Hcodedcaching,KNMD14}, device to device networks
\cite{ji2013wireless,jeon2015d2d}, multi-server topologies
\cite{shariatpanahi2015multi}, and heterogeneous wireless networks
\cite{HKDinfocom15}.

Coded caching was extended to non-uniform popularity models in
\cite{niesen2013,Zcodedcaching,ZhangArbitrary,ji2015random}, where the setup
considered a single user per cache requesting a file independently and
randomly according to some (arbitrary) probability distribution that
represents content popularity. The trade-off between memory and average
delivery rate was studied in these works. Our work differs from these
as it uses a deterministic multi-level popularity model introduced
in \cite{MLcodedcaching}, enabling a worst-case rather than average
case analysis. We analytically characterize the order-optimal
splitting parameters for the memory-sharing scheme, even with user
access to multiple caches. The dichotomy of order-optimal schemes
between having multiple users per cache and a single user per cache is
also demonstrated for this multi-level popularity model.

Other related work includes \cite{Gitzenis13} which derives scaling
laws for content replication in multihop wireless
networks; \cite{Ioannidis:2010} which explores distributed caching in
mobile networks using device-to-device
communications; \cite{AltmanAG13} which studies the benefit of coded
caching when the caches are distributed randomly; and \cite{YangH13}
which explores the benefits of adaptive content placement, using
knowledge of user requests.

\section{Setup, Notation, and Formulation}
\label{sec:setup}

\subsection{Setup}

Consider a network where a group of users request files from a server.
All files are assumed to be of size $F$ bits.
Prior to any user requests, a \emph{placement phase} occurs in which information about these files is placed in the caches of $K$ access points (APs); each cache has a capacity of $MF$ bits.
Then, in the \emph{delivery phase}, users connect to the different caches, and each requests a file based on an underlying file popularity model: more popular files are more likely to be requested.
The server then sends, through the base station (BS), a broadcast message of size $RF$ bits that all the users can hear.
The users combine the broadcast with the contents of their cache to recover the file that they have requested.
Clearly, there is a trade-off between the values of $M$ (the ``cache memory'') and $R$ (the ``broadcast rate''): the larger the cache memory, the more information the caches can store, and hence the smaller the broadcast rate needed to serve the requests.
Our goal is to characterize this trade-off.

\subsection{Multi-level popularity}
\label{sec:setup-popularity}

In multi-media applications such as video-on-demand, we often find that a small number of files are requested by many more users than the rest of the files.
This difference in popularity can easily influence the caching system described above.
For example, when deciding what to store in the (limited-capacity) caches, one would want to give more of the cache memory to the more popular files, since they will be, on average, requested more often.

Different popularity models have been considered in the literature.
The simplest model was studied in \cite{maddah-ali2012,maddah-ali2013}.
In this model, all files are equally popular, meaning that there is no preference among the users to choose one file over the others.
The results were of a worst-case nature, i.e., they are true for all possible (valid) combinations of user demands.
While this model allowed the first approximate-optimality caching result and introduced the idea of coded caching, it is not an accurate representation of typical multi-media data.
To introduce the effect of file popularity, we studied a multi-level popularity model in \cite{MLcodedcaching,HKDinfocom15,MLCCusers}, in which files are divided into different popularity classes (\emph{levels}), and files within each class are equally popular.
The results are also based on a worst-case analysis.
Finally, probabilistic models were also studied in the literature.
In these, user demands are stochastic and follow some probability distribution, and the focus is on average results rather than worst-case ones.
There has been focus on Zipf distributions \cite{niesen2013,Zcodedcaching}, which can be seen in examples such as YouTube videos (see \figurename~\ref{fig:youtube} based on data from \cite{YoutubeRepository}), but also on arbitrary distributions \cite{niesen2013,ZhangArbitrary}.

\begin{figure}
\centering
\includegraphics[width=\myplotwidth]{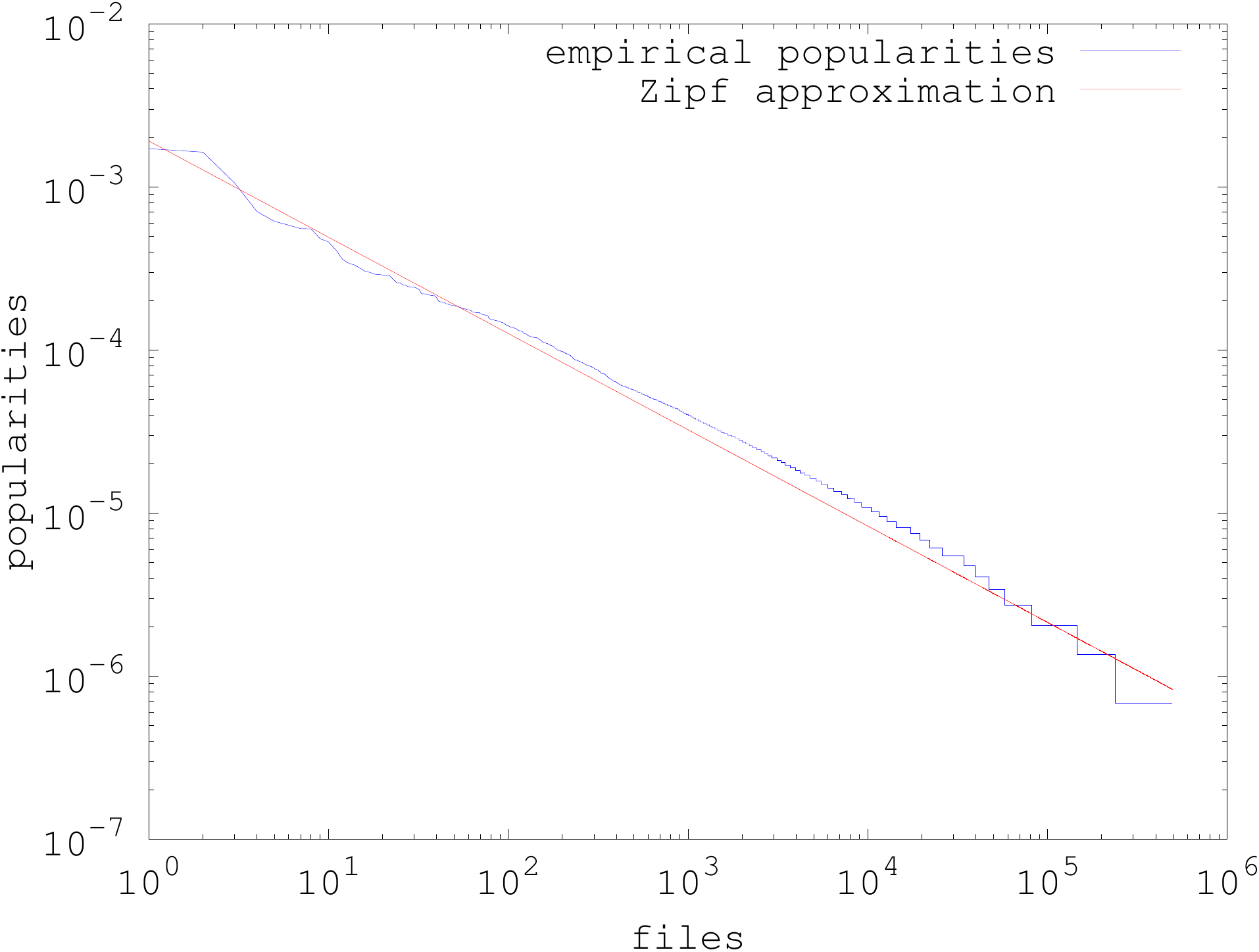}
\caption{Empirical popularities of some YouTube videos (based on number of views), with an approximating Zipf distribution.}
\label{fig:youtube}
\end{figure}

The popularity model that we consider in this paper is the \emph{multi-level} model.
The files are divided into $L$ \emph{popularity levels}, such that all files in a single level are equally popular.%
\footnote{For a discussion of the multi-level popularity model, see Section~\ref{sec:discretizing}.}
The levels consist of $N_1,\ldots,N_L$ files.
Furthermore, the total number of users in the system who are requesting files from each level $i$ is \emph{fixed} and \emph{known} to the designer \emph{a priori}.

To motivate the determinism in this last point, consider the following example.
Suppose there are two popularity levels, and assume a stochastic-demands setup where each user is three times as likely to request a file from the first level as he is from the second.
If there are $40$ users in the network, then we would expect that about $30$ of them will request a file from the first level, and $10$ from the second.
By the law of large numbers, when a large number of users is present in the system, we expect a concentration of the number of users requesting files from each level around their means.
Because of this concentration, the stochastic-demands model will closely resemble the determinism in the multi-level model that we adopt.

In this paper, we will use the phrase \emph{user profile} to refer to the arrangement of users across the caches.
Specifically, it is the number of users requesting a file from level $i$ at cache $k$, for every pair $(i,k)$.
Moreover, we define the \emph{user demand vector} or \emph{user request vector} as the vector of specific files requested by each user.

\subsection{Multi-level access}

Another aspect of the caching problem that we study is the possibility of users to access more than one cache.
This introduces a third dimension to the trade-off.
The more caches users can access, the more information they can gather from them for the same memory, and hence the less information they need from the broadcast transmission.
However, an increased degree can be costly to users who must now establish connections with multiple APs.

We introduce this multiple-access dimension to the problem by augmenting the multi-level popularity model with a \emph{multi-level access} aspect.
Users are required to access a certain number of caches based on the popularity level of the file they request.
Specifically, all users demanding a file from level $i$ must access $d_i$ caches.
We call $d_i$ the \emph{access degree} of level $i$.
Because we would prefer that users connect to the nearest caches, these $d_i$ caches are consecutive, with a cyclic wrap-around for symmetry, as in \figurename~\ref{fig:setup-F-ma}.

\subsection{Number of users}

As we will see later, the caching system that we have described behaves differently depending on the number of users \emph{per cache}.
In particular, it requires different strategies when this number is large or small.

We focus our attention on the two extreme cases of number of users.
In one extreme, we assume that there is only one user at every cache.
Recall from Section~\ref{sec:setup-popularity} that the popularity model assumes the number of users requesting a file from each level is fixed and known.
However, since there is only one user per cache, we do not know beforehand which user will connect to which cache.
Thus the user profile is unknown.
We call this extreme the \emph{single-user setup}.
In the other extreme, the number of users per cache is large enough that the concentration of the levels requested by each user manifests itself, not only on the overall set of users, but also on the set of users connected to a single cache.
In other words, we assume that the number of users requesting files from each level is fixed \emph{per cache}.
Furthermore, we assume a symmetry across the caches, so that every level is represented equally at each cache.
We call this setup the \emph{multi-user setup}.

After introducing the two setups, we now give their formal definitions.

\subsubsection{Multi-user setup}

\begin{figure}
\centering
\includegraphics[width=\myfigswidth]{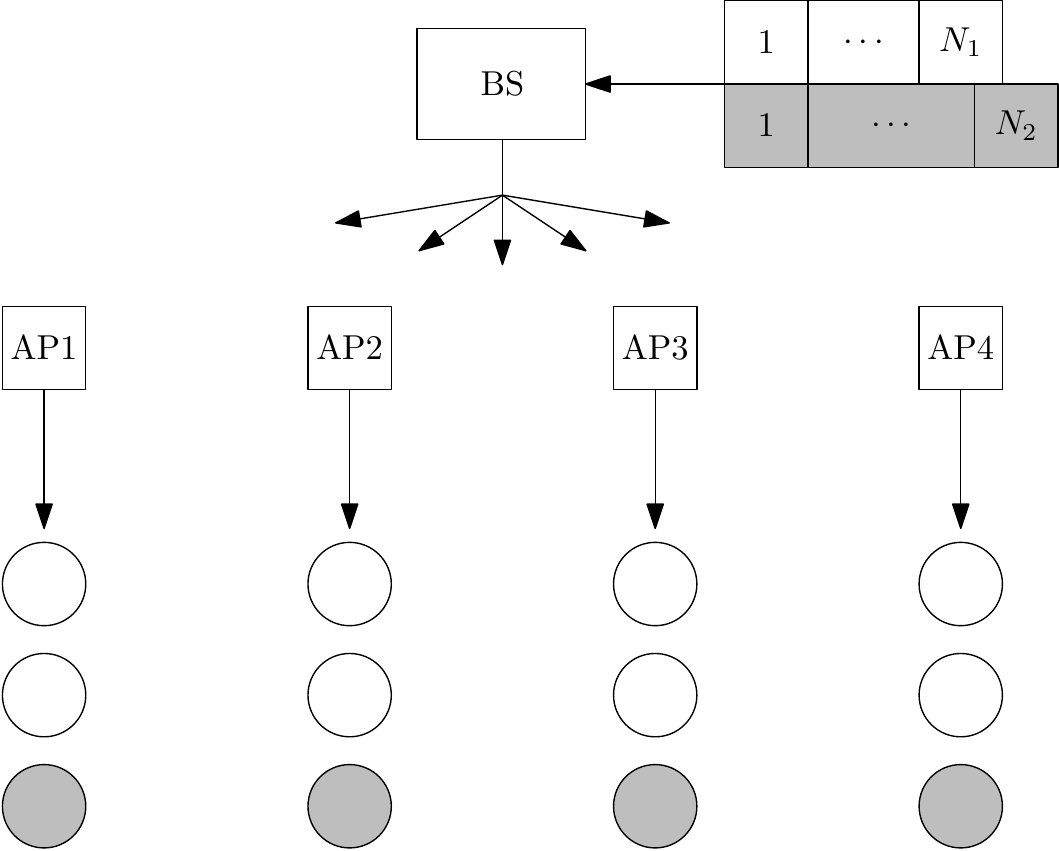}
\caption{Multi-user setup with $K=4$ caches, and $L=2$ levels with $(U_1,U_2)=(2,1)$ users per cache.
Both levels have an access degree of~$1$.}
\label{fig:setup-F}
\end{figure}

Consider the setup shown in \figurename~\ref{fig:setup-F-ma}.
For every level $i$ and every consecutive subset of $d_i$ caches, there are exactly $U_i$ users connecting to these caches and requesting a file from level $i$.
Notice that every level is represented equally at every cache.

\begin{figure}
\centering
\includegraphics[width=\myfigswidth]{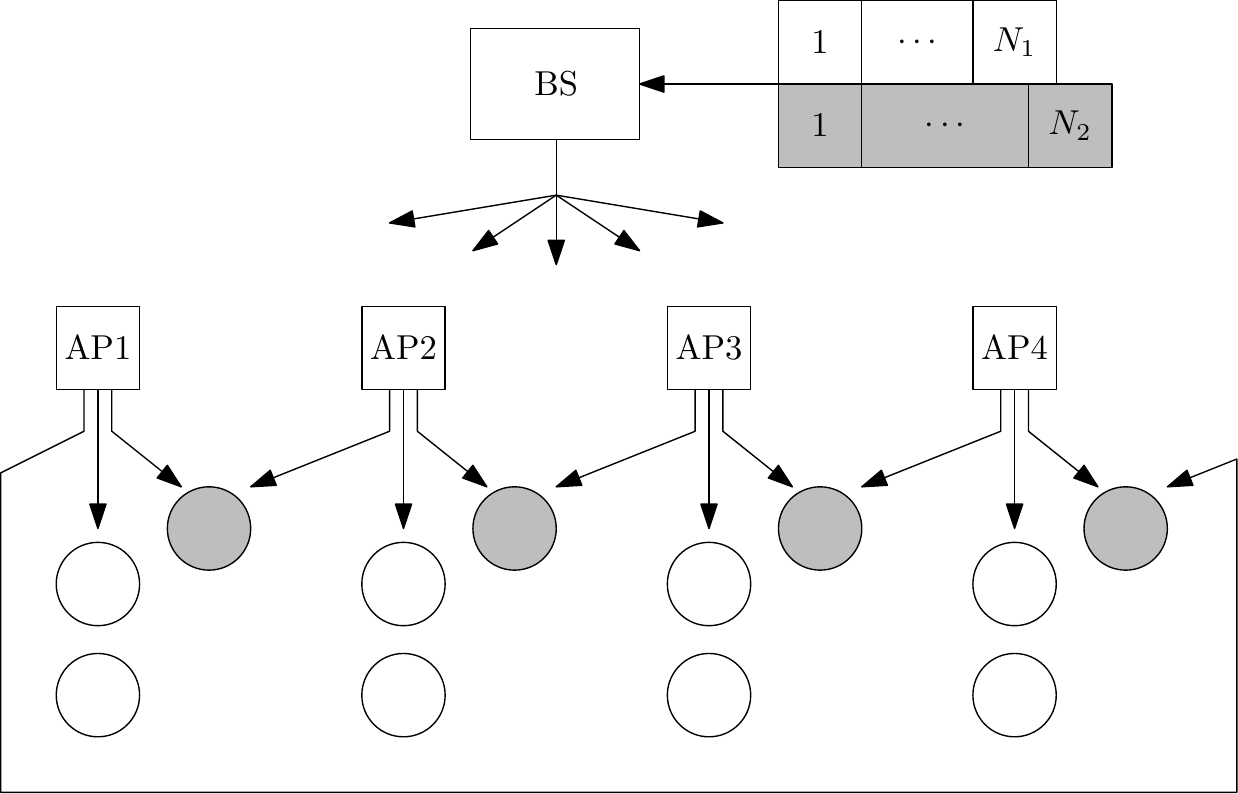}
\caption{Multi-user setup with $K=4$ caches, and $L=2$ levels with $(U_1,U_2)=(2,1)$ users per cache, and access degrees of $(d_1,d_2)=(1,2)$.}
\label{fig:setup-F-ma}
\end{figure}

\subsubsection{Single-user setup}

\begin{figure}
\centering
\includegraphics[width=\myfigswidth]{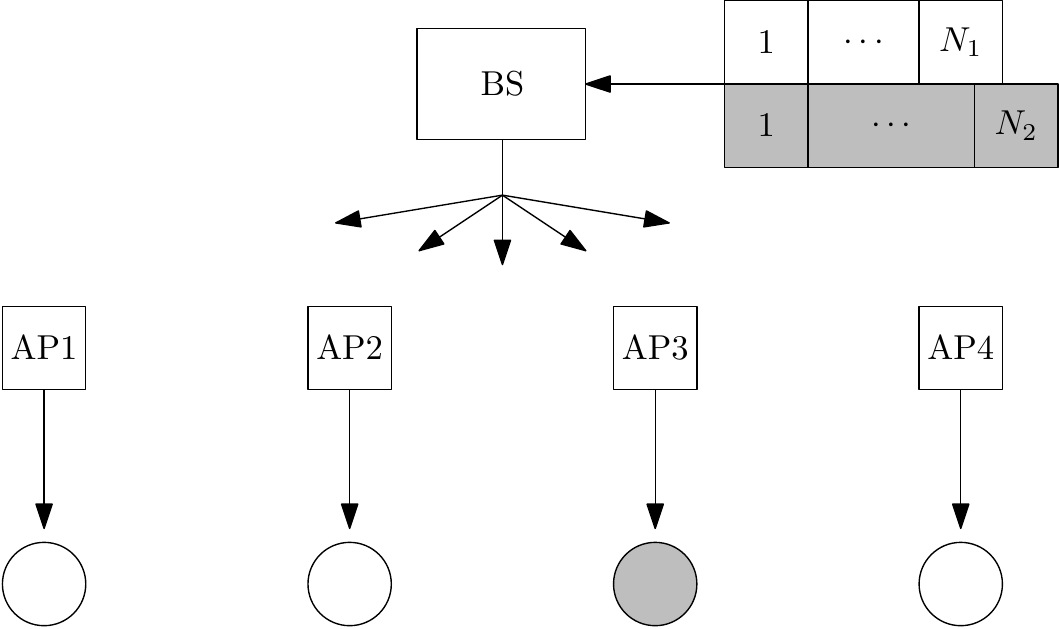}
\caption{Single-user setup with $K=4$ caches, and $L=2$ levels with $(K_1,K_2)=(3,1)$ users.}
\label{fig:setup-G}
\end{figure}

Consider now the setup in \figurename~\ref{fig:setup-G}, depicting the other extreme.
We have only one user connecting to every cache, for a total of $K$ users.
The only information known \emph{a priori} is that, for each level $i$, exactly $K_i$ out of the $K$ users will request a file from $i$.
However, we do not know which users these will be.

For symmetry, we assume only single-access in the single-user setup, i.e., $d_i=1$ for all levels $i$.

\subsection{Rate-memory trade-off}

Recall that our main goal is to characterize the trade-off between the broadcast rate $R$ and the cache memory $M$.
The determinism in the setups allows us to study worst-case trade-offs as opposed to average-case trade-offs.

We say that a pair $(R,M)$ is \emph{achievable} if there exists a placement-and-delivery strategy that uses caches of memory $M$ and transmits, for \emph{any} possible combination of user requests (valid within the popularity model), a broadcast message of rate at most $R$ that satisfies all said requests with vanishing probability as the file size $F$ grows.
Our goal is to find all such achievable pairs.
In particular, we wish to find the optimal rate-memory trade-off:
\[
R^\ast(M) = \inf\left\{ R : \text{$(R,M)$ is achievable} \right\},
\]
where the minimization is done over \emph{all} possible strategies.

\subsection{Problem formulation}
\label{sec:setup-formulation}

The problem of finding an exact characterization of the rate-memory trade-off is difficult even for the simplest cases \cite{maddah-ali2012}.
Therefore, in this paper, we will instead consider approximate characterizations.
In particular, we wish to find an achievability strategy that results in a rate-memory trade-off $R(M)$ such that:
\[
cR(M) \le R^\ast(M) \le R(M),
\]
where $c$ is some constant.
We say that such a strategy is \emph{order-optimal}.

We allow $c$ to depend on only one parameter: the maximum access degree $D=\max_i d_i$.
In practice, we do not expect that one user will be required to access a large number of caches, and so $D$ would be quite small.
However, $c$ must be independent of all other parameters.

\subsection{Regularity conditions}

We assume the following two regularity conditions.
First, for every popularity level $i$, there are more files than users.
In the multi-user setup, this means:
\begin{equation}
\label{eq:mu-reg1}
\forall i,\quad N_i\ge KU_i.
\end{equation}
In the single-user setup, we would write:
\begin{equation}
\label{eq:su-reg}
\forall i,\quad N_i \ge K_i.
\end{equation}
This can be seen, for example, in video applications such as Netflix, where ``files'' would be video segments of a few seconds to a few minutes.
To borrow an example from \cite{Hcodedcaching}, if a database has 1000 popular movies of length 100 minutes, and each movie is divided into files (segments) of one minute each, the result is 100,000 files.
It is unlikely that over 100,000 users will each be watching one of those 1000 movies at the same time.

Second, \emph{in the multi-user setup only}, we assume that no two levels have very similar popularities.
The popularity of a level can be written as the number of users per file of the level, i.e., as $U_i/N_i$.
Hence, if $i$ is a more popular level than $j$, the regularity condition states:%
\footnote{As we will see in later sections, the variables $U_i$ and $N_i$ will often appear inside square roots.
For this reason, phrasing the regularity condition using the square roots is more useful.}
\begin{equation}
\label{eq:mu-reg2}
\sqrt{\frac{U_i/N_i}{U_j/N_j}} \ge \frac{D}{\beta} = 198 D,
\end{equation}
where $\beta=\frac{1}{198}$ is called the \emph{level-separation factor}.
The reasoning behind this condition is that, if it did not hold for some levels $i$ and $j$, then we can think of them as essentially one level with $N_i+N_j$ files and $U_i+U_j$ users per cache.
The resulting popularity $\frac{U_i+U_j}{N_i+N_j}$ would be close to both $U_i/N_i$ and $U_j/N_j$.

\subsection{Notation table}

For reference, we present in \tablename~\ref{tbl:notation} all the notation that we use in this paper.

\begin{table}
\renewcommand{\arraystretch}{1.2}
\centering
\caption{Notation}
\label{tbl:notation}
\begin{tabular}{|c|l|}
\hline
\multicolumn{2}{|l|}{\textbf{All setups}}\\
\hline
$K$   & \# of caches\\
$L$   & \# of popularity levels\\
$N_i$ & \# of files in level $i$\\
$F$   & file size\\
$R$   & broadcast rate (normalized)\\
$M$   & cache memory (normalized)\\
$R^\ast(M)$ & optimal rate-memory trade-off\\
\hline
\multicolumn{2}{|l|}{\textbf{Multi-user setup}}\\
\hline
$U_i$ & \# of users per cache for level $i$\\
$d_i$ & access degree of level $i$\\
$D$   & maximum access degree\\
$\beta=\frac{1}{198}$ & level-separation factor\\
\hline
\multicolumn{2}{|l|}{\textbf{Single-user setup}}\\
\hline
$K_i$ & total \# of users for level $i$\\
\hline
\end{tabular}
\end{table}

\section{Preliminaries}
\label{sec:prelim}
\label{sec:preliminaries}
Traditional caching only uses multiple-unicast transmissions from the server to the users.
As a result, the total transmission size was proportional to the number of users, for any value of the cache memory.
Coded caching, initially introduced in \cite{maddah-ali2012}, brought a drastic improvement by eliminating the dependence of the transmission size on the number of users (except for very small cache memory).
This technique was shown to be approximately optimal in \cite{maddah-ali2012} (a centralized version) and in \cite{maddah-ali2013} (a decentralized version), under a setup with a single level of popularity and a single user at every cache, with a single-access structure.
We will refer to this setup as the Basic Setup, because it will form the basis of our main analysis.

The scheme for the Basic Setup, in its decentralized form (on which will we henceforth focus) consists in placing a random sampling of bits from all files in every cache, independently.
Consequently, there will be some overlap, but also some differences, in the bits present in every cache.
The BS then transmits linear combinations of these bits, taking advantage of the overlaps as well as the differences, so that the same linear combination can be useful for multiple users at once.
The resulting rate-memory trade-off is given in the following Lemma.

\begin{lemma}[Rate for the Basic Setup {\cite{maddah-ali2012}}]
\label{lemma:single-level}
For a single-level caching system with $K$ caches, $N$ files, a single user per cache with an access degree of $1$, and a cache memory of $M$, the following rate is achievable:
\[
R^0(M,K,N) = \min\left\{ \frac{N}{M} , K \right\} \cdot \left( 1-\frac{M}{N} \right).
\]
Furthermore, this rate is within a constant of the optimum.
\end{lemma}

Notice that, when $M>N/K$, then the rate becomes $(N/M-1)$, removing all dependence on the number of users.
Under traditional caching, the rate would have been $K(1-M/N)$.

\subsection{Generalizing to multi-user, multi-access}

In the multi-level setups, every level can have more than one user per cache, and can also have an access degree that is larger than $1$.
Thus, the first step towards solving the multi-level problem is to generalize the single-level problem to one where these two parameters are no longer $1$.
\figurename~\ref{fig:single-level} illustrates such a setup.

\begin{figure}
\centering
\includegraphics[width=\myfigswidth]{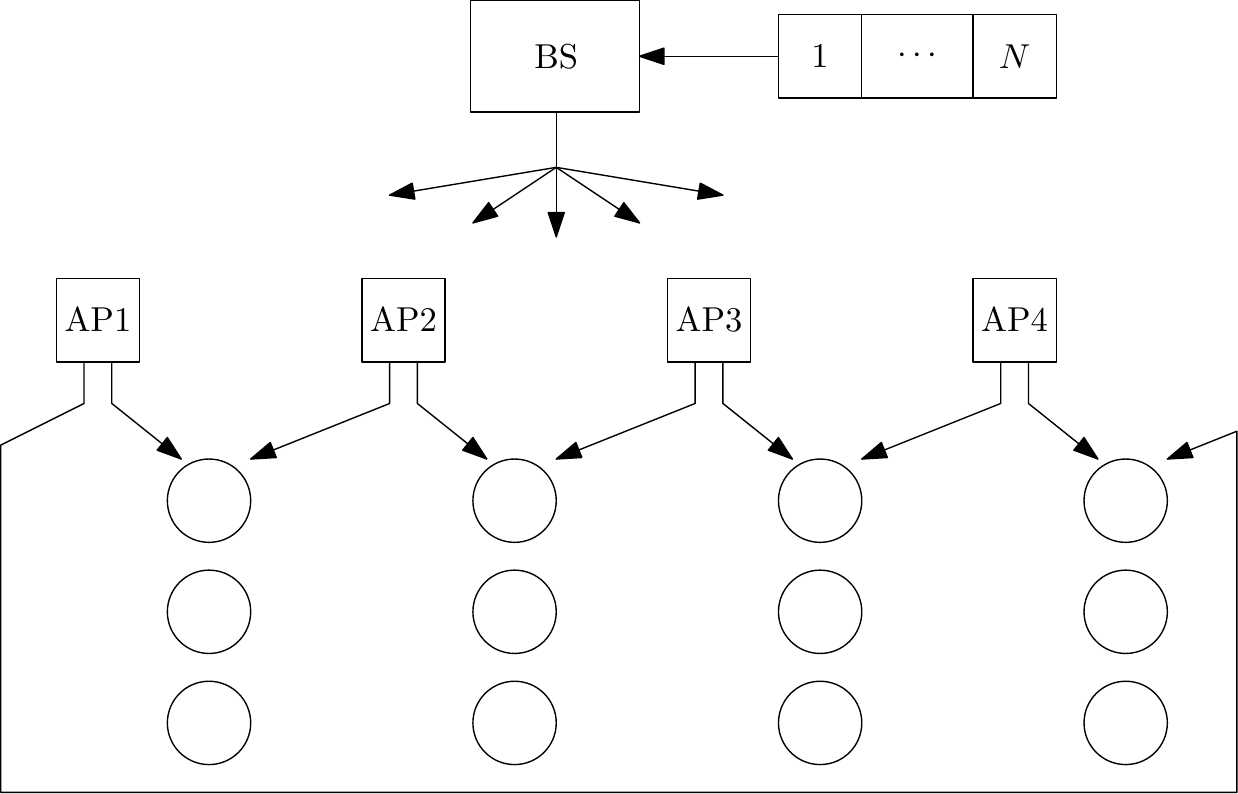}
\caption{Generalized single-level setup with $U=3$ users per cache, and an access degree $d=2$.}
\label{fig:single-level}
\end{figure}

We adopt the following strategy.
We first color all the caches into $d$ colors, such that every user is connected to exactly one cache of every color.
In parallel, we split every file into $d$ equal subfiles, and color each subfile using the same colors as the caches.
Next, we group all the $KU$ users into $dU$ groups of $K/d$ users each, such that no two users in the same group share any cache.
The end result is illustrated in \figurename~\ref{fig:single-level-groups}.

\begin{figure}
\centering
\includegraphics[width=\myfigswidth]{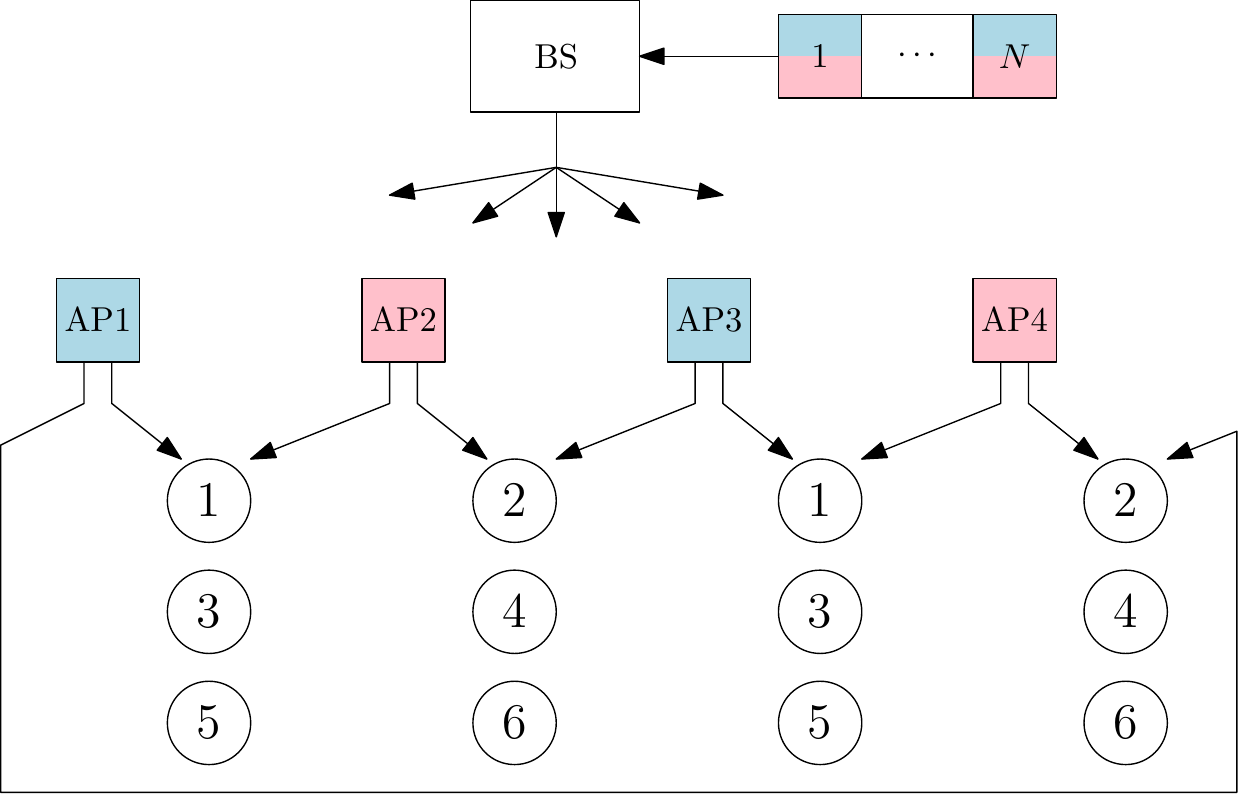}
\caption{An illustration of the scheme used on the example from \figurename~\ref{fig:single-level}.
Caches are colored into $d=2$ colors, and the files are divided and colored with the same colors.
In parallel, the users are divided into $dU=6$ groups, where users fro m the same group have no overlapping caches.}
\label{fig:single-level-groups}
\end{figure}

In the placement phase, we consider each of the $d$ colors separately.
The caches of a certain color perform a simple placement on the subfiles of the same color.
Since each file is split into $d$ equal subfiles, every cache can hold $dM$ subfiles, the equivalent of $M$ files.

In the delivery phase, each of the $dU$ groups is considered separately.
For every group and every color, we have a subsystem where $K/d$ users are each requesting a subfile of size $F/d$ bits from one like-colored cache of size $dM$ subfiles.
The total broadcast size is thus:
\[
RF = dU \cdot d \cdot R^0(dM,K/d,N)\cdot (F/d) \text{ bits.}
\]
Note that the memory given to the $R^0$ function is $dM$.
This is because the ``files'' in each subsystem have a size of $F/d$ bits, and thus the total cache size of $MF$ bits must be normalized by $F/d$ bits.

Using the above equation with Lemma~\ref{lemma:single-level}, we get the following theorem.

\begin{theorem}[Single-level rate-memory trade-off]
\label{thm:single-level}
Given a single-level caching system, with $N$ files, $K$ caches, $U$ users at each cache with access degree $d$, and a cache memory of $M$, the following rate is achievable:
\begin{equation}
\label{eq:single-level}
R^\text{SL}(M,K,N,U,d)
= U \cdot \min\left\{ \frac{N}{M} , K \right\}
\cdot \left( 1 - \frac{dM}{N} \right).
\end{equation}
\end{theorem}

While, in the single-access case, a cache memory $M<N$ meant a non-zero transmission rate, adding the multi-access aspect results in a zero achievable rate for the smaller memory value of $N/d$.
Intuitively, at $M=N/d$, it is as though we apply an erasure-correcting code on all the files and spread it across the caches, such that any $d$ caches can reproduce all files.

The results presented in this section will be key to our solution of the multi-level caching problem.
Indeed, in all that will follow, we use the above-described (decentralized) coded caching scheme as a black box that gives a rate $R^\text{SL}(M,K,N,U,d)$ for input parameters $M$, $K$, $N$, $U$, and $d$.

\subsection{A small multi-level example with exact characterization}

In order to illustrate the general multi-level problem, we will here present a small example that combines both multi-level popularity and multi-level access.
We give, for this example, an exact characterization of the rate-memory trade-off.

\begin{figure}
\centering
\includegraphics[width=.35\textwidth]{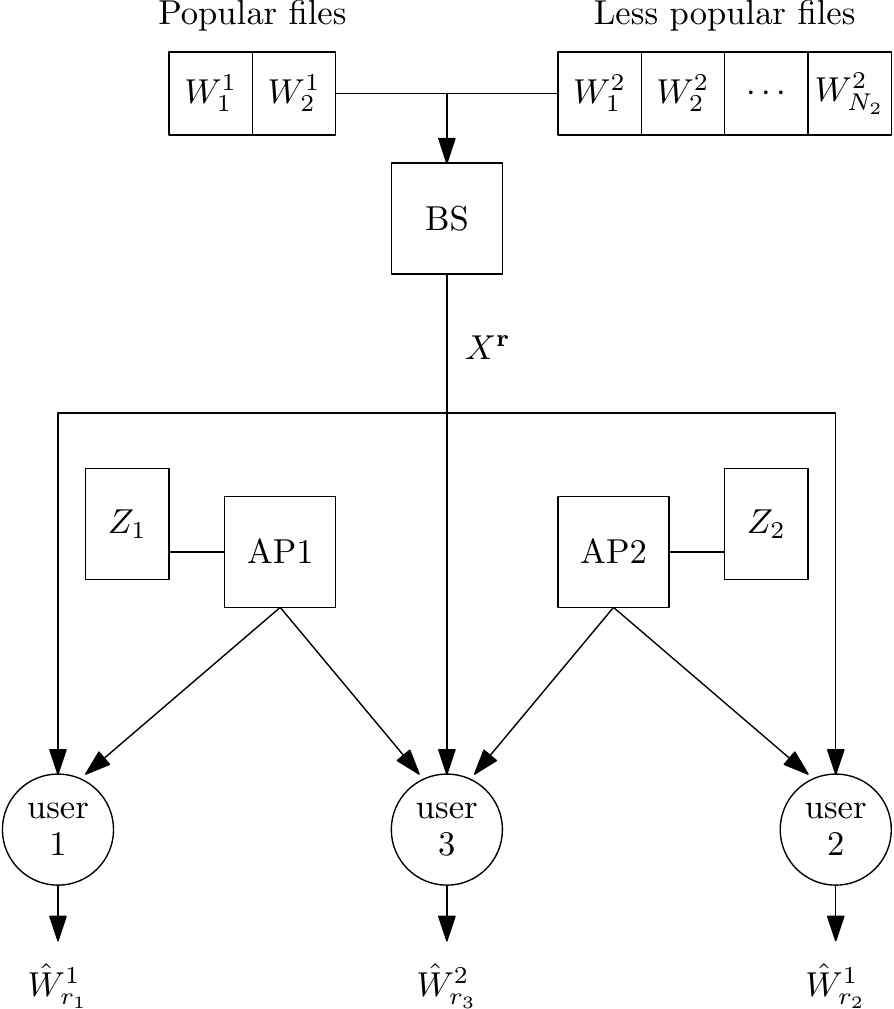}
\caption{Small example that illustrates multi-level popularity and access.}
\label{fig:small-ex-setup}
\end{figure}

Consider the setup in \figurename~\ref{fig:small-ex-setup}.
The server holds files from $L=2$ popularity levels.
The first level has $N_1=2$ files, and the second level has $N_2\ge4$ files.
There are two APs, each equipped with a cache of memory (i.e., size normalized by file size) $M$.
There is one user accessing each cache and requesting a file from level $1$ (users $1$ and $2$), and a third user accessing both caches and requesting a file from level $2$ (user $3$).

\begin{theorem}[Exact characterization for the small example]
\label{thm:small-example}
For the setup shown in \figurename~\ref{fig:small-ex-setup}, the optimal rate-memory trade-off is plotted in \figurename~\ref{fig:small-ex-region} and is characterized by:
\[
R^\ast(M)
= \max\left\{
3 - 2M,
\frac52 - M,
2 - \frac12 M,
1 - \frac{M-2}{N_2/2}
\right\}.
\]
\end{theorem}

\begin{figure}
\centering
\includegraphics[width=\myfigswidth]{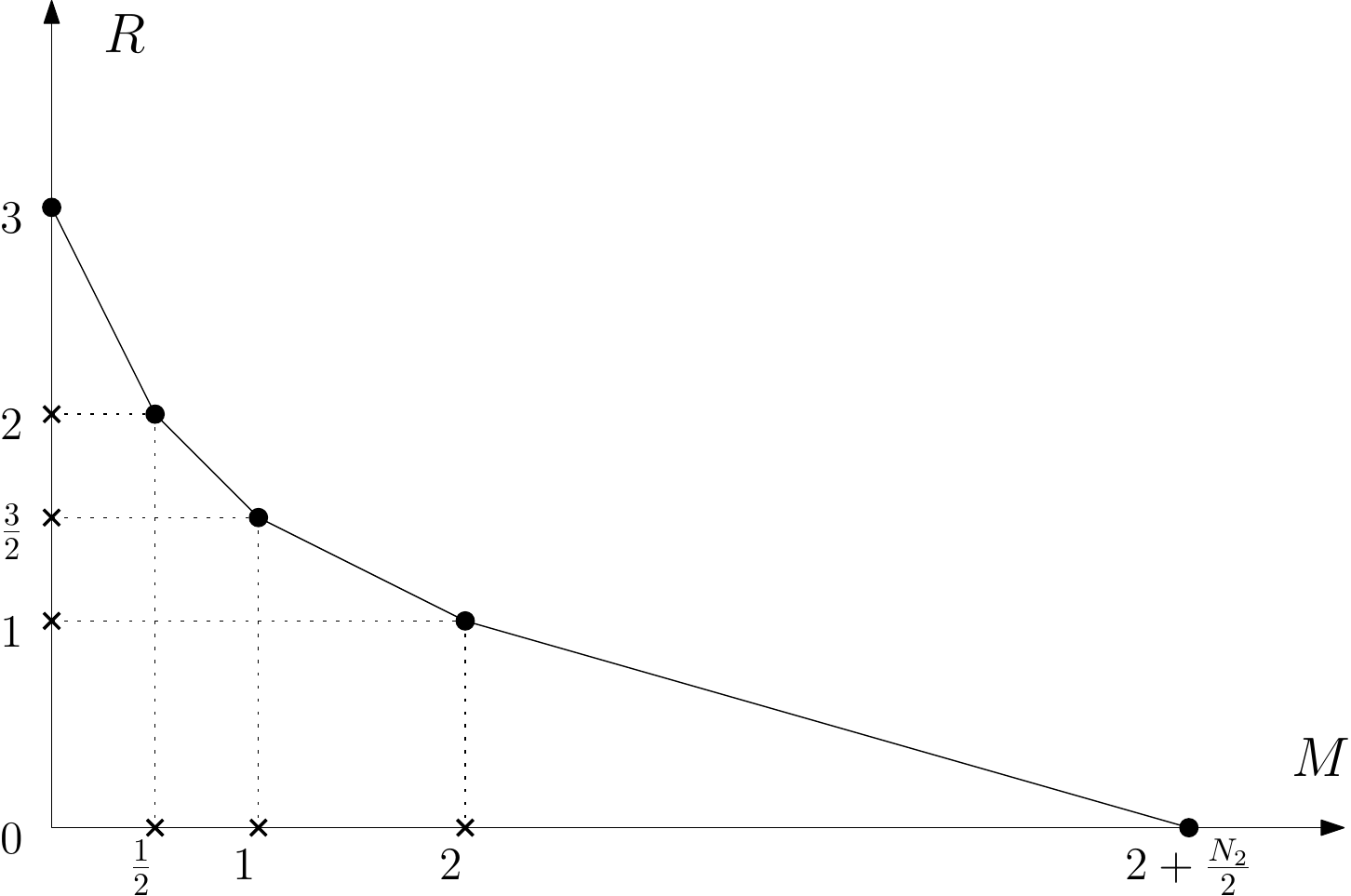}
\caption{Optimal rate-memory trade-off $R^\ast(M)$ for the small example.}
\label{fig:small-ex-region}
\end{figure}

The proof of Theorem~\ref{thm:small-example} is given in Appendix~\ref{app:small-example}.
However, for illustration and to gain intuition about the general problem, we will here briefly discuss the achievability scheme for two values of the cache memory $M$: $M=1$ (for which $R=3/2$ is achievable) and $M=1/2$ (for which $R=2$ is achievable).

Suppose that $M=1$, so that each cache can hold the equivalent of one file.
We first split each file in level $1$ into two equal parts: $W^1_n = ( W^1_{n,a}, W^1_{n,b} )$, for $n=1,2$.
Now, each cache exclusively stores one half of each popular file, and stores nothing from level $2$.
Thus the first cache will contain $(W^1_{1,a},W^1_{2,a})$ and the second cache will contain $(W^1_{1,b},W^1_{2,b})$.
When the users make their requests, the BS transmits $W^2_{r_3}$ completely for user 3, and sends a coded transmission $(W^1_{r_1,b}\oplus W^1_{r_2,a})$ for users 1 and 2 together.
Combining the transmission with the contents of their respective caches, each of users 1 and 2 can recover the file that they have requested.
In total, the BS would have transmitted one complete file ($W^2_{r_3}$), plus the equivalent of one half-file (the linear combination), for a total rate of $R=3/2$.

Suppose now that $M=1/2$, i.e., each cache can only hold the equivalent of half a file.
We again split the two level-$1$ files just like in the previous case.
However, this time the first cache stores $(W^1_{1,a}\oplus W^1_{2,a})$ while the second cache stores $(W^1_{1,b}\oplus W^1_{2,b})$.
When the users make their requests, the BS again transmits $W^2_{r_3}$ to serve user 3, but also sends $W^1_{r_1,b}$ and $W^1_{r_2,a}$ for users 1 and 2.
This allows them to recover the file that they have requested by combining the transmission with the side-information available at their caches.
Since the BS has transmitted one complete file and two half-files, the total rate is $R=2$.

While, in this small example, an exact characterization of the rate-memory trade-off was found, this is difficult in general.
For this reason, we focus our attention on order-optimality results as stated in Section~\ref{sec:setup-formulation}.

\section{Main Results}
\label{sec:results}
In this section, we provide the approximately optimal rate-memory trade-off for each of the two setups (multi-user and single-user).
We discuss how each such trade-off is achieved.

\subsection{Multi-user setup}
\label{sec:results-multiuser}

The placement-and-delivery strategy that we adopt for the multi-user setup is a \emph{memory-sharing} strategy.
It consists of dividing the cache memory between all the $L$ levels, and then treating each level as a separate caching sub-system, with the reduced memory.
In other words, we give level $i$ a memory $\alpha_iM$, where $\alpha_i\in[0,1]$ and $\sum_i\alpha_i=1$, and we then apply a single-level placement-and-delivery strategy for this level on this $\alpha_iM$ memory, separately from the other levels.
The total rate for this scheme is thus:
\begin{equation}
R^\text{MU}\left(M,K,\{N_i,U_i,d_i\}_i\right)
= \sum_{i=1}^L R^\text{SL}(\alpha_iM,K,N_i,U_i,d_i),
\end{equation}
where $R^\text{SL}(\cdot)$ is defined in \eqref{eq:single-level}.

By optimizing the overall rate over the memory-sharing parameters $\{\alpha_i\}_i$, we establish a memory allocation which we will show is order-optimal.
At a high level, this allocation is done by partitioning the popularity levels into three sets: $H$, $I$, and $J$.
The levels in $H$ have such a small popularity that they will get no cache memory.
On the opposite end of the spectrum, the most popular levels are assigned to $J$ and are given enough cache memory to completely store all their files in every cache.
Finally, the rest of the levels, in the set $I$, will share the remaining memory among themselves, obtaining some non-zero amount of memory but not enough to store all of their files.

Our choice of the $(H,I,J)$ partition and corresponding memory assignments are discussed in Section~\ref{sec:multi-user-achievability}.
This choice results in the following achievable rate.

\begin{theorem}
\label{thm:multi-user-achievability}
Given a multi-user caching setup, with $K$ caches, $L$ levels, and, for each level $i$, $N_i$ files and $U_i$ users per cache with access degree $d_i$, and a cache memory of $M$, the following rate%
\footnote{This expression of the rate is a slight approximation that we use here for simplicity as it is more intuitive.
An exact and complete description of the achievable rate is given in Section~\ref{sec:multi-user-achievability}.}
is achievable:
\[
R^\text{MU}(M) \approx \sum_{h\in H} KU_h
+ \frac{ \left( \sum_{i\in I} \sqrt{N_iU_i} \right)^2 }{ M - \sum_{j\in J}N_j/d_j }
- \sum_{i\in I}d_iU_i,
\]
where $(H,I,J)$ is a particular type of partition of the set of levels called an $M$-feasible partition.%
\footnote{This type of partition is defined in Definition~\ref{def:m-feasible} and elaborated upon in Section~\ref{sec:multi-user-achievability}.}
\end{theorem}

Intuitively, since a level $h\in H$ receives no cache memory, all requests from its $KU_h$ users must be handled directly from the broadcast.
Since, by regularity condition \eqref{eq:mu-reg1}, we have $N_i\ge KU_i$ for all levels $i$, then in the worst case a total of $KU_h$ distinct files must be completely transmitted for the users requesting files from level $h$.
The users in set $J$ require no transmission as the files are completely stored in all the caches; however, it does affect the rate through the memory available for levels in $I$.
This is apparent in the expression $M-\sum_{j\in J}N_j/d_j$.
Finally, the levels in $I$, having received some memory, require a rate that is inversely proportional to the effective memory and that depends on the level-specific parameters $N_i$, $U_i$, and $d_i$.

The structure of the $(H,I,J)$ partition that we have chosen allows us to efficiently compute it for all values of the cache memory $M$.
Indeed, we provide an algorithm in Section~\ref{sec:multi-user-achievability} that can find this partition, as well as the corresponding memory-sharing parameters $\alpha_i$, for every $M$ in $\Theta(L^2)$ running time.
Briefly, as $M$ is increased, levels get ``promoted'' from the set $H$ to $I$ to $J$.
The sequence of these promotions is directly determined by the popularity of the levels.

We now discuss the order-optimality of the memory-sharing scheme in the multi-user setup.
We develop new, non-cut-set lower bounds on the optimal rate, which use sliding-window entropy inequalities \cite{Liu14}, and show that the scheme achieves a rate that is within a constant factor of the optimal.
Note that this constant is independent of all the problem parameters except the largest AP access degree $D$.

\begin{theorem}\label{thm:multi-user-gap}
For all valid values of the problem parameters $K$, $L$,
$\{N_i,U_i,d_i\}_i$, and $M$, we have:
\[
\frac{R^\text{MU}(M)}{R^\ast(M)} \le c D,
\]
where $R^\text{MU}(M)$ is the rate achieved by memory-sharing, $R^\ast(M)$ is the optimal rate over all strategies, $D$ is the largest AP access degree $D=\max_i d_i$, and $c=9909$ is a constant (independent of all problem parameters).
\end{theorem}

The gap between the rate achieved by the memory-sharing strategy and the optimal rate is linear in $D$.
As we have argued earlier, we would not expect a situation where one user connects to too large a number of APs, and so $D$ can be thought of as a constant.

The lower bounds needed to prove Theorem~\ref{thm:multi-user-gap} have to include the effect of all the popularity levels on the transmission rate.
However, these effects can be very different, especially when some files are much more popular than others, something that cut-set bounds alone cannot account for.
Using sliding-window subset entropy inequalities \cite{Liu14}, we can combine multiple cut-set bounds that correspond to the different levels, without making any assumptions on the achievability scheme.
The resulting bounds bring out the necessity for memory-sharing.

Furthermore, the constant $c$ in Theorem~\ref{thm:multi-user-gap} is rather loose so as to simplify the analysis.
Numerics show that, in practice, this constant is much smaller.
For example, if we have $K=20$ caches, $L=3$ popularity levels consisting of $(N_1,N_2,N_3)=(200,20\,000,800\,000)$ files, $(U_1,U_2,U_3)=(10,5,1)$ users per cache, and access degrees of $(d_1,d_2,d_3)=(1,1,1)$, then the gap is less than $6.8$.

\subsection{Single-user setup}
\label{sec:single-user-results}

In the single-user setup, the scheme that we propose is quite different.
Instead of separating the levels, we \emph{cluster} a subset of them into a super-level that will be treated as essentially one level.
Specifically, we partition the levels into two subsets:
$H'$ and $I'$.
The set $I'$ will be clustered into one super-level, and all of the memory $M$ will be given to it, while $H'$ will be given no memory.

To understand how to choose $H'$ and $I'$, consider the following rough analysis.
Suppose that all levels except one (let's call it $j$) have been split into $H'$ and $I'$.
Then, the rate, using Theorem~\ref{thm:single-level}, would be:
\begin{IEEEeqnarray*}{rCl}
R
&=& \textstyle R^\text{SL}(0,\sum_{h\in H'}K_h,\sum_{h\in H'}N_h,1,1)\\
&& \textstyle {}+ R^\text{SL}(M,\sum_{i\in I'}K_i,\sum_{i\in I'}N_i,1,1)\\
&\approx& \sum_{h\in H'} K_h + \frac{\sum_{i\in I'}N_i}{M}.
\end{IEEEeqnarray*}
If we were to add level $j$ to $H'$, that would result in the addition of a $K_j$ term, since all $K_j$ requests would be completely served by the broadcast.
On the other hand, if it is added to $I'$, then we would get an additional $N_j/M$ term, since the total number of files in $I'$ would increase by $N_j$.
Clearly, it is beneficial to choose the smaller of the two quantities.

Though the above analysis is rough, its main idea still holds.
In general, we choose the partition $(H',I')$ as follows:
\begin{equation}
\label{eq:G-partition}
H' = \left\{ h\in\{1,\ldots,L\} : M < \frac{N_h}{K_h} \right\};
\quad I' = (H')^c.
\end{equation}
Then, by giving all of the memory to $I'$, we can apply a single-level caching-and-delivery scheme to obtain the rate in the following theorem.

\begin{theorem}
\label{thm:single-user-achievability}
Consider the multi-level, single-user setup with $L$ levels, $N_i$ files and $K_i$ users for each level $i$, and cache memory $M$.
Then, the following rate is achievable:
\[
R^\text{SU}(M)
=
\sum_{h\in H'} K_h +
\max\left\{ \frac{\sum_{i\in I'}N_i}{M} - 1 \,,\, 0 \right\},
\]
where $H'$ and $I'$ are as defined in \eqref{eq:G-partition}.
\end{theorem}

This scheme turns out to be order-optimal, as we state in the next theorem.

\begin{theorem}
\label{thm:single-user-gap}
The rate achieved by the clustering strategy in Theorem~\ref{thm:single-user-achievability} is within a constant multiplicative factor of the information-theoretic optimum:
\[
\frac{R^\text{SU}(M)}{R^\ast(M)} \le 72,
\]
where $R^\text{SU}(M)$ is the rate achieved by clustering, and $R^\ast(M)$ is the information-theoretically optimal rate.
\end{theorem}

Unlike in the multi-user case, cut-set bounds are sufficient to show order-optimality in this case.
Indeed, a single cut-set bound allows us to capture the fact that the user profile (i.e., the level of the file requested at each cache, as defined in Section~\ref{sec:setup-popularity}) is not determined beforehand.
At the same time, it brings out the necessity of clustering levels by mixing their demands.
As before, however, these bounds do not make any assumptions on the achievability strategy.

The scheme suggested in Theorem~\ref{thm:single-user-achievability} is similar to the results in \cite{Zcodedcaching} for Zipf popularity distributions and in \cite{ZhangArbitrary} for arbitrary distributions.
However, this is done for the multi-level popularity model, and Theorem~\ref{thm:single-user-gap} establishes a universal approximation for worst case rate, rather than average rate.

\subsection{Comparison with LFU}

One of the most common caching schemes conventionally in use is the Least-Frequently-Used (LFU) scheme.
This strategy stores the most popular files only, as many as the caches can hold.
In traditional LFU, no coding is done in the delivery phase.
We give an example for each of the two setups, to show how the respectively chosen schemes are superior to LFU in each context.

\subsubsection{Multi-user example}

Consider two levels ($L=2$), with $K=30$, $(N_1,N_2)=(600,1000)$, $(U_1,U_2)=(20,10)$, and $(d_1,d_2)=(1,1)$.
The rates achieved by the memory-sharing strategy and traditional LFU are plotted against memory $M$ in \figurename~\ref{fig:mu-ml-lfu}.
Memory-sharing performs up to $29$ times better than LFU in this example.

\begin{figure}
\centering
\includegraphics[width=\myplotwidth]{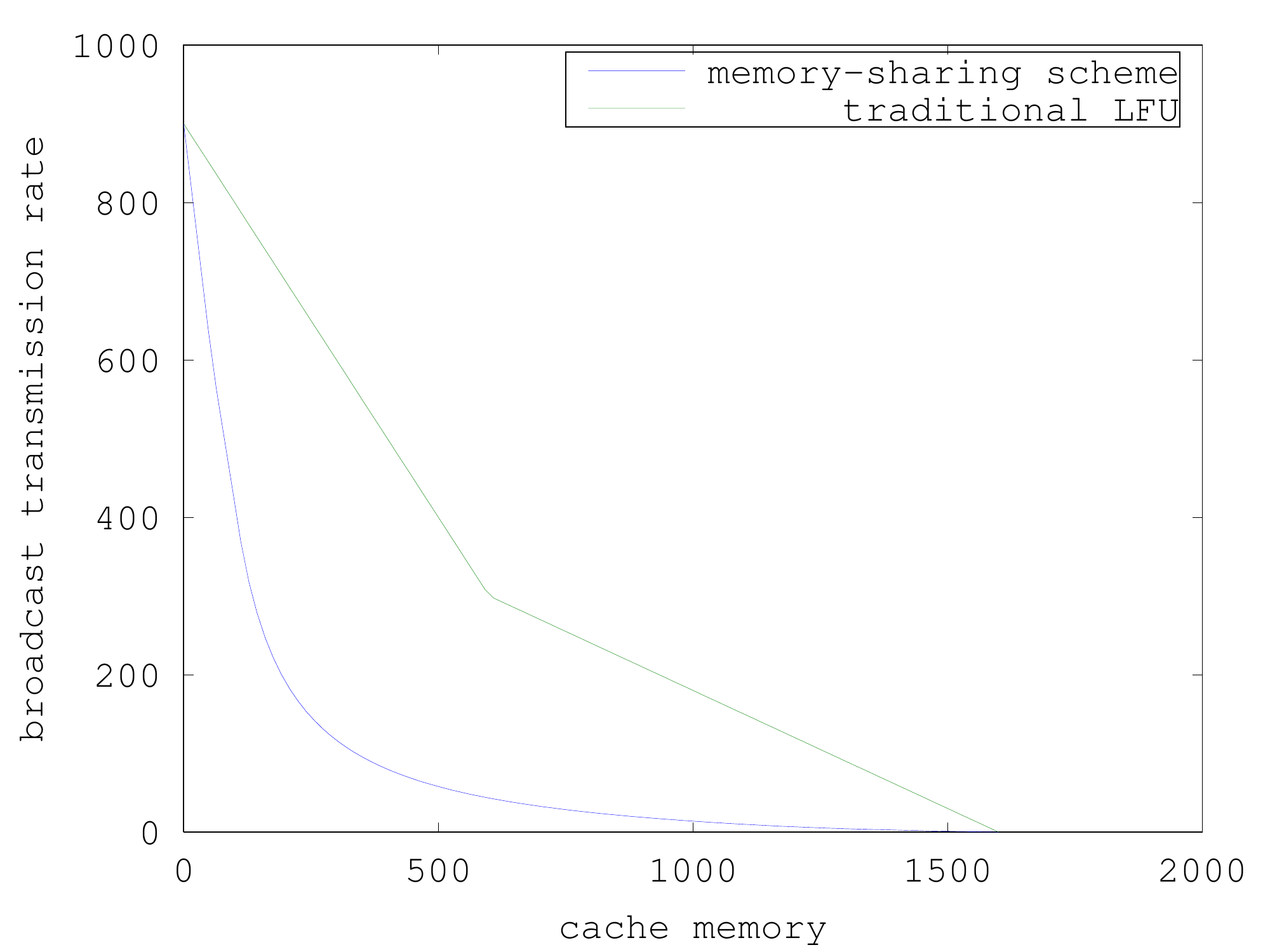}
\caption{Comparison of the memory-sharing scheme with traditional LFU, in the multi-user setup.}
\label{fig:mu-ml-lfu}
\end{figure}

\subsubsection{Single-user example}

Consider two levels ($L=2$), with $(N_1,N_2)=(500,1000)$ and $(K_1,K_2)=(30,15)$.
The rates achieved by the clustering strategy and traditional LFU are plotted against memory $M$ in \figurename~\ref{fig:su-ml-lfu}.
Clustering performs up to $22$ times better than LFU in this example.

\begin{figure}
\centering
\includegraphics[width=\myplotwidth]{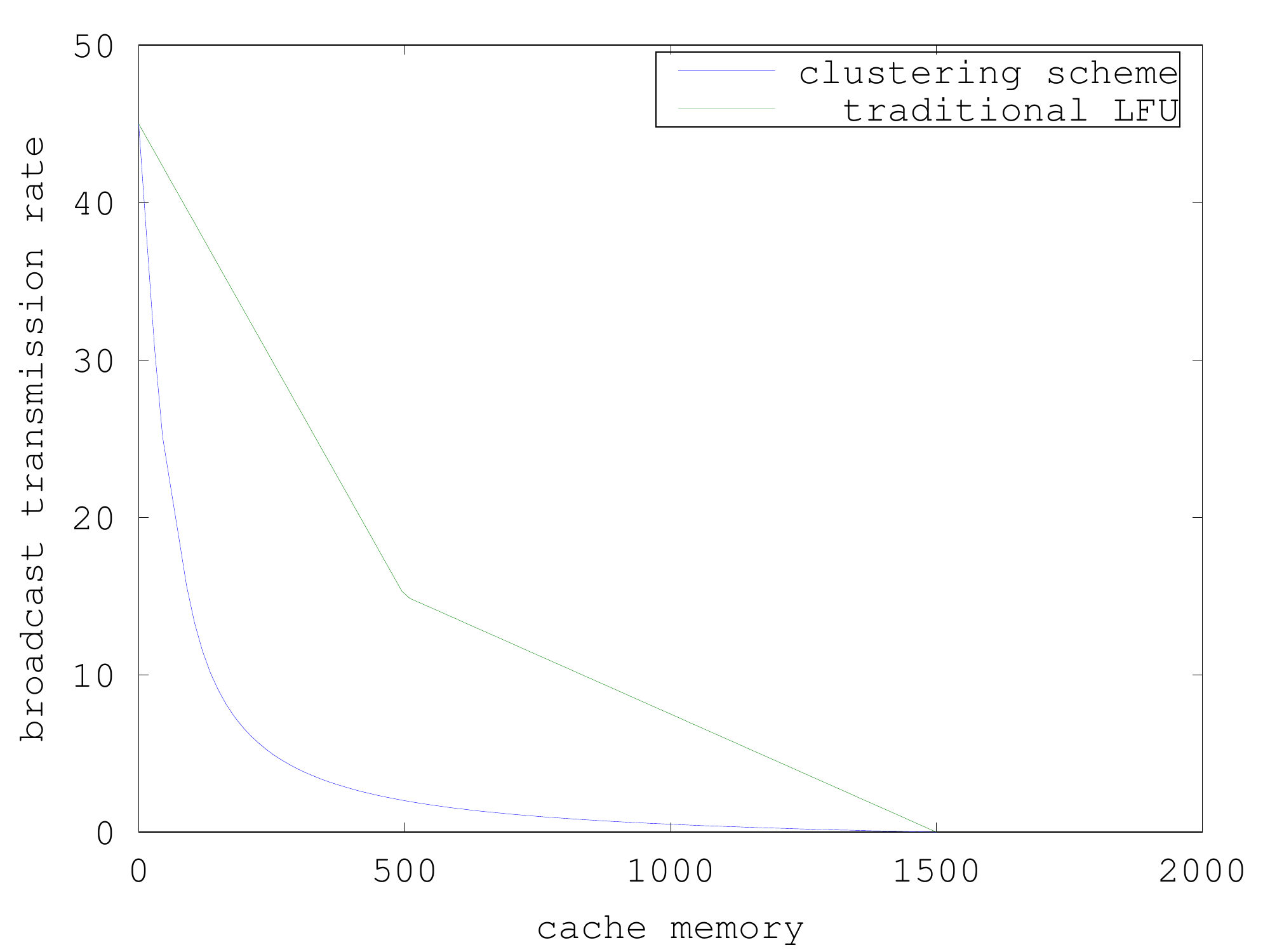}
\caption{Comparison of the clustering scheme with traditional LFU, in the single-user setup.}
\label{fig:su-ml-lfu}
\end{figure}

\section{The Multi-User Setup}
\label{sec:multi-user}

\subsection{Caching-and-delivery strategy: memory-sharing}
\label{sec:multi-user-achievability}

In this section, we describe the strategy used to achieve the rate approximated by the expression in Theorem~\ref{thm:multi-user-achievability}.
Moreover, we will give an exact upper bound on the achieved rate.

We will proceed in two steps.
First, we will discuss the $(H,I,J)$ partition of the set of levels as first described in Section~\ref{sec:results-multiuser}. 
This will be accompanied by an explanation of the memory-sharing parameters $\alpha_i$, which indicate what fraction of memory each level gets.
Second, we analyze the \emph{individual rate} achieved by every level $i$, after allocating $\alpha_iM$ memory to it, and combine all levels to produce the total rate achieved by the scheme.

While we actually define the $(H,I,J)$ partition based on the problem parameters, \emph{and then} choose the $\alpha_i$ values accordingly, in this paragraph we will proceed in the opposite order so that we expose the intuition behind the choices.
As discussed in Section~\ref{sec:results-multiuser}, the strategy involves finding a good partition $(H,I,J)$ of the set of levels, such that levels in $H$ are given no memory, levels in $J$ are given maximal memory, and levels in $I$ share the rest.
Thus, for all levels $h\in H$, we will assign $\alpha_hM=0$.
Similarly, every level $j\in J$ will receive $\alpha_jM=N_j/d_j$, since that is the amount of memory needed to completely store level $j$ and hence to require no BC transmission (i.e., $R^\text{SL}=0$; see Theorem~\ref{thm:single-level}).
What is left is to share the remaining memory $\left(M-\sum_{j\in J}N_j/d_j\right)$ among the levels in $I$.
More popular files should get more memory, and the popularity of a level $i$ is proportional to $U_i/N_i$.
Thus, we choose to give level $i$ a memory roughly $\alpha_iM\propto N_i\cdot\sqrt{U_i/N_i}$ (hence the memory \emph{per file} is proportional to $\sqrt{U_i/N_i}$).%
\footnote{The square root comes from minimizing an inverse function of $\alpha_i$.}

Intuitively, we want to choose $H$, $I$, and $J$ such that the above values of $\alpha_i$ are valid, i.e., $\alpha_i\in[0,1]$ for all $i$.
Based on the intuition, we get the partition described below.
\begin{definition}[$M$-feasible partition]
\label{def:m-feasible}
For any cache memory $M$, an $M$-feasible partition $(H,I,J)$ of the set of levels is a partition that satisfies:
\begin{IEEEeqnarray*}{lCl/rCcCl}
\forall h &\in& H, &
&& \tilde M &<& \frac1K \sqrt{\frac{N_h}{U_h}};\\
\forall i &\in& I, &
\frac1K \sqrt{\frac{N_i}{U_i}} &\le& \tilde M &\le& \left( \frac{1}{d_i}+\frac1K \right)\sqrt{\frac{N_i}{U_i}};\\
\forall j &\in& J, &
\left( \frac{1}{d_j} + \frac1K \right)\sqrt{\frac{N_j}{U_j}} &<& \tilde M,
\end{IEEEeqnarray*}
where $\tilde M = (M-T_J+V_I)/S_I$, and, for any subset $A$ of the levels:
\[
S_A = \sum_{i\in A} \sqrt{N_iU_i};
\quad
T_A = \sum_{i\in A} \frac{N_i}{d_i};
\quad
V_A = \sum_{i\in A} \frac{N_i}{K}.
\]
\end{definition}
We stress again that, while we used our own $\alpha_i$ values to determine $(H,I,J)$, this was done only for the intuition behind the choice.
The partition itself is defined solely based on the problem parameters, and not on our strategy.

Notice in Definition~\ref{def:m-feasible} how the different sets are largely determined by the quantity $\sqrt{N_i/U_i}$, for each level $i$.
This matches the idea that the most popular levels (i.e., those with the \emph{smallest} $N_i/U_i$) will be in $J$, while the least popular levels (those with the \emph{largest} $N_i/U_i$) will go to the set $H$.

After choosing an $M$-feasible partition, we share the memory among the levels using the following (precise) values of $\alpha_i$:
\begin{IEEEeqnarray*}{lCl"rCl}
\forall h &\in& H, & \alpha_hM &=& 0;\\
\forall i &\in& I, & \alpha_iM &=& \sqrt{N_iU_i}\cdot\tilde M - N_i/K;\\
\forall j &\in& J, & \alpha_jM &=& N_j/d_j.\IEEEyesnumber\label{eq:alpha-i}
\end{IEEEeqnarray*}
For completeness, the following proposition states the validity of this choice of memory-sharing parameters.
\begin{proposition}
The values of the memory-sharing parameters defined in \eqref{eq:alpha-i} satisfy:
\begin{enumerate}
\item $\alpha_i\ge0$ for all $i$;
\item $\sum_i\alpha_i=1$;
\item $\alpha_iM\le N_i/d_i$ for all $i$.
\end{enumerate}
Note that points 1 and 2 imply $\alpha_i\in[0,1]$.
\end{proposition}
\begin{proof}
All three results follow directly from Definition~\ref{def:m-feasible}.
\end{proof}

The structure of the solution described above allows us to efficiently compute the $\alpha_i$ values.
Indeed, Algorithm~\ref{alg:m-feasible} finds $\alpha_i$ for all levels $i$ and \emph{for all memory values $M$} in $\Theta(L^2)$ running time, where $L$ is the total number of levels.
While a detailed description and analysis of the algorithm is given in Appendix~\ref{app:multi-user-extras}, we briefly go over it in this section.
It proceeds in three main steps.
In the first step, the algorithm identifies the sequence of $(H,I,J)$ partitions that will occur as $M$ increases, using only the problem parameters.
For example, if there are two levels, there are two possible sequences, denoted below as S1 and S2:\\

\begin{tabular}{|l|c|c|c|c|c|c|}
\hline
& & $M=0$ & & & & large $M$\\
\hline
  & $H$ & $\{1,2\}$ & $\{2\}$ & $\emptyset$ & $\emptyset$ & $\emptyset$\\
S1& $I$ & $\emptyset$ & $\{1\}$ & $\{1,2\}$ & $\{2\}$ & $\emptyset$\\
  & $J$ & $\emptyset$ & $\emptyset$ & $\emptyset$ & $\{1\}$ & $\{1,2\}$\\
\hline
  & $H$ & $\{1,2\}$ & $\{2\}$ & $\{2\}$ & $\emptyset$ & $\emptyset$\\
S2& $I$ & $\emptyset$ & $\{1\}$ & $\emptyset$ & $\{2\}$ & $\emptyset$\\
  & $J$ & $\emptyset$ & $\emptyset$ & $\{1\}$ & $\{1\}$ & $\{1,2\}$\\
\hline
\end{tabular}
~\\

Notice that the $M$-feasible partition is the same throughout an entire interval of values for $M$.
In fact, there are only $2L$ non-trivial intervals: each interval is distinguished from the previous one by the ``promotion'' of a level from $H$ to $I$ or from $I$ to $J$.
As a result, computing only $2L$ intervals (actually, $(2L+2)$ intervals, to include the boundary cases) is enough to determine the $M$-feasible partition for \emph{all} values of $M\ge0$.

While Step 1 determines which of the above two sequences (S1 and S2) will occur, it does not determine the boundaries of the different regimes, i.e., the memory values at which the $(H,I,J)$ partition changes (except for trivial boundaries like $M=0$).
The second step computes these boundaries, using not only the problem parameters, but also the partitions themselves.

For every $M$, the corresponding $M$-feasible partition is determined by the algorithm based on the boundaries calculated in Step 2.
For example, suppose the algorithm decided that sequence S1 should be in use, and that the boundaries of $(H,I,J)=(\emptyset,\{1,2\},\emptyset)$ are some values $m_1$ and $m_2$.
Then, any $M\in[m_1,m_2]$ has $(\emptyset,\{1,2\},\emptyset)$ as its $M$-feasible partition.

\begin{algorithm}
\caption{An algorithm that constructs an $M$-feasible partition for all $M$.}
\label{alg:m-feasible}
\begin{algorithmic}[1]
\Require Number of caches $K$ and parameters $\{N_i,U_i,d_i\}_i$ for $i=1,\ldots,L$.
\Ensure An $M$-feasible partition for all $M$.

\ForAll{$i\in\{1,\ldots,L\}$}
  \State $\tilde m_i \gets (1/K)\sqrt{N_i/U_i}$
  \State $\tilde M_i \gets (1/d_i+1/K)\sqrt{N_i/U_i}$
\EndFor
\State $(x_1,\ldots,x_{2L}) \gets \mathrm{sort}(\tilde m_1,\ldots,\tilde m_L,\tilde M_1,\ldots,\tilde M_L)$.
%\Comment Sort $(\tilde m_i,\tilde M_i)$ in ascending order, and label the resulting sequence as $(x_1,\ldots,x_{2L})$.
\State

\State \emph{Step 1: Determine $(H,I,J)$ for each interval $(x_t,x_{t+1})$.}
\State Set $H_0\gets\{1,\ldots,L\}$, $I_0\gets\emptyset$, $J_0\gets\emptyset$.
\For{$t\in\{1,\ldots,2L\}$}
  \If{$x_t=\tilde m_i$ for some $i$}
    \State \emph{Promote level $i$ from $H$ to $I$}
    \State $H_t \gets H_{t-1} \setminus \{i\}$
	\State $I_t \gets I_{t-1} \cup \{i\}$
	\State $J_t \gets J_{t-1}$
  \ElsIf{$x_t=\tilde M_i$ for some $i$}
    \State \emph{Promote level $i$ from $I$ to $J$}
    \State $H_t \gets H_{t-1}$
	\State $I_t \gets I_{t-1} \setminus \{i\}$
	\State $J_t \gets J_{t-1} \cup \{i\}$
  \EndIf
\EndFor
\State

\State \emph{Step 2: Compute the limits of the intervals as $[Y_t,Y_{t+1})$.}
\ForAll{$t\in\{1,\ldots,2L\}$}
  \State $Y_t \gets x_t \cdot S_{I_t} + T_{J_t} - V_{I_t}$
\EndFor
\State $Y_{2L+1}\gets\infty$ \Comment For convenience
\State

\State \emph{Step 3: Determine the $M$-feasible partition for all $M$.}
\ForAll{$t\in\{1,\ldots,2L\}$}
  \State Set $(H_t,I_t,J_t)$ as the $M$-feasible partition of \emph{all} $M\in[Y_t,Y_{t+1})$
\EndFor
\end{algorithmic}
\end{algorithm}

The following lemma presents two important properties of $M$-feasible partitions.
It will be proved in Appendix~\ref{app:multi-user-extras}.
\begin{lemma}
\label{lemma:m-feasible-existence}
For any cache memory $M$, an $M$-feasible partition $(H,I,J)$ always exists.
Furthermore, the set $I$ is never empty as long as $M\le\sum_iN_i/d_i$, i.e., as long as individual caches do not have enough memory to store everything.
\end{lemma}

To properly analyze the achievable rate, we need to look more closely at the set $I$.
In the single-level, single-access scenario in \cite{maddah-ali2012,maddah-ali2013}, three regimes were identified, and they were analyzed separately.
Generalizing to the single-level, multi-access case, these regimes are: when $M<N/K$, when $M>cN/d$ for some constant $c\in(0,1)$, and the intermediate case.
We identify three similar regimes for each level in $i$.
Formally, define:
\begin{IEEEeqnarray*}{rCl}
I_0 &=& \left\{ i\in I : \tilde M < \frac2K\sqrt{\frac{N_i}{U_i}} \right\};\\
I_1 &=& \left\{ i\in I : \tilde M > \left(\frac{\beta}{d_i}+\frac1K\right)\sqrt{\frac{N_i}{U_i}} \right\};\\
I'  &=& I\setminus(I_0\cup I_1),
\IEEEyesnumber\label{eq:refined-partition}
\end{IEEEeqnarray*}
By choosing the $\alpha_i$ values in \eqref{eq:alpha-i}, these definitions are equivalent to: $I_0$ is the set of levels $i$ such that $\alpha_iM<N_i/K$; $I_1$ is such that $\alpha_iM>\beta N_i/d_i$ for all $i\in I_1$, where $\beta$ is the level-separation factor as defined in Regularity Condition \eqref{eq:mu-reg2}; and $I'$ consists of the remaining levels.

When $K\ge D/\beta$, then $I_0$, $I'$, and $I_1$ are mutually exclusive and form a partition of $I$.
For convenience, we call the resulting partition $(H,I_0,I',I_1,J)$ a \emph{refined $M$-feasible partition}.%
\footnote{The fact that this is a partition is important for the case $K\ge D/\beta$.
In the opposite case, this does not necessarily hold but it will not matter anyway.}

What follows is an important statement regarding the set~$I_1$.
\begin{proposition}[Size of $I_1$]
\label{prop:i1-size}
In any refined $M$-feasible partition $(H,I_0,I',I_1,J)$ as defined in Definition~\ref{def:m-feasible} and \eqref{eq:refined-partition}, the set $I_1$ contains at most one element.
\end{proposition}
\begin{proof}
Suppose that there exist two level $i,j\in I_1$ (and possibly others).
We will show that this violates regularity condition \eqref{eq:mu-reg2}.

Suppose without loss of generality that $i$ is more popular than $j$.
Since $i,j\in I_1$, then, by Definition~\ref{def:m-feasible} and \eqref{eq:refined-partition}:
\[
\left( \frac{\beta}{d_j}+\frac1K \right)\sqrt{\frac{N_j}{U_j}}
< M
\le \left( \frac{1}{d_i} + \frac1K \right)\sqrt{\frac{N_i}{U_i}}.
\]
However, this means:
\[
\sqrt{\frac{U_i/N_i}{U_j/N_j}}
< \frac{1/d_i+1/K}{\beta/d_j+1/K}
< \frac{d_j}{\beta d_i}
\le \frac{D}{\beta},
\]
which contradicts regularity condition \eqref{eq:mu-reg2}.
\end{proof}

Using the definition of a refined $M$-feasible partition, and the values of $\{\alpha_i\}_i$, we give upper bounds on the rates achieved individually for each level.
\begin{lemma}
\label{lemma:multi-user-achievability}
Given a refined $M$-feasible partition $(H,I_0,I',I_1,J)$, the individual rates of the levels are upper-bounded by:
\begin{IEEEeqnarray*}{lCl"rCl}
\forall h &\in& H, & R_h(M) &=& KU_h;\\
\forall i &\in& I_0\cup I', & R_i(M) &\le& \frac{2S_I\sqrt{N_iU_i}}{M-T_J+V_I};\\
\forall i_1 &\in& I_1, &
R_{i_1}(M) &\le& \frac1\beta d_{i_1}U_{i_1}\left( 1 - \frac{M-T_J}{N_{i_1}/d_{i_1}} \right)\\
&& & && {}+ \frac1\beta d_{i_1}U_{i_1}\frac{S_{I_0}+S_{I'}}{\sqrt{N_{i_1}U_{i_1}}};\\
\forall j &\in& J, & R_j(M) &=& 0.
\end{IEEEeqnarray*}
The total achieved rate is then:
\(
R(M) = \sum_{i=1}^L R_i(M).
\)
\end{lemma}
We relegate the proof of Lemma~\ref{lemma:multi-user-achievability} to Appendix~\ref{app:multi-user-extras}.
What follows is a brief explanation of the individual rates.
\begin{itemize}
\item
Since levels $h\in H$ get no memory, no information about their files can be stored in the caches.
Hence, the server will have to transmit a complete copy of every file requested by the $KU_h$ users in level $h$.
Therefore, $R_h(M)=KU_h$ in the worst case.

\item
Since levels $j\in J$ get maximal memory, any file requested from $j$ can be fully recovered using the caches.
Thus there is no need for the server to send anything for $j$, and hence $R_j(M)=0$.

\item
Finally, the levels $i\in I$ get enough of the remaining memory $M-T_J$ so that they behave as in Theorem~\ref{thm:single-level}.
However, since $\alpha_iM\approx\sqrt{N_iU_i}/S_I \cdot (M-T_J)$, we get a rate of:
\[
R_i(M)
\approx \frac{N_iU_i}{\alpha_iM} - d_iU_i
\approx \frac{S_I\sqrt{N_iU_i}}{M-T_J} - d_iU_i.
\]
Notice how $R_i(M)$ is inversely proportional to the memory remaining after storing $J$.
This behavior is captured in the expressions for $I_0\cup I'$.
However, for the levels in $i_1\in I_1$, which get almost, but not quite, maximal memory, the individual rate $R_{i_1}(M)$ behaves more closely to a linear function of the memory.

\end{itemize}

\subsection{Outer bounds}
\label{sec:multi-user-converse}

In the single-level setup introduced in \cite{maddah-ali2012}, cut-set lower bounds on the optimal rate were sufficient for proving order-optimality of the scheme.
The idea was to choose cuts that include a certain number of caches $s$, as well as enough broadcast messages to decode almost all files when combined with the $s$ caches.
The choice of $s$ depended most importantly on the size of the memory $M$.
For example, for very large $M$, the bounds would reflect the fact that a single cache should contain enough information so that it need only be combined with a small broadcast message to decode one file.
In contrast, for small $M$, the bounds would instead show that a large broadcast message is needed to serve all $s$ users.

In our multi-level situation, the goal of our lower bounds is to show that different popularity levels require different amounts of memory, ranging from zero to maximal memory.
However, a single cut-set bound would force all levels to abide by the same number of caches $s$, a choice that would not reflect the plurality of memory allocation.
New lower bounds were thus required.

The lower bounds that we use allow choosing for each level $i$ the number of caches $s_i$ that corresponds to the amount of memory we expect it to get.
In essence, each level $i$ will exclusively be part of a cut-set bound with $s_i$ caches, which will give a lower bound that loosely corresponds to  the individual rate $R_i$ for this level.

The challenge is to incorporate all of these cut-set bounds together, such that the same caches appear in cut-set bounds for all levels.
This process is described next.
We start with a cut-set bound for the level with the smallest $s_i$.
Using Fano's inequality, we can derive lower bounds on its individual rate.
We then proceed to lower-bound the current quantity using a cut-set bound for the level with the next smallest $s_i$, and repeat the process.
The transition from one cut-set bound to the next can be achieved using sliding-window subset entropy inequalities \cite[Theorem~3]{Liu14}:%
\footnote{For these inequalities to work, we have to take not just one cut-set bound per level, but an averaging of $K$ similar cut-set bounds that ends up including all caches an equal number of times.}
\begin{lemma}[Sliding-window subset entropy inequality {\cite[Theorem~3]{Liu14}}]
\label{lemma:sliding-window}
Given $K$ random variables $(Y_1,\ldots,Y_K)$, we have, for every $s\in\{1,\ldots,K-1\}$:
\[
\frac1s\sum_{i=1}^K H\left(Y_i,\ldots,Y_{\langle i+s-1\rangle}\right)
\ge
\frac{1}{s+1}\sum_{i=1}^K
H\left( Y_i,\ldots,Y_{\langle i+s\rangle} \right),
\]
where we define $\langle i\rangle = i$ if $i\le K$ and $\langle i\rangle = i-K$ if $i>K$.
\end{lemma}

The resulting lower bounds, in their final form, are given in the following lemma.
\begin{lemma}
\label{lemma:multi-user-converse}
Consider the multi-level, multi-user caching setup.
Let $b\in\mathbb{N}^+$ and $t\in\{1,\ldots,K\}$.
Furthermore, for every level $i$, let $s_i\in\mathbb{N}^+$ such that $s_it\in\{d_i,\ldots,\floor{K/2}\}$.
Then, for every memory $M$, the optimal rate can be bounded from below by:
\[
R^\ast(M)
\ge \sum_{i=1}^L
\lambda_i
\cdot \min\left\{ (s_it-d_i+1)U_i \,,\,\frac{N_i}{s_ib} \right\}
- \frac tb M,
\]
where $\lambda_i$ is a constant (introduced for technical reasons) defined as: $\lambda_i=1$ if $s_it=d_i$; and $\lambda_i=\frac12$ if $s_it>d_i$.
\end{lemma}

The full proof of Lemma~\ref{lemma:multi-user-converse} is given in Appendix~\ref{app:multi-user-converse}.
However, we will here give a brief intuition behind the expression shown above.
Every term in the sum corresponds to a level $i$.
Consider a single level $i$.
Roughly, if we ignore all the other levels in the summation, and assuming $d_i=1$ for simplicity, then the inequality can be rearranged approximately as follows:
\[
s_ib R + s_it M \ge \min\left\{ s_it\cdot s_ib\cdot U_i \,,\, N_i \right\}.
\]
This expression is essentially a cut-set bound, saying that, with $s_it$ caches and $s_ib$ broadcast messages, up to $s_it\cdot s_ib\cdot U_i$ files from level $i$ can be decoded (since there are $U_i$ users per cache for this level), unless this number excedes the total number of files $N_i$.

\subsection{Approximate optimality}
\label{sec:multi-user-gap}

In order to prove order-optimality of the memory-sharing scheme, we must use Lemma~\ref{lemma:multi-user-converse} with appropriate parameters.
Our goal is to get a resulting lower bound on the optimal rate $R^\ast(M)$ such that the ratio between the achievable rate and $R^\ast(M)$---henceforth called the \emph{gap}---is minimized.
More specifically, the result we seek is a constant upper bound on the gap.

For technical reasons, several cases need to be considered for which different values are chosen for the parameters in Lemma~\ref{lemma:multi-user-converse}.
In particular, the values of choice in one case would violate the conditions imposed on them in another case.
In this section, we will illustrate the proof methodology by approximating one of these cases.
The complete and rigorous proof is given in Appendix~\ref{app:multi-user-gap}.

Recall the refined $M$-feasible partition introduced in Definition~\ref{def:m-feasible} and in \eqref{eq:refined-partition}.
As previously mentioned, this partition only depends on the problem parameters, not on the achievability strategy.
The cases of interest for the proof of order-optimality are:
\begin{itemize}
\item Case 1a: $I_1=\emptyset$ and $J\not=\emptyset$;
\item Case 1b: $I_1=\emptyset$ and $J=\emptyset$;
\item Case 2: $I_1\not=\emptyset$.
\end{itemize}
There is also a special Case 0 for when $K$ is small (specifically, $K<D/\beta$).
All the other cases assume $K$ is large.

The case that we will focus on in this section is Case 1a.
Note that $I_1=\emptyset\implies I=I_0\cup I'$.
For simplicity, we will assume that $d_i=1$ for all $i$.
Furthermore, most of the analysis will be approximate.

By Lemma~\ref{lemma:multi-user-achievability}, the achievable rate in Case 1a can be upper-bounded by:
\begin{equation}
\label{eq:approx-achievable-rate}
R(M) \le \sum_{h\in H} KU_h + \frac{2S_I^2}{M-T_J},
\end{equation}
since $S_I=\sum_{i\in I}\sqrt{N_iU_i}$ (see Definition~\ref{def:m-feasible}).
Now consider Lemma~\ref{lemma:multi-user-converse} with the following parameters:
\begin{IEEEeqnarray*}{lCl"rCl}
&& & t &=& 1;\\
\forall h&\in&H, & s_h &\approx& \frac12K;\\
\forall i&\in&I, & s_i &\approx& \frac{\sqrt{N_i/U_i}}{2\tilde M};\\
\forall j&\in&J, & s_j &=& 1;\\
&& & b &\approx& 4 \tilde M^2,
\end{IEEEeqnarray*}
where, as in Definition~\ref{def:m-feasible}, $\tilde M = (M-T_J+V_I)/S_I \approx (M-T_J)/S_I$.
The values of the $s_i$'s are very similar to those used in the single-level setup in \cite{maddah-ali2012}.
Indeed, the levels with the smallest memory ($h\in H$) are handled using cut-set bounds that consider a fraction of the total number of caches, while the levels with the largest memory ($j\in J$) are handled with bounds that consider only one cache.
Thus, the lemma lower-bounds the optimal rate by:
\begin{IEEEeqnarray*}{rCl}
R^\ast(M)
&\ge& \sum_{h\in H} \min\left\{ s_hU_h, \frac{N_h}{s_hb} \right\}
+ \sum_{i\in I} \min\left\{ s_iU_i, \frac{N_i}{s_ib} \right\}\\
&&{} + \sum_{j\in J} \min\left\{ s_jU_j, \frac{N_j}{s_jb} \right\}
- \frac{M}{b}.
\end{IEEEeqnarray*}
Substituting the values of the parameters, and utilizing the inequalities on $\tilde M$ that define the refined $M$-feasible partition, we get:
\begin{IEEEeqnarray*}{rCl}
R^\ast(M)
&\ge& \sum_{h\in H} \frac12 KU_h
+ \sum_{i\in I} \frac{\sqrt{N_iU_i}S_I}{2(M-T_J)}
+ \sum_{j\in J} \frac{N_j}{b}
- \frac{M}{b}\\
&=& \sum_{h\in H} \frac12KU_h
+ \frac{S_I^2}{2(M-T_J)}
- \frac{M-T_J}{b}\\
&=& \sum_{h\in H} \frac12KU_h
+ \frac{S_I^2}{2(M-T_J)}
- \frac{(M-T_J)S_I^2}{4(M-T_J)^2}\\
&=& \sum_{h\in H} \frac12KU_h
+ \frac{S_I^2}{4(M-T_J)}\\
&\overset{(a)}{\ge}& \frac{1}{8}\cdot R(M).
\end{IEEEeqnarray*}
The last inequality $(a)$ is due to \eqref{eq:approx-achievable-rate}.

Thus, the result is a constant multiplicative gap ($8$ in this example) between the achievable rate and the optimal rate for this regime.
The full proof in Appendix~\ref{app:multi-user-gap} derives such gaps for each one of the four regimes mentioned above (Cases 0, 1a, 1b, and 2), and finally combines them into one multiplicative gap that holds for all cases.

\section{The Single-User Setup}
\label{sec:single-user}

\subsection{Caching-and-delivery strategy: clustering}
\label{sec:single-user-achievability}

\begin{proof}[Proof of Theorem~\ref{thm:single-user-achievability}]
Recall from Section~\ref{sec:single-user-results} how the memory is divided among the sets $H'$ and $I'$, defined in \eqref{eq:G-partition}: all of the available memory is given to levels in $I'$, which is treated as one super-level.
As a result, all requests for files from $H'$ must be handled by complete file transmissions from the BS.
Since there are $\sum_{h\in H'}K_h$ users making such requests, the result is the same amount of transmissions in the worst case.
Therefore, the message sent to all users requesting from a level in $H'$ has the following rate:
\begin{equation}
\label{eq:su-ach-h}
R_{H'} = \sum_{h\in H'} K_h.
\end{equation}

For the set $I'$, now considered as one super-level, we use the single-level strategy from \cite{maddah-ali2013}.
Although only a subset of the caches is active in our setup, the same strategy still applies.
Indeed, the placement in \cite{maddah-ali2013} is a random sampling of the files in all the caches; we do the same placement in this case.
In the delivery phase, we now know the caches to which the users of $I'$ connected.
We perform a delivery as in \cite{maddah-ali2013}, assuming that only these caches were ever present in the system.

For illustration, suppose that there were $K=4$ caches.
Furthermore, assume the partition $(H',I')$ was performed such that $3$ users will request files from the set of levels $I'$.
In the placement phase, we store a random sample of the files in $I'$ in each of the four caches.
In the delivery phase, suppose the three users requesting files from $I'$ connect to caches $1$, $2$, and $4$.
In this case, the BS will send a complete file transmission for user $3$ (since he requests a file from $H'$), and will treat users $1$, $2$, and $4$ as though they were part of a single-level caching system with only three caches.

The rate required for $I'$ can be directly derived from Lemma~\ref{lemma:single-level}, using $\sum_{i\in I'} K_i$ caches, $\sum_{i\in I'} N_i$ files, and $1$ user per cache.
In addition, we have, from \eqref{eq:G-partition}, that $M\ge N_i/K_i$ for all $i\in I'$.
This implies
\(
M \ge (\sum_{i\in I'}N_i)/(\sum_{i\in I'}K_i),
\)
and hence the rate for $I'$ is:
\begin{equation}
\label{eq:su-ach-i}
R_{I'} = \max\left\{  \frac{\sum_{i\in I'}N_i}{M}-1 \,,\, 0 \right\}.
\end{equation}
The maximization with zero is needed because it is possible to have $M>\sum_{i\in I'}N_i$.

By combining \eqref{eq:su-ach-h} with \eqref{eq:su-ach-i}, we get a total broadcast rate of:
\[
R(M) = R_{H'} + R_{I'}
= \sum_{h\in H'} K_h
+ \max\left\{ \frac{\sum_{i\in I'}N_i}{M}-1 , 0 \right\}.
\]
This proves Theorem~\ref{thm:single-user-achievability}.
\end{proof}

It will be helpful for the later analysis to refine the partition $(H',I')$ as follows.
\begin{definition}
\label{def:G-partition}
Define the following partition $(G,H,I,J)$ of the set of levels:
\begin{IEEEeqnarray*}{rCl}
G &=& \left\{ g : M < N_g/K_g \text{ and } K_g\le5 \text{ and } M \le N_g/6 \right\};\\
H &=& \left\{ h : M < N_h/K_h \text{ and } K_h\ge6 \right\};\\
I &=& \left\{ i : N_i/K_i \le M \le N_i/6 \right\};\\
J &=& \left\{ j : M > N_j/6 \right\}.
\end{IEEEeqnarray*}
\end{definition}

We rewrite the achievable rate in terms of this new partition:
\begin{equation}
\label{eq:single-user-achievable-rate}
R(M)
\le 5\cdot|G|
+ \sum_{h\in H} K_h
+ \frac{\sum_{i\in I}N_i}{M}
+ \left[ \frac{\sum_{j\in J}N_j}{M} - 1 \right]^+,
\end{equation}
where $[x]^+=\max\{x,0\}$, and keeping in mind that $K_g\le5$ for $g\in G$.
If we define $N_J=\sum_{j\in J}N_j$, we can upper-bound the last term by:
\begin{equation}
\label{eq:su-ub-J}
\left[ \frac{N_J}{M} - 1 \right]^+
\le
\begin{cases}
N_J/M & \text{if $M<N_J/6$;}\\
6\left( 1 - M/N_J \right) & \text{if $N_J/6\le M < N_J$;}\\
0 & \text{if $M\ge N_J$.}
\end{cases}
\end{equation}
Effectively, the set $J$ now behaves as a unit.

Much of what defines the above partition hinges on the following question: What happens if we know the user profile \emph{a priori}?
If we do, we could imagine a strategy of total separation of the levels: every level $i$ is treated as a single-level system with $K_i$ caches and users, and $N_i$ files.
Each level can thus get the entirety of the memory of their specific $K_i$ caches.
Using Lemma~\ref{lemma:single-level}, we can calculate the total rate for this hypothetical situation as:
\[
R'(M)
= \sum_{i=1}^L
\min\left\{ K_i, \frac{N_i}{M} \right\} \left( 1 - \frac{M}{N_i} \right)^+,
\]
which can be approximated by:
\[
R'(M)
\le 5\cdot|G| + \sum_{h\in H} K_h + \sum_{i\in I} \frac{N_i}{M}
+ 6\sum_{j\in J}\left( 1-\frac{M}{N_j} \right)^+.
\]
Interestingly, the $G$, $H$, and $I$ terms in this rate expression are exactly the same as those in \eqref{eq:single-user-achievable-rate}.
Therefore, except for levels $j\in J$ where $M>N_j/6$, our scheme would not benefit from prior knowledge of the user profile.
However, the set $J$ is in general limited by the lack of this knowledge.
Indeed, if $M>N_j$ for all $j\in J$ but $M<N_J=\sum_{j\in J} N_j$, then it is possible to store any level in every cache, but it is not possible to store all levels in all caches.
Thus knowing which level is at which cache can bring the rate down to zero.

Finally, $G$ and $H$, which are both subsets of $H'$ (though not necessarily a partition; the definition of $J$ allows some of its levels to be in $H'$ too) were separated because levels behave differently when their number of caches is very small.
Specifically, a level in $H$ can transition into $I$ as the memory increases, but a level in $G$ immediately jumps to $J$ since $N_g/K_g>N_g/6$.

\subsection{Outer bounds and approximate optimality}
\label{sec:single-user-gap}

As mentioned in Section~\ref{sec:single-user-results}, we use a cut-set bound to lower-bound the optimal rate.
The idea is to send a certain number $b$ of broadcast messages $X_1,\ldots,X_b$ that serve certain requests.
We choose these requests as follows.
For every level $i\in G\cup H\cup I$, consider a certain number $s_i\le K_i$ of caches.
These caches are distinct across levels.
For all the $b$ broadcasts, the users connected to these $s_i$ caches will altogether request $s_ib$ distinct files from level $i$ if there are that many; otherwise they request all $N_i$ files.
For the levels in the set $J$, we collectively consider some $s_J$ caches (distinct from the rest).
The users at these $s_J$ caches will use all $b$ broadcasts to decode as many files from the set $J$ as possible, up to $s_Jb$ files.
Let $n_J$ denote this number; it will be determined later.

If we let $S=\sum_{i\not\in J}s_i + s_J$ be the total number of caches considered, then, by Fano's inequality:
\begin{IEEEeqnarray*}{rCl}
bR + SM
&\ge& H\left( Z_1,\ldots,Z_S,X_1,\ldots,X_b \right)\\
&\ge& \sum_{i\not\in J} \min\left\{ s_ib,N_i \right\} + n_J\\
R^\ast(M) &\ge& \sum_{i\not\in J} s_i\left( \min\left\{ 1,\frac{N_i}{s_ib} \right\} - \frac{M}{b} \right)\\
&& {} + s_J\left( \frac{n_J}{s_Jb} - \frac{M}{b} \right)\\
&=& \sum_{i\not\in J} v_i + v_J. \IEEEyesnumber\label{eq:single-user-lower-bounds}
\end{IEEEeqnarray*}

We will analyze each of the $v_i$ and $v_J$ terms separately.
We identify two cases for which the analysis is slightly different.

The first case is when $M<1/6$.
Because of regularity condition \eqref{eq:su-reg}, this implies $M<1/6<1\le N_i/K_i$ for all levels $i$, and thus the achievable rate can be bounded by:
\begin{equation}
\label{eq:su-ub-0}
R(M) \le \sum_{i=1}^L K_i.
\end{equation}
The second case, which is more interesting, is when $M\ge1/6$.

The details of the analysis are given in Appendix~\ref{app:single-user-gap}.
We will here give a brief outline of the procedure, with approximations to avoid burying the essence of the argument under technicalities.
The idea is to lower-bound each term $v_i$ or $v_J$ to match the corresponding term in the rate expression.
Let us focus on the case $M\ge1/6$, and suppose we choose $b\approx6M\ge1$.
For simplicity, we will only look at the sets $H$ and $I$.

Consider a level $h\in H$, and let $s_h\approx K_h/6$.
Then, by the definition of $H$:
\[
\frac{N_h}{s_hb} = \frac{N_h}{K_hM} \ge 1.
\]
As a result,
\[
v_h = \frac{K_h}{6}\left( 1 - \frac{M}{6M} \right) = \frac{5}{36}\cdot K_h,
\]
which matches the corresponding achievable rate term in \eqref{eq:single-user-achievable-rate}, up to the constant $(5/36)$.

Now consider a level $i\in I$, and let $s_i\approx N_i/(6M)$.
Then,
\[
\frac{N_i}{s_ib} = \frac{N_i}{(N_i/6M)\cdot6M} = 1,
\]
and hence:
\[
v_i = \frac{N_i}{6M}\left( 1 - \frac{M}{6M} \right) = \frac{5}{36}\cdot\frac{N_i}{M}.
\]
Again, this matches the corresponding achievable rate term in \eqref{eq:single-user-achievable-rate}, up to the constant $(5/36)$.

Applying a similar procedure for every level, we get matching lower bounds (up to a constant) and thus prove Theorem~\ref{thm:single-user-gap}.

\section{Comparison of the Two Setups}
\label{sec:comparison}

In this section, we first compare the memory-sharing and the
clustering strategies, and we explore the dichotomy among the two
setups that is emphasized by the difference between strategies.  We
will then discuss why such a dichotomy exists, and explain the need
for different lower bounds for each setup.  Finally, we explore a new
problem that combines both setups by including both multi-user and
single-user levels.

\subsection{Comparing the two caching-and-delivery strategies}

We have previously argued that memory-sharing is the best scheme to use in the multi-user case, while clustering is the near-optimal strategy in the single-user case.
However, could one (or both) of these schemes be good enough for both situations?
We will show, in this section, that it is not the case: in the single-user setup, memory-sharing can give a gap between its rate and the optimum that increases linearly with $L$; meanwhile, the rate achieved by clustering in the multi-user case can be arbitrarily far from the optimal rate.
We give examples of these two cases.

Consider a multi-user setup with two levels.
Suppose that there is enough memory so that both levels are to be partially stored in the caches.
With the memory-sharing scheme, this means $I=\{1,2\}$, which would, by Theorem~\ref{thm:multi-user-achievability}, give a rate of approximately:
\begin{IEEEeqnarray*}{rCl}
R
&\approx& \frac{\left( \sqrt{N_1U_1}+\sqrt{N_2U_2} \right)^2}{M}\\
&=& \frac1M\left[ N_1U_1 + N_2U_2 + 2\sqrt{N_1U_1N_2U_2} \right].
\end{IEEEeqnarray*}
On the other hand, if we had clustered the two levels into one, then this super-level would have $(N_1+N_2)$ files and $(U_1+U_2)$ users per cache, resulting in the following rate:
\begin{IEEEeqnarray*}{rCl}
R
&\approx& \frac{(N_1+N_2)(U_1+U_2)}{M}\\
&=& \frac1M \left[ N_1U_1 + N_2U_2 + N_1U_2 + N_2U_1 \right].
\end{IEEEeqnarray*}

However, we know that the geometric mean of any two numbers is always smaller than their arithmetic mean.
By considering the two numbers $N_1U_2$ and $N_2U_1$, we get:
\[
2\sqrt{N_1U_2N_2U_1} \le N_1U_2 + N_2U_1,
\]
and specifically the ratio between them can get arbitrarily large when the popularities of the two levels become significantly different (i.e., $U_2/N_2 \ll U_1/N_1$).
Intuitively, if the two levels had similar popularities, then memory-sharing gives them similar amounts of memory, effectively merging them.
However, if their popularities were very different, then they should be given drastically different portions of the memory.

Consider now the single-user case with $L$ levels, and suppose again that the memory is such that all levels will be partially stored.
Let us assume that $N_1=\cdots=N_L$.
Using the clustering scheme, we get the following approximate rate:
\[
R \approx \frac{N_1+\cdots+N_L}{M}
= \frac{LN_1}{M}.
\]
However, with memory-sharing, we would get:
\[
R \approx \frac{\left( \sqrt{N_1}+\cdots+\sqrt{N_L} \right)^2}{M}
= \frac{L^2N_1}{M},
\]
which is larger by a factor of $L$.
Essentially, we are sending $L$ broadcasts, one per level, when we could send just one broadcast for all $L$ levels.

\subsection{Analysis of the dichotomy between the setups}

The dichotomy between the two extremes is striking.
They require different strategies, and the strategy that is good for one setup is not so for the other.
This suggests a fundamental difference between the two setups.

To understand this difference, consider what happens when sending a coded broadcast message.
Each message targets a specific subset of users.
If, in this subset, there exist two users that are connected to the same cache, then these users have access to the exact same side information.
As a result, no coding can be done across these two users, and there is hence no benefit in including them in the same broadcast.

There are in fact two opposing forces at work on the caching-and-delivery strategy.
The first is a popularity-centric force: it pushes on the strategy to allocate more memory to the more popular files.
The second force is coding-centric: it encourages increasing the number of coding opportunities.
These two forces are at odds, since coding \emph{across levels} performs best when the files in these levels are given the same memory, regardless of popularity.

With that in mind, consider again \figurename~\ref{fig:setup-F} and \figurename~\ref{fig:setup-G}.
Notice how, in the multi-user setup, there are multiple rows of users, each of which consists of users from the \emph{same} popularity level.
Each such row is a complete set of users with no common caches: any additional users we add would have access to the same cache as some other user.
Thus, it is sufficient to consider them in a broadcast transmission that is separate from all other rows.
Since, as a result, no two levels will share the same broadcast message, it can only be beneficial to choose the best possible division of the memory, based on popularities.
In the single-user setup, however, there is only one row of users that contains all the users from all the levels.
It is hence possible to generate coding opportunities across levels.
Merging is thus a better option in this situation, and merging is most efficient when all levels receive equal memory per file.

\subsection{The difference in the lower bounds}

Complementing the difference in the achievable strategies between the two setups, we see a difference in the lower bounds to the optimal rate.
In the multi-user case, we use a combination of cut-set bounds, one for each popularity level.
However, we use a unified cut-set bound for all the levels in the single-user case.
We elaborate on this difference in this section.

In the single-level setup studied in \cite{maddah-ali2012}, cut-set bounds were given to lower-bound the optimal rate for every value of memory $M$.
Depending on the value of $M$, a certain number of caches (and hence users) were considered and used in the cut-set bounds.
Roughly speaking, about $M$ broadcast messages are sent to the users at $N/M$ distinct caches, allowing them to decode $M\cdot (N/M)=N$ files in total.
In the multi-level setups (both multi-user and single-user), the lower bounds retain this idea.
However, There are two crucial differences between the setups that force different choices of lower bounds.

The first difference is in the role of each cache vis-\`a-vis the popularity levels.
In both setups, different levels are given different memory values.
However, the \emph{same} cache must be simultaneously used for \emph{all} levels in the multi-user setup; in the single-user setup each cache is bound to a single level at any moment.
As a result, a single cut-set bound can still encompass all levels in the single-user setup, but will not be enough in the multi-user setup.

The second difference is in the uncertainty of the user profile.
In the single-user setup, there are situations where this uncertainty is significant enough to impact the achievable rate.
This is especially true for levels nearing their maximal storage (the set $J$ in Definition~\ref{def:G-partition}), as described in Section~\ref{sec:single-user-achievability}.
The lower bounds should incorporate this notion by considering demands from different levels at the same cache.
As a result, a single cut-set bound unifying all levels (at least the levels involved in this uncertainty) becomes necessary.

\subsection{Mixing the setups}

So far, we have looked at the two extremes: either all levels were represented at all the caches, or none of them were.
A natural problem arises: that of studying intermediate cases.
The simplest form such intermediate cases can take is one where levels of both types are present.

Specifically, there are two classes of popularity levels: $\mathcal{F}$ and $\mathcal{G}$.
The class $\mathcal{F}$ consists of levels $i$ that are represented by exactly $U_i$ users \emph{at every cache}.
In contrast, there is exactly one row of users that represents all the levels in the class $\mathcal{G}$: each level $i\in\mathcal{G}$ is represented by $K_i$ of those users.

The most natural strategy to employ in this situation would be to superpose the multi-user and the single-user strategies.
In particular, we divide the memory $M$ into $\gamma M$ and $(1-\gamma)M$, for some $\gamma\in[0,1]$.
We give the first part to $\mathcal{F}$ and the second part to $\mathcal{G}$, and apply their respective strategies on their part of the memory.
We believe this to be the best strategy, but proving its order-optimality requires developing new lower bounds that consider levels of both classes at the same time; this is part of our on-going work.

\section{Discussion and Numerical Evaluations}
\label{sec:discussion}

In the previous sections, we presented theoretical results for any given set of popularity levels and associated user access structures.
However, in practice, what is available is a ``continuous'' popularity distribution over the entire set of $N$ files, and it is up to the designer to choose: \textsf{(a)} the number of popularity levels; \textsf{(b)} which files to assign to which level; and \textsf{(c)} the corresponding user access degree for each popularity level.
For each such choice, our theoretical results characterize the minimum broadcast transmission rate, and we study, in this section, the impact of these choices on the transmission rate.
Furthermore, while our theoretical model assumed that, for each popularity level, the number of users fixed, we relax this assumption here by allowing each user to randomly connect to one of the $K$ APs and request a file stochastically, according to the underlying popularity distribution.
Finally, we will also compare the performance of our scheme with that of the traditional LFU approach, as well as the information-theoretic lower bounds presented before.
These evaluations will only focus on the multi-user setup, as it is less discussed in the literature.
Since the single-user setup utilizes a scheme similar to what is already in the literature for arbitrary distributions \cite{ZhangArbitrary,Zcodedcaching}, we do not feel it is necessary to include it in this discussion.
However, we do provide a brief comparison of the clustering scheme with the literature in Section~\ref{sec:discussion-clustering}

We use a YouTube dataset \cite{YoutubeRepository} for our evaluations.
This dataset provides the number of views of videos in a set of $N\approx500\,000$, over some period of time.
Thus, these views can be thought of as approximating the popularity of the videos.
\figurename~\ref{fig:youtube} shows the popularity distribution of the videos (normalized number of views), which resembles a Zipf distribution similar to those commonly observed for multimedia content \cite{breslau1999web}.

When using the YouTube dataset in the following sections, we will often omit the total number of users considered.
This is because the total broadcast rate is always directly proportional to the total number of users, when the \emph{fraction} of users per level is fixed.
In this situation, the fraction is determined by the distribution in \figurename~\ref{fig:youtube} and the levels into which the files were split, thus the total number of users will not affect the behavior of the system.

\subsection{Discretizing a continuous popularity distribution}
\label{sec:discretizing}

Our first step is to divide the files in the YouTube dataset into a certain number of levels, based on the popularity profile in \figurename~\ref{fig:youtube}.
Let there be $K=75$ caches.
We consider small, moderate, and large values of $M / N$ ($0.03$, $0.2$, and $0.7$) and set the user access degree $d_i = 1$ for every level $i$, so as to study the impact of the number of levels on the broadcast rate in isolation.
For increasing values of $L$, we find the division of the files into $L$ levels that minimizes the rate achieved by the memory-sharing scheme using a brute-force search.
We plot the minimum achievable rate versus $L$ in \figurename~\ref{fig:youtube-R-vs-L}.
As is easily apparent, while there is a significant gain in performance between treating all files as one level and dividing them into two levels, the gain    decreases with diminishing returns as $L$ increases.
This shows the importance of dividing files into multiple levels, but also suggests that 3--4 levels are sufficient to derive most of the benefits. 

\begin{figure}
\centering
\includegraphics[width=\myplotwidth]{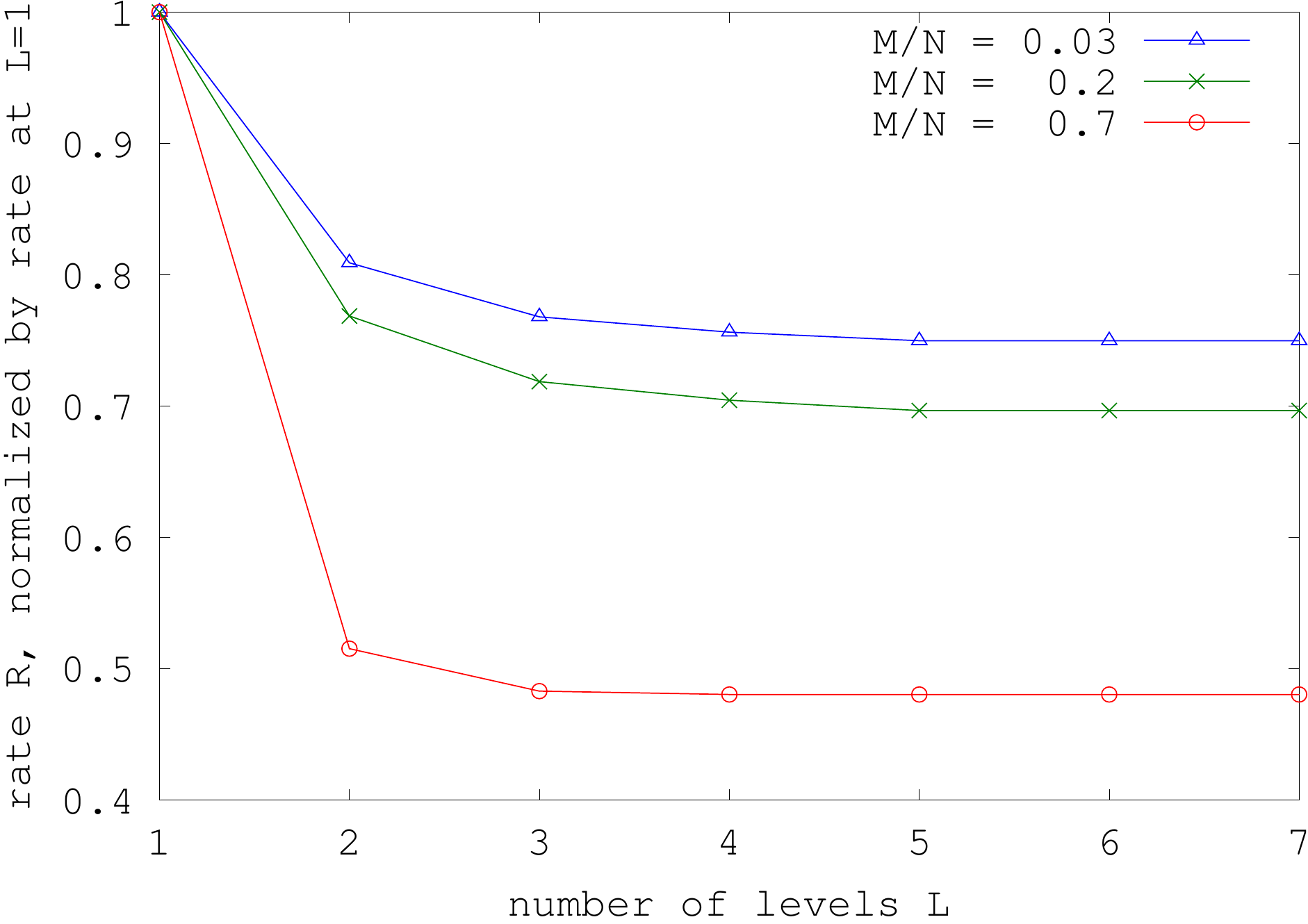}
\caption{Rate achieved by the memory-sharing scheme vs.\ number of levels, for different values of cache memory.
For each $L$, we choose the $L$ levels that minimize the achievable rate (using brute-force search).
For ease of comparison, the rate values have been normalized by the rate at $L=1$.}
\label{fig:youtube-R-vs-L}
\end{figure}

We remind the reader that the popularity profile in \figurename~\ref{fig:youtube} is purely empirical.
It is based on the number of views of the videos in the dataset, collected over some period of time.
By the very nature of the data, if two videos have received, let's say, 1000 and 1100 views, then this does not really imply that the first file is more popular than the other.
This is especially true of videos with very small number of views.
By grouping files into popularity classes, we acknowledge the difference in popularity of very different files, while simultaneously not distinguishing between files whose \emph{empirical} popularities are close.

\subsection{Impact of multi-access on the achievable rate}
To study the effect of multi-access in isolation, we will fix the partition we use to divide the files into different levels and then look at different multi-access structures.
Suppose again $K=75$, and consider $L = 2$ levels, with $N_1 = 0.2N$, and $N_2 = 0.8N$ files.
We plot in \figurename~\ref{fig:youtube-multi-access} the broadcast rate of our scheme as a function of the normalized memory $M / N$, for four different access structures $(d_1, d_2)$: $(1,1)$, $(1, 2)$, $(2,1)$, and $(2,2)$.

As one would expect, allowing for multi-access greatly improves the transmission rate.
For example, the rate for the multi-access system with $(d_1 = 2, d_2 = 2)$ is smaller than the rate for the single-access system with $(d_1 = 1, d_2 = 1)$.
The cases $(d_1 = 1, d_2 = 2)$ and $(d_1 = 2, d_2 = 1)$ provide a more interesting comparison.
For small memory size $M$, the former gives a lower rate since the cache memory mainly contains files from level $1$, and so giving higher access to level $1$ is more beneficial in reducing the rate.
On the other hand, as $M$ grows and files from level $2$ start occupying a significant portion of the memory, it becomes more efficient to give higher access to level $2$ since it has many more files than level $1$. 

While greater cache access helps reduce the rate, there is also a cost associated with it in terms of the increased delay in gathering data from multiple APs, as well as a reduced rate as a user connects to farther APs.
In general, for a given multi-level setup with parameters $L$, $K$, $\{N_i,U_i\}$, and $M$, such a cost can be included in the rate optimization framework
as one or more inequalities of the form $\mathrm{cost}_j(K,\{U_i,d_i\}_i)\le C_j$, for some maximum cost $C_j$.
The above optimization problem can be numerically solved by a designer in order to identify the optimal access structure for the multi-level system under consideration.
However, to derive some intuition about how the costs impact the optimal multi-access structure, let us consider a setup with $L = 3$ levels, and with $N_1 = 0.04N$, $N_2 = 0.13N$, and $N_3 = 0.83N$ files in the three levels.
Say we want to include both a maximum degree constraint $d_i \le 3$ for each level $i$, as well as an average degree constraint $(\sum_i U_i d_i ) / U \le 2$.
Then, \figurename~\ref{fig:youtube-dvM} plots the optimal access structure vs.\ the normalized memory size.
As before, when the memory is small, the optimal access structure is one which satisfies $d_1\ge d_2\ge d_3$, but this relation becomes reversed as the memory increases.

\begin{figure}
\centering
\includegraphics[width=\myplotwidth]{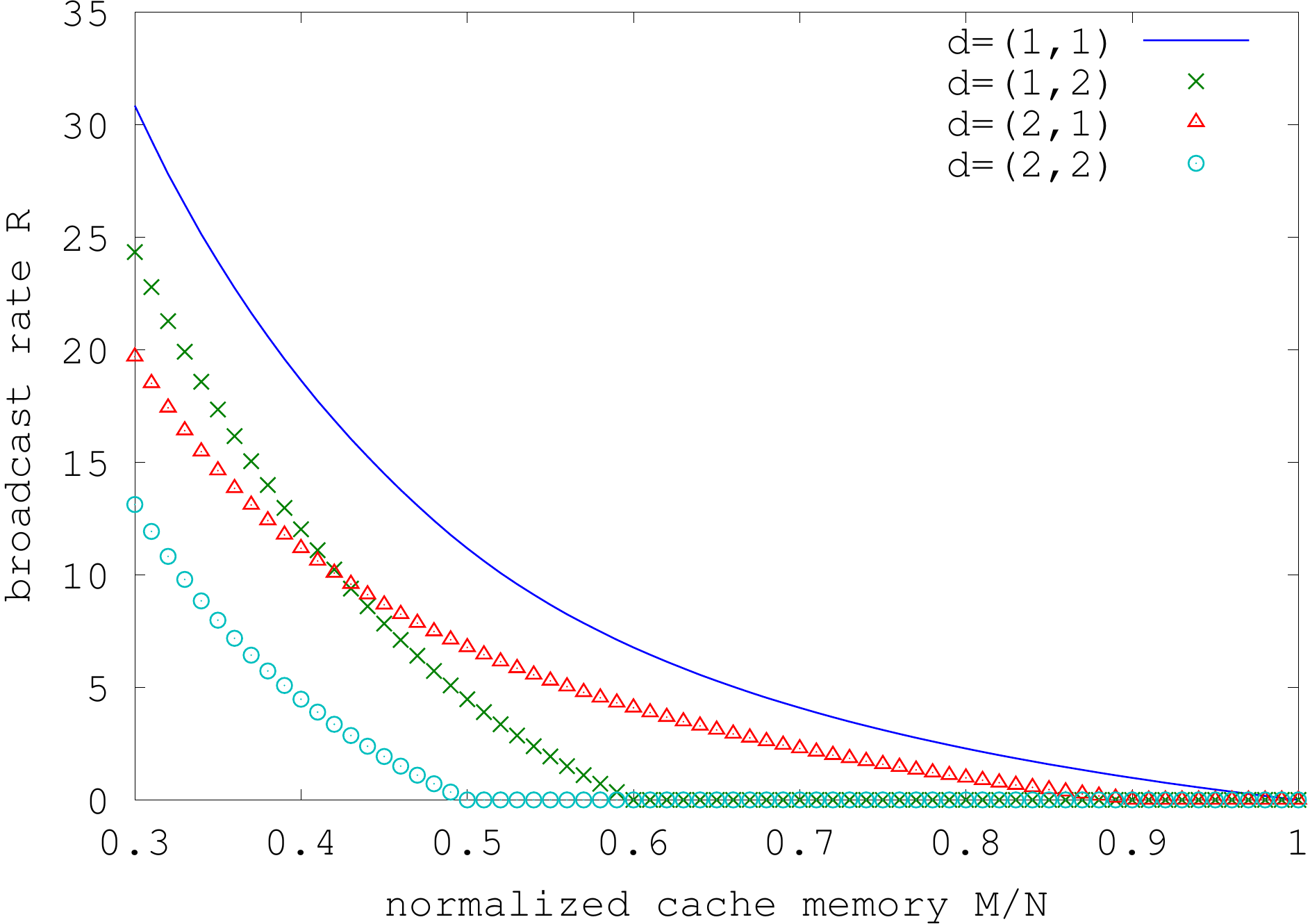}
\caption{Achievable rate vs.\ cache memory in a two-level setup, for different access structures.}
\label{fig:youtube-multi-access}
\end{figure}

\begin{figure}[t]
\centering
\includegraphics[width=\myplotwidth]{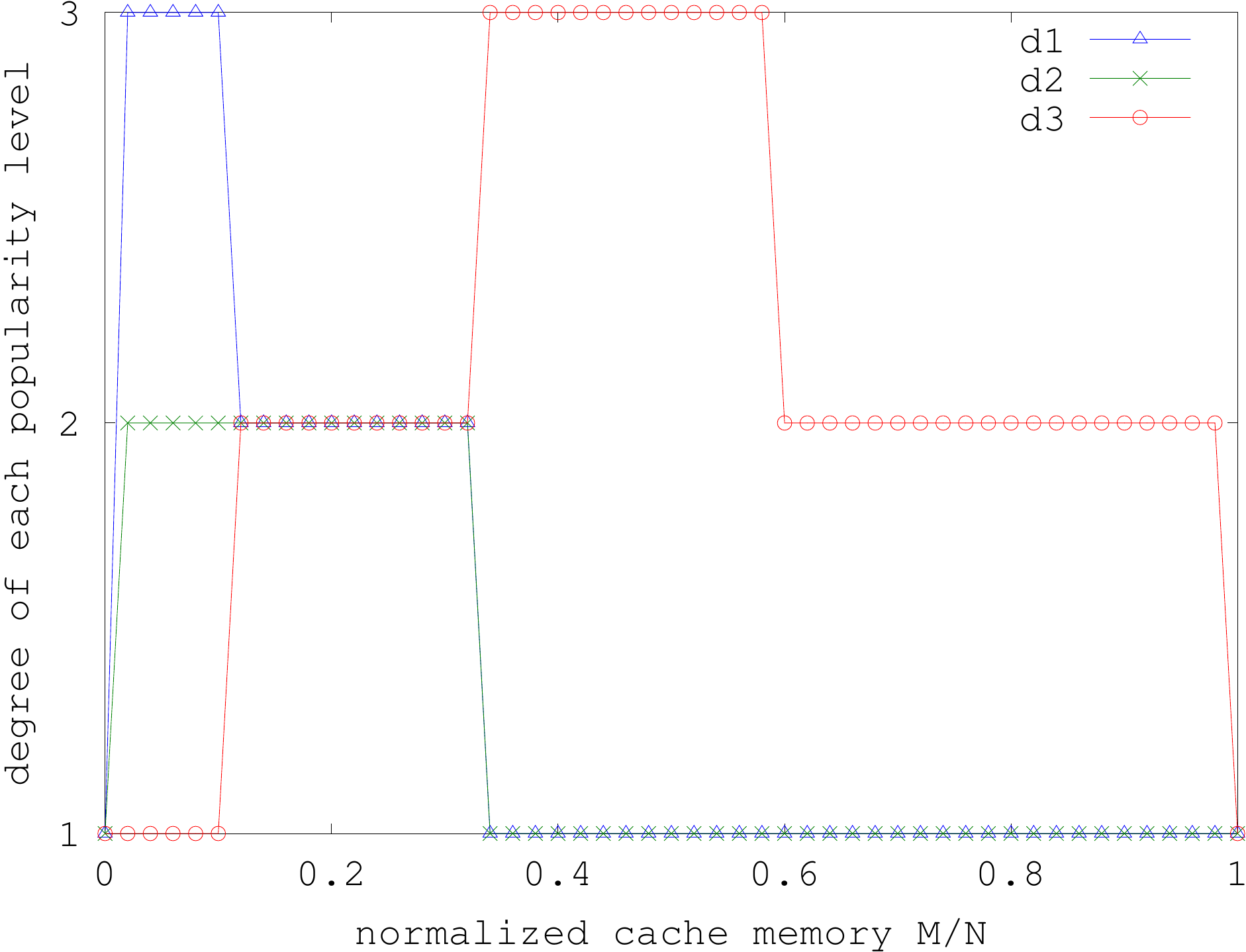}
\caption{Optimal access structure vs.\ memory, with $d_{\max}=3$, $d_\mathrm{avg}=2$.}
\label{fig:youtube-dvM}
\end{figure}

\subsection{Stochastic variations in user profiles}

The theoretical setup and results presented in the previous sections assumed a symmetric and deterministic user profile across all the APs.
In particular, exactly $U_i$ users are assigned to each AP to request files from level $i$.
This section aims at evaluating the robustness of memory-sharing to asymmetry and stochasticity in the user profiles across caches.

We consider a setup where each of the $KU$ users in the system randomly connects to one of the $K$ APs and requests a file stochastically, according to the YouTube popularity distribution in \figurename~\ref{fig:youtube}.
The scheme we use here is a simple variation of the one for the worst-case setup: the files are split into two levels, and the placement is done using memory-sharing based on their average popularity.
The delivery phase is almost identical to the worst-case delivery, with the exception that, because of a lack of determinism in the user profile, not all caches will have the same number of users per level.
This is handled by simply trying to group as many users of the same level as possible in each broadcast transmission.

We ran simulations for this setup using the above strategy, and we plot here the empirically achieved rate against the cache memory in \figurename~\ref{fig:simulations}.
For comparison, we also show the rate predicted by our theoretical model, which splits the files into two levels and assumes a symmetric user profile across the caches.
Clearly, the theory very closely predicts the empirical results for a random user profile, thus demonstrating the robustness of our theoretical results to stochastic variations across APs and justifying their utility in practice.  

\begin{figure}
\centering
\includegraphics[width=\myplotwidth]{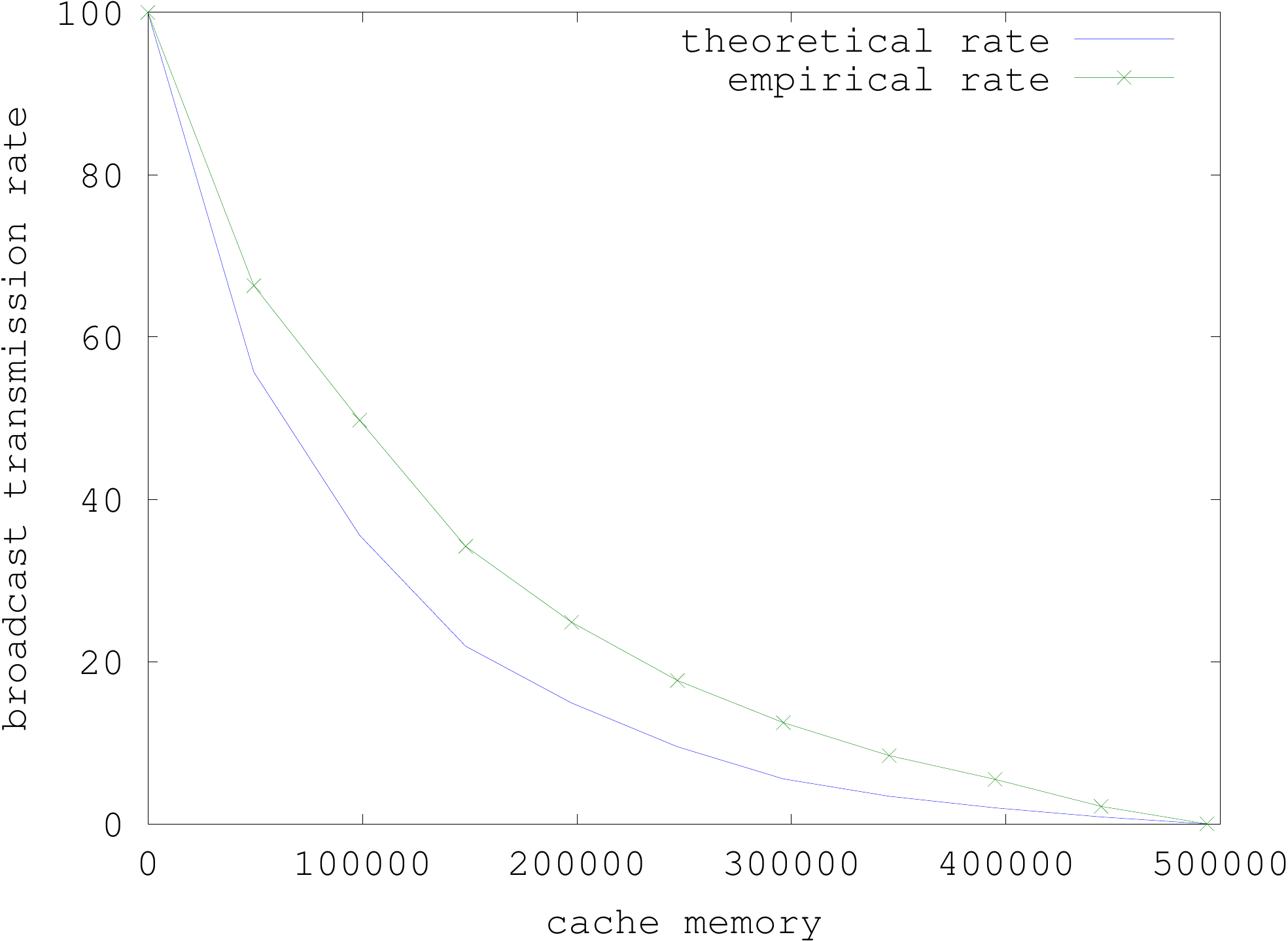}
\caption{Comparison of the theoretical rate with the empirical rate, based on simulations of demands over the YouTube dataset, with $5$ caches and $100$ total users.
The theoretical rate is off by a factor of up to $2.8$ from the empirical.}
\label{fig:simulations}
\end{figure}

\subsection{Comparison with Least-Frequently Used (LFU)}

In this section, we compare the performance of memory-sharing with that of the traditional LFU scheme using simulations on the YouTube data.
For any memory size $M$, LFU fully stores the $M$ most popular files, so that requests for more popular files are completely served from the cache, and requests for less popular files are fully handled by the BS transmission.
The results, given in \figurename~\ref{fig:lfu-comparison}, show the superiority of memory-sharing.

\begin{figure}
\centering
\includegraphics[width=\myplotwidth]{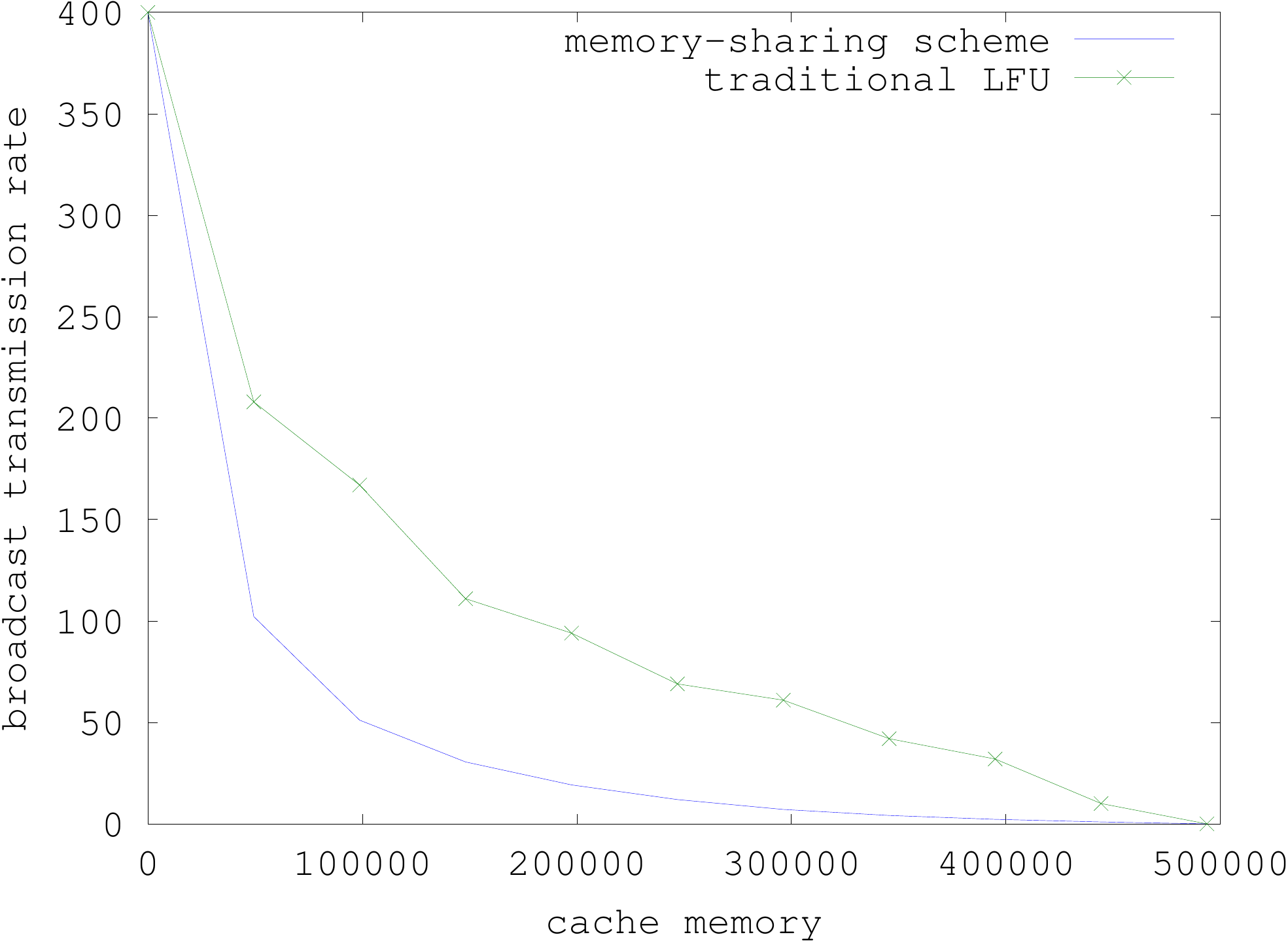}
\caption{Comparison of the memory-sharing scheme with traditional LFU.
The memory-sharing scheme achieves up to a factor-$14.5$ in gain over LFU.}
\label{fig:lfu-comparison}
\end{figure}

\subsection{Comparison with different memory-sharing strategies}

The memory-sharing scheme proposed for the multi-user setup relies on a very specific division of the memory among the levels.
A natural question that arises is if this memory-sharing is necessary.
In this section, we compare its behavior with that of different, perhaps more natural, memory-sharing schemes.
In particular, we compare with:
\begin{itemize}
\item \textbf{Coded LFU:}
This scheme gives as much memory as possible for the most popular levels.
We call it ``coded LFU'' because it allocates memory in a similar way to traditional LFU, but delivers a coded broadcast message to the users.
\item \textbf{Uniform memory-sharing:}
This scheme ignores all popularities and gives the same amount of memory to all files.
\end{itemize}

Our memory-sharing scheme---which we will call ``optimal memory-sharing'' in this section to distinguish it from the others---performs better than the other two.
We note, however, that coded LFU and uniform memory-sharing are both special cases of optimal memory-sharing.
Indeed, if all levels have exactly the same popularities, then optimal memory-sharing will give all files equal memory.
On the other hand, when the level popularities are separted enough, then optimal memory-sharing will end up always prioritizing the most popular levels, and thus it reduces to coded LFU.
In the middle case, it provides benefits that both other schemes lack.

As an example, consider the following setup.
There are $K=100$ caches, and $L=3$ levels with $(N_1,N_2,N_3)=(2000,5000,50\,000)$.
We have $(U_1,U_2,U_3)=(20,10,5)$, and single-access degree $(d_1,d_2,d_3)=(1,1,1)$.
Then, optimal memory-sharing performs more than $6$ times better than both other strategies.

\subsection{Numerical gap}\label{sec:numerics-gap}

As discussed in Section~\ref{sec:results}, the multiplicative gap in Theorem~\ref{thm:multi-user-gap} results from many generous approximations in bounding the achievable rate.
Numerical results suggest that this gap is in fact much smaller.
We here give a few examples of these results.
\begin{itemize}
\item
If $K=10$, $L=3$, $(N_1,N_2,N_3)=(500,1500,8000)$, $(U_1,U_2,U_3)=(9,5,1)$, and $(d_1,d_2,d_3)=(1,3,5)$, then we get a gap of approximately $6$.

\item
Increasing the number of caches and files: if $K=20$, $L=3$, $(N_1,N_2,N_3)=(200,20\,000,800\,000)$, $(U_1,U_2,U_3)=(10,5,1)$, and $(d_1,d_2,d_3)=(1,1,1)$, then we get a gap of approximately $6.8$.

\item
The same setup as the previous point, with access degrees of $(d_1,d_2,d_3)=(1,2,3)$ gives a gap of about $7.6$.
\end{itemize}

In most examples we have tried, the gap was in the range $5$-$10$, regardless of the access degree.
In the worst case, the largest gap our numerics have shown was about $45$, in a situation with $L=3$ and $D=5$.

\subsection{Relation of the single-user setup to the literature}
\label{sec:discussion-clustering}

As previously mentioned, the clustering scheme we proposed for the single-user setup is similar to strategies discussed in the literature for similar setups with stochastic demands \cite{ZhangArbitrary,Zcodedcaching}.
In these setups, users request files based on a probability distribution, and the average rate is analyzed (as opposed to the worst-case rate in our case).
In both cases, the proposed strategy is to divide the files into two sets, based on some threshold, and store only the most popular files.

Recall that our strategy clusters the levels $i$ whose popularity is such that $M\ge N_i/K_i$.
The popularity $p_i$ of a file of level $i$ is proportional to $K_i/N_i$, and hence if we normalize so that the sum of popularities is 1, we get:
\[
p_i = \frac{K_i/N_i}{\sum_{j=1}^L N_j\cdot K_j/N_j} = \frac{K_i}{KN_i}.
\]
Thus, the condition $M\ge N_i/K_i$ can be rewritten as $p_i\ge1/KM$.
This is exactly the threshold used in \cite{ZhangArbitrary} to determine which files to cluster and store in the caches, and which files to leave out of the caches.

%%fakesection{bibliography}
\bibliographystyle{IEEEtran}
\bibliography{caching}

\appendices
\section{Proofs for the multi-user setup}
\label{app:multi-user}

\subsection{Proof of Lemma~\ref{lemma:multi-user-converse}}
\label{app:multi-user-converse}

\subsubsection{A small example for illustration}

Before we give the full proof of Lemma~\ref{lemma:multi-user-converse}, we will start with a simple example for illustration.
Consider a multi-level, multi-user caching system with $K=6$ caches and $L=3$ levels.
Suppose that $U_1=U_2=U_3=1$, and let $N_1$, $N_2$, and $N_3$ be some large numbers (their exact value is not important for this example).
Finally, assume a single-access structure for all users, i.e., $d_1=d_2=d_3=1$.

As we have discussed in Section~\ref{sec:multi-user-converse}, the lower bounds on the optimal rate that we wish to obtain are a sum of $L$ cut-set bounds, each pertaining to a single level.
A cut-set bound for level $i$ consists of a certain number of caches and broadcast messages, such that the users at these caches can use the broadcast to cooperatively decode a set of files from level~$i$.
For example, if we consider the $sU_i$ users of level $i$ connected to some $s$ caches, we can send $b$ broadcast messages, tailored for the correct request vectors, so that the users can collectively decode $\min\{sU_i\cdot b,N_i\}$ files of level $i$.
Furthermore, the values of $s$ and $b$ are usually chosen so that the final bound matches the achieved individual rate of level $i$.
Thus, they depend on value of the rate and hence on the memory available to the level.

In our example, suppose that the memory we expect each level to receive dictates the following number of caches to consider: a single cache for level $1$; two caches for level $2$; and three caches for level $3$.
We define $s_1=1$, $s_2=2$, and $s_3=3$ to be these numbers.
In order to group together the cut-set bounds of all the levels with their different numbers of caches, we resort to the sliding-window subset entropy inequality presented in Lemma~\ref{lemma:sliding-window}.

In this example, we will consider six separate broadcast messages.
To avoid burying the essence of the argument under technicalities, we omit the request vectors from the discussion of the example.
For simplicity, we will assume that the messages are chosen so that, any time we encounter a group of $p$ users and $q$ messages, the users are able to decode a total of $pq$ messages.
This relaxation will be discarded when we discuss the general proof.

Let the caches be labeled with $Z_1,\ldots,Z_6$ and the broadcast messages with $X_1,\ldots,X_6$.
We start with the level with the smallest $s_i$, in this case level $1$.
Since $s_1=1$, this means its cut-set bound will consider a single cache.
In fact, the cut-set bound should have the following form by Fano's inequality:
\[
H(Z_1,X_1) \ge H(Z_1,X_1|\mathcal{W}^1) + 1\cdot F,
\]
where $\mathcal{W}^1$ represents the set of files of level $1$ that got decoded.
In this case, the single level-$1$ user at cache $Z_1$ can only decode one file when given just $X_1$, hence the $1\cdot F$ term.
However, to take full advantage of the sliding-window entropy inequality, we will take the average over six cut-set bounds for the same level, one for each cache:
\begin{IEEEeqnarray*}{rCl}
RF+MF
&\ge& \frac16\left[ H(Z_1,X_1)+\cdots+H(Z_6,X_6) \right]\\
&\ge& \frac16\left[ H(Z_1,X_1|\mathcal{W}^1)+\cdots+H(Z_6,X_6|\mathcal{W}^1) \right]\\
&&{}+ 1\cdot F.
\end{IEEEeqnarray*}

After completing the cut-set bound for level $1$, we now attempt to transition to the cut-set bound for level $2$.
Because we have taking the average of six instances of cut-set bounds, we can use Lemma~\ref{lemma:sliding-window} to obtain $6$ new cut-set bounds with $s_2=2$ caches each (for ease of notation, we write $Y_k=(Z_k,X_k)$):
\begin{IEEEeqnarray*}{rCl}
RF+MF
&\ge& \frac16\left[ H(Y_1|\mathcal{W}^1)+H(Y_2|\mathcal{W}^1)+H(Y_3|\mathcal{W}^1) \right.\\
&& \left.{} +H(Y_4|\mathcal{W}^1)+H(Y_5|\mathcal{W}^1)+H(Y_6|\mathcal{W}^1) \right]\\
&& {} + 1\cdot F\\
&\overset{(a)}{\ge}& \frac16\cdot\frac12\left[ H(Y_1,Y_2|\mathcal{W}^1)+H(Y_2,Y_3|\mathcal{W}^1) \right.\\
&& {} + H(Y_3,Y_4|\mathcal{W}^1) + H(Y_4,Y_5|\mathcal{W}^1)\\
&& \left. {} + H(Y_5,Y_6|\mathcal{W}^1) + H(Y_6,Y_1|\mathcal{W}^1) \right]\\
&& {} + 1\cdot F\\
&\overset{(b)}{\ge}& \frac16\cdot\frac12\left[ H(Y_1,Y_2|\mathcal{W}^1,\mathcal{W}^2)
+H(Y_2,Y_3|\mathcal{W}^1,\mathcal{W}^2) \right.\\
&& {} + H(Y_3,Y_4|\mathcal{W}^1,\mathcal{W}^2) + H(Y_4,Y_5|\mathcal{W}^1,\mathcal{W}^2)\\
&& \left. {} + H(Y_5,Y_6|\mathcal{W}^1,\mathcal{W}^2) + H(Y_6,Y_1|\mathcal{W}^1,\mathcal{W}^2) \right]\\
&& {} + 1\cdot F + 4\cdot F.
\end{IEEEeqnarray*}
Here, inequality $(a)$ uses Lemma~\ref{lemma:sliding-window}, while inequality $(b)$ uses Fano's inequality on the level-$2$ cut-set bounds.
Since each bound involves two users and two broadcast messages, the total number of files decoded is $4$, hence the $4\cdot F$ term.
The set of decoded level-$2$ files is denoted by $\mathcal{W}^2$.

We proceed again with the transition from level $2$ to level $3$.
Just like before, we first apply Lemma~\ref{lemma:sliding-window} to obtain entropy terms with the correct number of caches $s_3$, and then apply Fano's inequality to decode files from level $3$, labeled $\mathcal{W}^3$.
\begin{IEEEeqnarray*}{rCl}
RF+MF
&\ge& \frac16\cdot\frac12\left[ H(Y_1,Y_2|\mathcal{W}^1,\mathcal{W}^2)
+H(Y_2,Y_3|\mathcal{W}^1,\mathcal{W}^2) \right.\\
&& {} + H(Y_3,Y_4|\mathcal{W}^1,\mathcal{W}^2) + H(Y_4,Y_5|\mathcal{W}^1,\mathcal{W}^2)\\
&& \left. {} + H(Y_5,Y_6|\mathcal{W}^1,\mathcal{W}^2) + H(Y_6,Y_1|\mathcal{W}^1,\mathcal{W}^2) \right]\\
&& {} + 1\cdot F + 4\cdot F\\
&\overset{(a)}{\ge}& \frac16\cdot\frac13
\left[
H(Y_1,Y_2,Y_3|\mathcal{W}^1,\mathcal{W}^2) \right.\\
&&{} + H(Y_2,Y_3,Y_4|\mathcal{W}^1,\mathcal{W}^2)\\
&&{} + \cdots\\
&& \left. {} + H(Y_6,Y_1,Y_2|\mathcal{W}^1,\mathcal{W}^2)
\right]\\
&&{} + 1\cdot F + 4\cdot F\\
&\overset{(b)}{\ge}& \frac16\cdot\frac13
\left[
H(Y_1,Y_2,Y_3|\mathcal{W}^1,\mathcal{W}^2,\mathcal{W}^3) \right.\\
&&{} + H(Y_2,Y_3,Y_4|\mathcal{W}^1,\mathcal{W}^2,\mathcal{W}^3)\\
&&{} + \cdots\\
&& \left. {} + H(Y_6,Y_1,Y_2|\mathcal{W}^1,\mathcal{W}^2,\mathcal{W}^3)
\right]\\
&&{} + 1\cdot F + 4\cdot F + 9\cdot F.
\end{IEEEeqnarray*}
Again, inequality $(a)$ uses Lemma~\ref{lemma:sliding-window}, and inequality $(b)$ uses Fano's inequality on level $3$, which is considering groups of $3$ users and $3$ broadcast messages, allowing the decoding of a total of $9$ files.

Finally, by the non-negativity of entropy, the lower bound becomes:
\begin{IEEEeqnarray*}{rCl}
RF+MF &\ge& (1+4+9)\cdot F\\
R+M &\ge& 14.
\end{IEEEeqnarray*}

\subsubsection{The general proof}

The process involved in the general case is similar to the one shown in the example above.
We start with the level with the smallest $s_i$ and write an average of $K$ cut-set bounds for this level, one for each sequence of consecutive $s_i$ caches.
After applying Fano's inequality, we use Lemma~\ref{lemma:sliding-window} to transition to the level with the next smallest $s_i$.
In this discussion, the popularity of the levels is not important, but the order of their $s_i$'s is.
Specifically, if $s_i<s_j$, it does not matter which of $i$ and $j$ is more popular.
Thus, we can assume without loss of generality that $s_1\le\cdots\le s_L$.

In the example, we started with cut-set bounds with a single cache each.
For technical reasons, this will not give us good enough bounds in general, and so the initial cut-set bounds will consist of $t$ caches, for a general $t\in\{1,\ldots,K\}$.
In fact, every consecutive $t$ caches will be clustered into an inseperable group.
The group consisting of the $t$ caches that start with cache $k$ is labeled as:
\[
\mathcal{Z}^t_k = \left(
Z_k, \ldots, Z_{\langle k+t-1\rangle}
\right),
\]
where $\langle m\rangle$ is defined for integers $m$ as in Lemma~\ref{lemma:sliding-window}, i.e., $\langle m\rangle=m$ if $m\le K$ and $\langle m\rangle=m-K$ if $m>K$.
To every cache group $\mathcal{Z}^t_k$, we associate a broadcast-message group $\mathcal{X}^b_k$, which consists of $b$ messages serving different user demands.

Recall that a level-$i$ user needs to connect to $d_i$ \emph{consecutive} caches in order to decode whichever file he has requested.
Applying Lemma~\ref{lemma:sliding-window} should hence keep only consecutive caches in the same cut-set bounds, which will allow for the maximum number of users to be active in the decoding of the files (and thus produce a larger number of files from the same cut-set bound).
While this was fairly simple to ensure in the example above, we must show that we can still do it even after introducing the $t$-groups of caches.
This is done in the next paragraph.

Let $g$ be the GCD of $t$ and $K$.
Then, the following sequence:
\[
\left( \mathcal{Z}^t_1, \mathcal{Z}^t_{\langle t+1\rangle}, \ldots, \mathcal{Z}^t_{\langle (K/g)t+1\rangle} \right),
\]
starts at cache $1$ and ends at cache $K$.
Each cache $Z_k$ appears exactly $(t/g)$ times in the sequence.
Furthermore, every pair of caches that are consecutive in the sequence are also consecutive in the system.
For example, suppose $t=4$ and $K=6$.
Then their GCD would be $g=2$, with $K/g=3$, and the sequence would be:
\[
\Bigl(
(Z_1,Z_2,Z_3,Z_4),
(Z_5,Z_6,Z_1,Z_2),
(Z_3,Z_4,Z_5,Z_6)
\Bigr).
\]
Notice that every cache appears in the sequence $t/g=2$ times, and that consecutive caches remain so in the sequence.

We are now ready to prove the lower bounds.
Start with an average of $K/g$ cut-set bounds consisting of one group of $t$ caches (and their associated broadcast messages) each:
\[
bRF+tMF
\ge \frac{g}{K} \sum_{k=1}^{K/g} H\left( \mathcal{Z}^t_{(k-1)g+1}, \mathcal{X}^b_{(k-1)g+1} \right).
\]
This follows from $bR+tM\ge H(\mathcal{Z}^t_k,\mathcal{X}^b_k)$ for all $k$.
Notice how all the indices $[(k-1)g+1]$ are the same modulo $g$.
We can include all other caches to get:
\[
bRF+tMF
\ge \frac{1}{K} \sum_{k=1}^K H\left( \mathcal{Z}^t_k, \mathcal{X}^b_k \right).
\]

The first step is to use Lemma~\ref{lemma:sliding-window} to obtain cut-set bounds with $s_1t$ caches for level $1$.
\begin{IEEEeqnarray*}{rCll}
bRF+tMF
&\ge& \IEEEeqnarraymulticol{2}{l}{
  \frac{1}{K} \sum_{k=1}^1 H\left( \mathcal{Z}^t_k, \mathcal{X}^b_k \right)
}\\
&\ge& \frac{1}{K} \cdot \frac{1}{s_1}
\sum_{k=1}^K & H\left( \mathcal{Z}^t_k,\mathcal{Z}^t_{\langle k+t\rangle},\ldots,\mathcal{Z}^t_{\langle k+(s_1-1)t\rangle}, \right.\\
&& &\qquad\left. \mathcal{X}^b_k,\mathcal{X}^b_{\langle k+t\rangle},\ldots,\mathcal{X}^b_{\langle k+(s_1-1)t\rangle} \right)
\end{IEEEeqnarray*}
Each entropy term in the sum now consists of $s_1t$ \emph{consecutive} caches and $s_1b$ broadcast messages.
For simplicity, we write it as $\hat H_k(s_1t,s_1b)$, for $k\in\{1,\ldots,K\}$, and its conditional version as $\hat H_k(s_1t,s_1b|Q)$ for any random variable $Q$.
Let $\mathcal{W}^1$ denote the set of level-$1$ files that can be decoded in every cut-set bound in the sum, and let $p_1$ be its size, to be determined later.
For technical reasons, we will only apply Fano's inequality on a fraction of the cut-set bounds; the rest are simply lower-bounded by conditioning.
The meaning and value of this fraction is discussed below, but for now we will call it $\lambda_1$.
Then, we use Fano's inequality to get:
\begin{IEEEeqnarray*}{rCl}
bRF+tMF
&\ge& \frac{1}{K}\cdot\frac{1}{s_1}\sum_{k=1}^K \hat H_k(s_1t,s_1b)\\
&\ge& \frac{1}{K}\cdot\frac{1}{s_1}\sum_{k=1}^K \hat H_k(s_1t,s_1b|\mathcal{W}^1) + \frac{\lambda_1p_1}{s_1}\cdot F.
\end{IEEEeqnarray*}
By applying Lemma~\ref{lemma:sliding-window} again, we obtain cut-set bounds pertaining to level $2$.
We lower-bound a fraction $\lambda_2$ of them in turn using Fano's inequality, and repeat for all levels.
Thus, we have:
\begin{IEEEeqnarray*}{rCl}
bRF+tMF
&\ge& \frac{1}{K}\cdot\frac{1}{s_1}\sum_{k=1}^K \hat H_k(s_1t,s_1b|\mathcal{W}^1) + \frac{\lambda_1p_1}{s_1}\cdot F\\
&\overset{(a)}{\ge}& \frac{1}{K} \cdot \frac{1}{s_2} \sum_{k=1}^K \hat H_k(s_2t,s_2b|\mathcal{W}^1) + \frac{\lambda_1p_1}{s_1}\cdot F\\
&\overset{(b)}{\ge}& \frac{1}{K} \cdot \frac{1}{s_2} \sum_{k=1}^K \hat H_k(s_2t,s_2b|\mathcal{W}^1,\mathcal{W}^2)\\
&& {} + \frac{\lambda_1p_1}{s_1}\cdot F + \frac{\lambda_2p_2}{s_2}\cdot F\\
&\ge& \cdots\\
&\ge& \sum_{i=1}^L \frac{\lambda_ip_i}{s_i}\cdot F.
\IEEEyesnumber\label{eq:mu-lower-bounds-steps}
\end{IEEEeqnarray*}
where the inequality marked with $(a)$ uses Lemma~\ref{lemma:sliding-window} and the one marked with $(b)$ uses Fano's inequality.

For any $m<K$ \emph{consecutive} caches, the number of level-$i$ users that are connected to $d_i$ of those caches is exactly $(m-d_i+1)U_i$.
Therefore, given $s_it<K$ consecutive caches and $s_ib$ broadcast messages, the users at these caches should, in principle, decode up to $p_i=\min\left\{ (s_it-d_i+1)U_i\cdot s_ib, N_i \right\}$ files from level $i$.
This requires choosing the broadcast messages for the correct user demands.
However, consider a pair of broadcast message $X'$ and user $u$ that appears in multiple cut-set bounds.
In other words, the user $u$ is expected to decode a file using the message $X'$ on multiple occasions.
We must ensure that there are no contradictions, i.e., that the message $X'$ always delivers the same file to user $u$.

It turns out it is not always possible to do that and still be able to decode the maximal number of files at every cut-set bound.
However, because $s_it\le K/2$ (as constrained in the statement of Lemma~\ref{lemma:multi-user-converse}), we can get around this issue by ``ignoring'' half of the cut-set bounds.
To understand this, suppose we wish to determine the user demand vector of each broadcast message to be sent.
Starting with $\hat H_1(s_it,s_ib)$: there are a total of $(s_it-d_i+1)U_i$ relevant users, and $s_ib$ broadcasts.
We can design these broadcasts to allow these users to decode all $p_i$ files.
Next, we move on to $\hat H_2(s_it,s_ib)$.
In this entropy term, exactly $tU_i$ users from the previous term are replaced by $tU_i$ \emph{brand new} users; the other users are still the same.
Thus, if the broadcast messages serve the same files to the new users as they did to the old, the same total number of files can be decoded.
This can go on as long as every step introduces new $tU_i$ users, which is true for all but the last $(s_it-d_i)$ steps.
Indeed, the term $\hat H_{K-(s_it-d_i)+1}(s_it,s_ib)$ re-introduces users that previously appeared in $\hat H_1(s_it,s_ib)$.
Thus, for every level $i$, only a fraction $\frac{K-(s_it-d_i)}{K}$ of the cut-set bounds can decode all $p_i$ files using Fano's inequality.
We bound this fraction by:
\[
\frac{K-(s_it-d_i)}{K}
\ge \begin{cases}
\frac12 & \text{if $s_it>d_i$};\\
1 & \text{if $s_it=d_i$},
\end{cases}
\]
and we define $\lambda_i$ to be the right-hand side of the inequality.

By substituing $p_i$ for its value in \eqref{eq:mu-lower-bounds-steps}, we get the final bounds:
\begin{IEEEeqnarray*}{rCl}
bRF+tMF
&\ge& \sum_{i=1}^L \frac{\lambda_i}{s_i} \cdot \min\left\{ (s_it-d_i+1)U_i \cdot s_ib, N_i \right\} \cdot F\\
R &\ge& \sum_{i=1}^L \lambda_i \cdot \min\left\{ (s_it-d_i+1)U_i, \frac{N_i}{s_ib} \right\} - \frac{t}{b}M,
\end{IEEEeqnarray*}
which concludes the proof of Lemma~\ref{lemma:multi-user-converse}.

\subsection{Proof of approximate optimality (Theorem~\ref{thm:multi-user-gap})}
\label{app:multi-user-gap}

In this section, we will use the information-theoretic lower bounds determined in Lemma~\ref{lemma:multi-user-converse} to give an upper bound on the ratio between the rate achieved by the memory-sharing scheme and the optimal rate (the ``gap'').
In order to do that, we must consider a few different cases, in each of which different values are chosen for the parameters $b$, $t$, and $\{s_i\}_i$ defined in Lemma~\ref{lemma:multi-user-converse}.
When these parameters are chosen, we will need to evaluate the minimization seen in the expression of the lower bound, for every level.
For simplicity, we define, for each level $i$:
\begin{equation}
\label{eq:ai}
A_i = \min\left\{ (s_it-d_i+1)U_i \,,\, \frac{N_i}{s_ib} \right\}.
\end{equation}
To find or bound the value of $A_i$, we must evaluate the following comparision:
\begin{equation}
\label{eq:comparison}
bs_i(s_it-d_i+1) \overset{?}{\lessgtr} \frac{N_i}{U_i}.
\end{equation}

Before the main analysis, we consider the case where the number of caches is bounded.
Specifically, we consider $K < k_0 = D/\beta$, and we call this ``Case 0''.
Then, we consider the more interesting case where $K$ is unbounded, and divide that into two main regimes.
The first, ``Case 1'', is when the set $I_1$, defined in Definition~\ref{def:m-feasible} and \eqref{eq:refined-partition}, is empty; the second, ``Case 2'', is when it is not empty.

For convenience, we will assume, without loss of generality, that the levels are numbered from most popular to least popular.
In other words:
\begin{equation}
\label{eq:level-ordering}
U_1/N_1 \ge \cdots \ge U_L/N_L.
\end{equation}

\subsubsection{Case 0: $K < k_0 = D/\beta$}

Recall the refined $M$-feasible partition described in Definition~\ref{def:m-feasible} and \eqref{eq:refined-partition}.
By \eqref{eq:alpha-i}, we know that the memory given to a level $i_0\in I_0$ is at most:
\[
\alpha_{i_0}M < (2/K)\sqrt{N_{i_0}/U_{i_0}}.
\]
However, if $K < D/\beta$, then, for any level $i\in I'$:
\begin{IEEEeqnarray*}{rCl}
\alpha_iM
&\le& \left( \frac1K + \frac{\beta}{d_i} \right) \sqrt{\frac{N_i}{U_i}}\\
&=& \frac1K \left( 1 + \frac{\beta K}{d_i} \right) \sqrt{\frac{N_i}{U_i}}\\
&<& \frac1K \left( 1 + \frac{D}{d_i} \right) \sqrt{\frac{N_i}{U_i}}\\
&\le& \frac{D+1}{K} \sqrt{\frac{N_i}{U_i}}\\
&=& \frac{D+1}{2} \cdot \frac{2}{K}\sqrt{\frac{N_i}{U_i}}.
\end{IEEEeqnarray*}
This suggests that the memory given to a level in $I'$ will not be much larger than that given to a level in $I_0$.
As a result, levels in $I'$ are expected to behave similarly to those in $I_0$: they receive so little memory that their impact on the overall transmission rate is not much different from that of the levels in $H$, which get zero memory.
The effective result is that the subset $I_1$ ``dominates'' the set $I$.

With these observations in mind, we next describe a near-equivalent formulation of the achievability scheme, which is more suitable for the case $K<D/\beta$.
This formulation will emphasize the fact that $I_1$ is the dominant subset of $I$, by essentially reducing $I$ to $I_1$ and relegating $I_0$ and $I'$ to $H$.

A more precise explanation follows.
Find the (unique) level $i^\ast$ such that:
\[
\sum_{i=1}^{i^\ast-1} \frac{N_i}{d_i}
\le M
\le \sum_{i=1}^{i^\ast} \frac{N_i}{d_i}.
\]
Recall that the levels are numbered from most popular to least popular, as seen in \eqref{eq:level-ordering}.

Partition the set of levels into three sets $(H,I,J)$, which will serve the same purpose as the partition described in Definition~\ref{def:m-feasible}.
We set $H=\{i^\ast+1,\ldots,L\}$, $I=\{i^\ast\}$, and $J=\{1,\ldots,i^\ast-1\}$.
We then proceed as usual: the levels in $J$ are fully stored in the caches; the levels in $H$ are given no memory; and the single level in $I$, $i^\ast$, is given the remaining memory $M-T_J$.
The conventional scheme is applied to each level and its corresponding memory, resulting in the following transmission rate:
\begin{IEEEeqnarray*}{rCl}
R
&=& \sum_{h\in H} KU_h
+ KU_{i^\ast} \left( 1 - \frac{M-T_J}{N_{i^\ast}/d_{i^\ast}} \right)\\
&\le& k_0 \left[
\sum_{h\in H} U_h
+ U_{i^\ast} \left( 1 - \frac{M-T_J}{N_{i^\ast}/d_{i^\ast}} \right)
\right]. \IEEEyesnumber \label{eq:mu0-ach}
\end{IEEEeqnarray*}

For the lower bounds, consider Lemma~\ref{lemma:multi-user-converse}, with $t=1$, $s_i=d_i$ for all $i$, and $b=\ceil{N_{i^\ast}/d_{i^\ast}U_{i^\ast}}$.
Thus $\lambda_i=1$ for all $i$.
We will first analyze the comparisons in \eqref{eq:comparison} for every level.

For $i^\ast$, we have:
\[
bs_{i^\ast}(s_{i^\ast}t-d_{i^\ast}+1)
= bd_{i^\ast}
\ge \frac{N_{i^\ast}}{U_{i^\ast}}.
\]

For $j\in J$, we have:
\[
bs_j(s_jt-d_j+1)
= bd_j
\ge \frac{N_{i^\ast}}{U_{i^\ast}} \cdot \frac{d_j}{d_{i^\ast}}
\ge \frac{N_{i^\ast}}{U_{i^\ast}D}
\ge \frac{N_j}{U_j},
\]
by regularity condition \eqref{eq:mu-reg2}.

Finally, for $h\in H$, we have:
\[
bs_h(s_h-d_h+1)
= bd_h
\le 2\frac{N_{i^\ast}}{U_{i^\ast}} \cdot \frac{d_h}{d_{i^\ast}}
\le \frac{N_h}{U_h},
\]
again by the regularity condition \eqref{eq:mu-reg2}.

Putting these together, we determine the value of $A_i$, defined in \eqref{eq:ai}, for each $i$ and get the following lower bound on the optimal rate:
\begin{IEEEeqnarray*}{rCl}
R^\ast(M)
&\ge& \sum_{h\in H} U_h
+ \frac{N_{i^\ast}}{d_{i^\ast}b}
+ \sum_{j\in J} \frac{N_j}{d_jb}
- \frac{M}{b}\\
&=& \sum_{h\in H} U_h
+ \frac{N_{i^\ast}/d_{i^\ast} - (M-T_J)}{b}\\
&\ge& \sum_{h\in H} U_h
+ \frac{N_{i^\ast}/d_{i^\ast}
- (M-T_J)}{2N_{i^\ast}/d_{i^\ast}U_{i^\ast}}\\
&\ge& \frac12\left[
  \sum_{h\in H} U_h
+ U_{i^\ast} \left(
1 - \frac{M-T_J}{N_{i^\ast}/d_{i^\ast}}
\right)
\right]. \IEEEyesnumber \label{eq:mu0-conv}
\end{IEEEeqnarray*}

By combining \eqref{eq:mu0-ach} with \eqref{eq:mu0-conv}, we get the following bound on the gap:
\begin{equation}
\label{eq:mu0-gap}
\frac{R(M)}{R^\ast(M)} \le 2k_0 = 2D/\beta = 396 D.
\end{equation}

\subsubsection{Case 1: $I_1=\emptyset$}

When $I_1=\emptyset$, the rate achieved by the memory-sharing scheme is, according to Lemma~\ref{lemma:multi-user-achievability}:
\begin{equation}
\label{eq:mu1-ach}
R(M)
\le \sum_{h\in H} KU_h
+ \frac{2S_I^2}{M-T_J+V_I}.
\end{equation}

For technical reasons, the lower bounds analysis for this case has to be broken down into two subcases, depending on whether or not the set $J$ is empty.

Since we will be dealing with many floors and ceilings, here are a few remarks on these operations.
\begin{itemize}
\item
If $n\ge1$ is an integer, then $x\ge n \implies \floor{x}\ge n$;
\item
For all $x>0$, $\floor{x}\ge x-1$ and $\ceil{x}\le x+1$;
\item
If $x\ge1$, then $\floor{x}\ge x/2$ and $\ceil{x}\le 2x$;
\item
If $a\le x\le b$, then $\floor{a}\le\floor{x}\le\floor{b}$ and $\ceil{a}\le\ceil{x}\le\ceil{b}$.
\end{itemize}

\paragraph{Case 1a: $J\not=\emptyset$}

Recall that $\tilde M=(M-T_J+V_I)/S_I$.
Consider Lemma~\ref{lemma:multi-user-converse}, with the following parameters:
\begin{IEEEeqnarray*}{lCl"rCl}
&& & t &=& 1;\\
\forall h&\in&H, &
  s_h &=& \floor{\frac18 K};\\
\forall i&\in&I, &
  s_i &=& \floor{\frac{\sqrt{N_i/U_i}}{8\tilde M}};\\
\forall j&\in&J, &
  s_j &=& d_j;\\
&& & b &=& \floor{\delta \tilde M^2},
\end{IEEEeqnarray*}
where $\delta = D/\beta$.
Notice that $\lambda_j=1$ for all $j\in J$, and $\lambda_i\ge\frac12$ for all other levels $i$.

The first step is to verify that these parameters satisfy all of the constraints imposed on them in Lemma~\ref{lemma:multi-user-converse}.
For ease of reference, we repeat the constraints here:
\begin{itemize}
\item $t\in\{1,\ldots,K\}$;
\item $s_i\in\mathbb{N}^+$ such that $s_it\in\{d_i,\ldots,\floor{K/2}\}$, for any level~$i$;
\item $b\in\mathbb{N}^+$.
\end{itemize}
The parameters $t$ and $s_j$, $j\in J$, can be easily seen to satisfy theirs.
Also, it follows from Definition~\ref{def:m-feasible} and from regularity condition \eqref{eq:mu-reg1} that, for some $j\in J$:
\begin{IEEEeqnarray*}{rCl}
\delta \tilde M^2
&\ge&
  \delta
  \cdot \left( \frac{1}{d_j}+\frac1K \right)^2
  \frac{N_j}{U_j}
\ge
  D/\beta 
  \cdot \left( \frac{1}{d_j}+\frac1K \right)^2
  \cdot K\\
&\ge&
  k_0
  \cdot \left( \frac{1}{D} \right)^2
  \cdot k_0
= \left( \frac{k_0}{D} \right)^2
\ge 1,
\end{IEEEeqnarray*}
and therefore $b\ge1$.

As for $s_h$, $h\in H$: we have $s_h\le \floor{K/2}$ trivially, and $s_h\ge d_h$ is true because $K\ge k_0=D/\beta\ge 8d_h$.

Finally, consider $i\in I$.
First, recall that $I_1=\emptyset$, and therefore $i\notin I_1$.
Then, using Definition~\ref{def:m-feasible}, we have:
\begin{IEEEeqnarray*}{rCl}
\frac{\sqrt{N_i/U_i}}{8\tilde M}
&\ge& \frac18\cdot\frac{1}{\frac{\beta}{d_i}+\frac1K}
= \frac{1}{8(\beta+d_i/K)}\cdot d_i\\
&\ge& \frac{1}{16\beta}\cdot d_i
\ge d_i,
\end{IEEEeqnarray*}
which implies $s_i\ge d_i$.
Also, we have:
\begin{IEEEeqnarray*}{rCl}
s_i
&=& \floor{ \frac{\sqrt{N_i/U_i}}{8\tilde M} }
\le \floor{ \frac18\cdot\frac{1}{\frac1K} }
\le \floor{K/2}.
\end{IEEEeqnarray*}
Thus all the parameters satisfy their constraints.

Next, we will compute bounds on the $A_i$ terms by evaluating the comparison \eqref{eq:comparison} for all levels $i$.

For $h\in H$, the comparison \eqref{eq:comparison} gives:
\[
bs_h(s_ht-dh+1)
\le bs_h^2t
\le \delta\tilde M^2 \cdot \frac{1}{64}K^2
\le \frac{\delta}{64}\cdot\frac{N_h}{U_h}.
\]
Thus:
\begin{IEEEeqnarray*}{rCl}
A_h
&\ge& \min\left\{ 1 , \frac{64}{\delta} \right\}
  \cdot (s_ht-d_h+1)U_h\\
&=& \min\left\{ 1 , \frac{64\beta}{D} \right\}
  \cdot (s_ht-d_h+1)U_h\\
&\ge& \frac{64\beta}{D}
  \cdot \left(\frac18 K - d_h\right) U_h\\
&\ge& \frac{64\beta}{D}
  \cdot \left( \frac18 - \beta \right)
  \cdot KU_h\\
&\ge& \frac{8\beta}{D}
  \cdot \left( 1 - 8\beta \right)
  \cdot KU_h.
\end{IEEEeqnarray*}

Consider now $i\in I$.
The comparison \eqref{eq:comparison} gives:
\[
bs_i(s_it-d_i+1)
\le bs_i^2t
\le \delta\tilde M^2 \cdot \frac{1}{64}\frac{N_i/U_i}{\tilde M^2}
= \frac{\delta}{64}\cdot\frac{N_i}{U_i}.
\]
Therefore,
\begin{IEEEeqnarray*}{rCl}
A_i
&\ge& \min\left\{ 1 , \frac{64}{\delta} \right\}
  \cdot (s_it-d_i+1)U_i\\
&=& \frac{64\beta}{D}
  \cdot (s_it-d_i+1)U_i\\
&\ge& \frac{64\beta}{D}
  \cdot \left( \frac{\sqrt{N_i/U_i}}{8\tilde M} - d_i \right) U_i\\
&=& \frac{64\beta}{D}
  \cdot \left( \frac18
    - d_i\cdot\frac{\tilde M}{\sqrt{N_i/U_i}} \right)
  \cdot \frac{\sqrt{N_iU_i}}{\tilde M}\\
&\ge& \frac{64\beta}{D}
  \cdot \left( \frac18
    - \left( \beta+\frac{d_i}{K} \right) \right)
  \cdot \frac{\sqrt{N_iU_i}}{\tilde M}\\
&\ge& \frac{64\beta}{D}
  \cdot \left( \frac18 - 2\beta \right)
  \cdot \frac{\sqrt{N_iU_i}}{\tilde M}\\
&\ge& \frac{8\beta}{D}
  \cdot \left( 1 - 16\beta \right)
  \cdot \frac{\sqrt{N_iU_i}}{\tilde M}.
\end{IEEEeqnarray*}

Finally, let us look at $j\in J$.
The comparison in \eqref{eq:comparison} gives:
\begin{IEEEeqnarray*}{rCl}
bs_j(s_jt-d_j+1)
&=& bd_j\\
&\ge& \frac12 \delta\tilde M^2 \cdot d_j\\
&\ge& \frac12 (D/\beta) \left( \frac{1}{d_j}+\frac1K \right)^2\frac{N_j}{U_j} d_j\\
&\ge& \frac{1}{2\beta}d_j^2\cdot \left( \frac{1}{d_j} \right)^2 \cdot \frac{N_j}{U_j}\\
&\ge& \frac{N_j}{U_j}.
\end{IEEEeqnarray*}
Therefore:
\[
A_j = \frac{N_j}{d_jb}.
\]

We can now use the values of $\{A_i\}_i$ to lower-bound the optimal rate as in Lemma~\ref{lemma:multi-user-converse}.
Recall also the values of $\lambda_i$ mentioned at the beginning of this section.
\begin{IEEEeqnarray*}{rCl}
R^\ast(M)
&\ge&
  \sum_{h\in H}
    \frac12 \cdot \frac{8\beta}{D}
	\cdot (1-8\beta)
	\cdot KU_h\\
&&{} +
  \sum_{i\in I}
    \frac12 \cdot \frac{8\beta}{D}
	\cdot (1-16\beta)
	\cdot \frac{\sqrt{N_iU_i}}{\tilde M}\\
&&{} +
  \sum_{j\in J}
    \frac{N_j}{d_jb}
  - \frac{M}{b}\\
&=&
  \sum_{h\in H}
    \frac{4\beta}{D}
	\cdot (1-8\beta)
	\cdot KU_h\\
&&{} +
    \frac{4\beta}{D}
	\cdot (1-16\beta)
	\cdot \frac{S_I\cdot\sum_{i\in I}\sqrt{N_iU_i}}{M-T_J+V_I}\\
&&{} -
  \frac{M-T_J}{b}\\
&\ge&
  \sum_{h\in H}
    \frac{4\beta}{D}
	\cdot (1-8\beta)
	\cdot KU_h\\
&&{} +
    \frac{4\beta}{D}
	\cdot (1-16\beta)
	\cdot \frac{S_I^2}{M-T_J+V_I}\\
&&{} -
  \frac{M-T_J+V_I}{\frac12\delta(M-T_J+V_I)^2/S_I^2}\\
&=&
  \sum_{h\in H}
    \frac{4\beta}{D}
	\cdot (1-8\beta)
	\cdot KU_h\\
&&{} +
    \left[
	  \frac{4\beta}{D}
	  \cdot (1-16\beta)
	  - \frac{2}{\delta}
	\right]
    \cdot \frac{S_I^2}{M-T_J+V_I}\\
&=&
  \sum_{h\in H}
    \frac{\beta}{D}
	\cdot (4-32\beta)
	\cdot KU_h\\
&&{} +
    \frac{\beta}{D}
	\left(1-32\beta\right)
    \cdot \frac{2S_I^2}{M-T_J+V_I}.
\IEEEyesnumber\label{eq:mu1a-conv}
\end{IEEEeqnarray*}

Combining \eqref{eq:mu1a-conv} with \eqref{eq:mu1-ach}, we get:
\begin{equation}
\label{eq:mu1a-gap}
\frac{R(M)}{R^\ast(M)}
\le
\frac{D/\beta}{1-32\beta} \le 237 D.
\end{equation}

\paragraph{Case 1b: $J=\emptyset$}

We know that $I$ is never an empty set.
Since $J=\emptyset$, then the most popular level is in $I$.
By \eqref{eq:level-ordering}, that level is level $1$, and hence $1\in I$.
Moreover, we note that $T_J=0$ when $J=\emptyset$, and thus $\tilde M=(M+V_I)/S_I$.

Consider again Lemma~\ref{lemma:multi-user-converse}, and define the following parameters:
\begin{IEEEeqnarray*}{lCl"rCl}
&& & t &=& \floor{\frac{1}{32}\frac{\sqrt{N_1/U_1}}{\tilde M}};\\
\forall h&\in&H, &
  s_h &=& \floor{2\frac{K\tilde M}{\sqrt{N_1/U_1}}};\\
\forall i&\in&I, &
  s_i &=& \floor{2\sqrt{\frac{N_i/U_i}{N_1/U_1}}};\\
&& & b &=& \floor{8\tilde M\sqrt{N_1/U_1}}.
\end{IEEEeqnarray*}
Thus, we can only say that $\lambda_i\ge\frac12$ for all levels $i$.

As before, we must first verify that these parameters satisfy all their constraints.
Let us start with $b$:
\[
8\tilde M\sqrt{N_1/U_1}
\ge 8 \cdot \frac1K\sqrt{N_1/U_1}\cdot\sqrt{N_1/U_1}
=   \frac{8N_1}{KU_1}
\ge 1,
\]
by regularity condition \eqref{eq:mu-reg1}, and therefore $b\ge1$.

Next, we will examine $t$ and the $\{s_i\}_i$.
It suffices to show $t\ge1$, $s_i\ge1$, $s_it\ge d_i$, and $s_it\le\floor{K/2}$ for all $i$.
We can see that $t\ge1$ because:
\[
\frac{\sqrt{N_1/U_1}}{32\tilde M}
\ge \frac{1/32}{\frac{\beta}{d_1}+\frac1K}
\ge \frac{1}{64\beta}\cdot d_1
\ge 1,
\]
where the first inequality follows from $1\in I\setminus I_1$ (recall $I_1=\emptyset$) and Definition~\ref{def:m-feasible} combined with \eqref{eq:refined-partition}.

For $h\in H$, we have $2K\tilde M/\sqrt{N_1/U_1}\ge1$ directly from Definition~\ref{def:m-feasible}, and thus $s_h\ge1$.
Moreover,
\[
s_ht
\ge
  \frac{1}{64} \cdot \frac{\sqrt{N_1/U_1}}{\tilde M}
  \cdot \frac{K\tilde M}{\sqrt{N_1/U_1}}
= \frac{K}{64} \ge D \ge d_h,
\]
and,
\[
s_ht
\le \frac{1}{32}\cdot\frac{\sqrt{N_1/U_1}}{\tilde M}
  \cdot 2\frac{K\tilde M}{\sqrt{N_1/U_1}}
= K/16 \le \floor{K/2}.
\]

Consider now $i\in I$.
Because level $1$ is the most popular level, then $N_i/U_i \ge N_1/U_1$ for all $i$, and hence $s_i\ge1$.
Moreover:
\[
s_it
\ge \frac{\sqrt{N_i/U_i}}{64\tilde M}
\ge \frac{1/64}{\frac{\beta}{d_i} + \frac1K}
\ge \frac{1}{128\beta} d_i
\ge d_i.
\]
Finally,
\[
s_it
\le \frac{\sqrt{N_i/U_i}}{16\tilde M}
\le K/16 \le \floor{K/2}.
\]
Thus, all the parameters satisfy their conditions.

We now move to the comparisons in \eqref{eq:comparison} which allow us to give bounds on $A_i$ for all $i$.

Consider $h\in H$.
Note that:
\begin{IEEEeqnarray*}{rCl}
bs_h(s_ht-d_h+1)
&\le& bs_h^2t\\
&\le& 8\tilde M\sqrt{N_1/U_1}
  \cdot \left[ \frac{2K\tilde M}{\sqrt{N_1/U_1}} \right]^2
  \cdot \frac{\sqrt{N_1/U_1}}{32\tilde M}\\
&=& \left[ K\tilde M \right]^2\\
&\le& \frac{N_h}{U_h}.
\end{IEEEeqnarray*}
Therefore,
\begin{IEEEeqnarray*}{rCl}
A_h
&=& (s_ht-d_h+1)U_h\\
&\ge& \left[
  \left( \frac{K\tilde M}{\sqrt{N_1/U_1}} \right)
  \left( \frac{\sqrt{N_1/U_1}}{32\tilde M} - 1 \right)
  - (d_h-1)
\right] U_h\\
&\ge& \left[
  \left( \frac{1}{32} - \frac{\tilde M}{\sqrt{N_1/U_1}} \right)
  - \frac{D-1}{k_0}
\right] KU_h\\
&\overset{(a)}{\ge}& \left[ 
  \frac{1}{32} - \left( \frac{\beta}{d_1} + \frac{1}{K} \right) - \frac{D}{k_0} - \frac{1}{k_0}
\right] KU_h\\
&\overset{(b)}{\ge}& \left(
  \frac{1}{32} - 2\beta
\right) \cdot KU_h.
\end{IEEEeqnarray*}
Here, $(a)$ uses Definition~\ref{def:m-feasible} with \eqref{eq:refined-partition} and $1\in I\setminus I_1$, and $(b)$ uses $d_1\ge1$ and $D/k_0=\beta$.

Now consider $i\in I$.
We have:
\begin{IEEEeqnarray*}{rCl}
bs_i(s_it-d_i+1)
&\le& bs_i^2t\\
&\le& 8\tilde M\sqrt{N_1/U_1}
  \cdot 4\frac{N_i/U_i}{N_1/U_1}
  \cdot \frac{\sqrt{N_1/U_1}}{32\tilde M}\\
&=& \frac{N_i}{U_i}
\end{IEEEeqnarray*}

Therefore,
\begin{IEEEeqnarray*}{rCl}
A_i
&=& (s_it-d_i+1)U_i\\
&\ge& \left[ 
  \sqrt{\frac{N_i/U_i}{N_1/U_1}}
  \cdot \left( \frac{\sqrt{N_1/U_1}}{32\tilde M} - 1 \right)
  - (d_i-1)
\right] U_i\\
&=& \left[ 
  \frac{1}{32}
  - \frac{\tilde M}{\sqrt{N_1/U_1}}
  - (d_i-1)\cdot\frac{\tilde M}{\sqrt{N_i/U_i}}
\right] \frac{\sqrt{N_iU_i}}{\tilde M}\\
&\ge& \left[ 
  \frac{1}{32}
  - \left( \frac{\beta}{d_1}+\frac1K \right)
  - (d_i-1)\cdot\left( \frac{\beta}{d_i}+\frac1K \right)
\right] \frac{\sqrt{N_iU_i}}{\tilde M}\\
&=& \left[
  \frac{1}{32}
  - 2\beta - \left( \frac{1}{d_1}-\frac{1}{d_i} \right)\beta
\right] \frac{\sqrt{N_iU_i}}{\tilde M}\\
&\ge& \left(
  \frac{1}{32} - 3\beta
\right) \frac{\sqrt{N_iU_i}}{\tilde M}.
\end{IEEEeqnarray*}

Finally, we will also give an upper bound to the following useful quantity:
\begin{IEEEeqnarray*}{rCl}
\frac{tM}{b}
&\le&
    \frac{\sqrt{N_1/U_1}}{32\tilde M}
  \cdot \frac{2}{8\tilde M\sqrt{N_1/U_1}}
  \cdot M\\
&\le& \frac{S_I^2}{128(M+V_I)^2} \cdot (M+V_I)\\
&=& \frac{S_I^2}{128(M+V_I)}.
\end{IEEEeqnarray*}

Combining the $A_i$ terms, along with the $\lambda_i$ factors, we get the following lower bound:
\begin{IEEEeqnarray*}{rCl}
R^\ast(M)
&\ge&
  \sum_{h\in H}
    \frac12 \cdot \left( \frac{1}{32} - 2\beta \right)
	\cdot KU_h\\
&&{} +
  \sum_{i\in I}
    \frac12 \cdot \left( \frac{1}{32} - 3\beta \right)
	\cdot \frac{\sqrt{N_iU_i}}{\tilde M}\\
&&{} -
  \frac{tM}{b}\\
&\ge&
  \sum_{h\in H}
    \frac12 \cdot \left( \frac{1}{32} - 2\beta \right)
	\cdot KU_h\\
&&{} +
    \frac12 \cdot \left( \frac{1}{32} - 3\beta \right)
	\cdot \frac{S_I^2}{M+V_I}\\
&&{} - \frac{S_I^2}{128(M+V_I)}\\
&=&
  \sum_{h\in H}
    \frac12 \cdot \left( \frac{1}{32} - 2\beta \right)
	\cdot KU_h\\
&&{} +
    \left[
	  \frac12 \cdot \left( \frac{1}{32} - 3\beta \right)
	  - \frac{1}{128}
	\right]
	\cdot \frac{S_I^2}{M+V_I}.\\
&=&
  \sum_{h\in H}
    \frac{1}{64} (1-64\beta)
	\cdot KU_h\\
&&{} +
    \frac{1}{256}(1-192\beta)
	\cdot \frac{2S_I^2}{M+V_I}.
\IEEEyesnumber\label{eq:mu1b-conv}
\end{IEEEeqnarray*}

Combining \eqref{eq:mu1b-conv} with \eqref{eq:mu1-ach}:
\begin{IEEEeqnarray*}{rCl}
\frac{R(M)}{R^\ast(M)}
&\le&
\frac{256}{1-192\beta}
= 8448.
\IEEEyesnumber\label{eq:mu1b-gap}
\end{IEEEeqnarray*}

\subsubsection{Case 2: $I_1\not=\emptyset$}

In this section, there exists at least one level in $I_1$.
However, by Proposition~\ref{prop:i1-size}, the set $I_1$ cannot contain more than one level.
Therefore, in this section we have $I_1=\{i_1\}$ for some level $i_1$.

Using Lemma~\ref{lemma:multi-user-achievability}, we can write the achievable rate as follows:
\begin{IEEEeqnarray*}{rCl}
R(M)
&\le&
  \sum_{h\in H} KU_h
  + \frac{2S_IS_{I\setminus I_1}}{M-T_J+V_I}\\
&&{} +
  \frac1\beta d_{i_1}U_{i_1}
    \left( 1-\frac{M-T_J}{N_{i_1}/d_{i_1}} \right)\\
&&{} +
  \frac1\beta d_{i_1}U_{i_1} \cdot
	\frac{S_{I\setminus I_1}}{\sqrt{N_{i_1}U_{i_1}}}.
\IEEEyesnumber\label{eq:mu2-ach}
\end{IEEEeqnarray*}

Let $\gamma_1=2.965$ and $\gamma_2=0.482$.
Consider Lemma~\ref{lemma:multi-user-converse} with the following parameters.
\begin{IEEEeqnarray*}{lCl"rCl}
&& & t &=& 1;\\
\forall h&\in&H, &
  s_h &=& \floor{2\beta K};\\
\forall i&\in&I\setminus I_1, &
  s_i &=& \floor{\gamma_1\beta \frac{\sqrt{N_i/U_i}}{\tilde M}
    + \gamma_2 d_{i_1}\sqrt{\frac{U_{i_1}/N_{i_1}}{U_i/N_i}}};\\
&& & s_{i_1} &=& d_{i_1};\\
\forall j&\in&J, &
  s_j &=& d_j;\\
&& & b &=& \ceil{\frac{N_{i_1}}{d_{i_1}U_{i_1}}}.
\end{IEEEeqnarray*}
Thus we have $\lambda_{i_1}=\lambda_j=1$ for all $j\in J$, and $\lambda_i,\lambda_h\ge\frac12$ for $i\in I$ and $h\in H$.

Recall the following inequalities from Definition~\ref{def:m-feasible} regarding $i\in I$:
\[
\frac1K\sqrt{\frac{N_i}{U_i}}
\le \tilde M
\le \left( \frac1K + \frac{1}{d_i} \right) \sqrt{\frac{N_i}{U_i}}.
\]
Thus, for any $i,j\in I$, we have:
\[
\sqrt{\frac{U_i/N_i}{U_j/N_j}} \le \frac{K}{d_i} + 1.
\]
Moreover, for $h\in H$ and $i_1\in I_1$:
\[
\left( \frac1K + \frac{\beta}{d_{i_1}} \right)
\le \tilde M
\le \frac1K\sqrt{\frac{N_h}{U_h}},
\]
and hence:
\[
\sqrt{\frac{U_{i_1}/N_{i_1}}{U_h/N_h}}
\ge \beta\frac{K}{d_{i_1}}+1.
\]

We will now verify that the parameters satisfy their constraints.
The parameters $t$, $s_{i_1}$, and $s_j$ trivially satisfy all their constraints.
Moreover, $N_{i_1}\ge KU_{i_1}\ge d_{i_1}U_{i_1}$, which implies $b\ge1$.
Regarding $s_h$:
\[
K/2 \ge 2\beta K \ge 2\beta\cdot D/\beta = 2D \ge d_h,
\]
and hence $d_h\le s_h\le \floor{K/2}$.
We are hence only left with $s_i$.

Consider the following:
\begin{IEEEeqnarray*}{rCl}
&& \gamma_1\beta\frac{\sqrt{N_i/U_i}}{\tilde M}
+ \gamma_2d_{i_1}\sqrt{\frac{U_{i_1}/N_{i_1}}{U_i/N_i}}\\
&\ge& \gamma_1\beta\frac{\sqrt{N_i/U_i}}{\tilde M}
\ge \frac{\gamma_1\beta}{\beta/d_i+1/K}\\
&\ge& \frac{\gamma_1\beta}{2\beta}d_i
\ge d_i,
\end{IEEEeqnarray*}
and hence $s_i\ge d_i$.
Furthermore,
\begin{IEEEeqnarray*}{rCl}
&& \gamma_1\beta\frac{\sqrt{N_i/U_i}}{\tilde M}
+ \gamma_2 d_{i_1}\sqrt{\frac{U_{i_1}/N_{i_1}}{U_i/N_i}}\\
&\le& \gamma_1\beta K
+ \gamma_2  d_{i_1} \left( \frac{K}{d_{i_1}} + 1 \right)\\
&=& \left[ \gamma_1 \beta + \gamma_2\left(1 + \frac{d_{i_1}}{K} \right)\right] K\\
&\le& \left[ \gamma_2 + (\gamma_1+\gamma_2)\beta \right] K\\
&\le& K/2,
\end{IEEEeqnarray*}
and hence $s_i=s_it\le \floor{K/2}$.

We will now evaluate the comparisons in \eqref{eq:comparison} to give bounds on the $A_i$ terms, for every $i$.

Let us start with $h\in H$.
We have $bs_h(s_ht-d_h+1)\le bs_h^2t=bs_h^2$, and:
\begin{IEEEeqnarray*}{rCl}
bs_h^2
&\le& \left( \frac{N_{i_1}}{d_{i_1}U_{i_1}} + 1 \right)
  \cdot 4\beta^2 K^2\\
&\le& \frac{N_{i_1}}{d_{i_1}U_{i_1}}
  \left( 1 + \frac{d_{i_1}}{K} \right)
  \cdot 4\beta^2 K^2\\
&\overset{(a)}{\le}& \frac{N_h}{d_{i_1}U_h}
  \frac{1+d_{i_1}/K}{(\beta+d_{i_1}/K)^2} \cdot \frac{d_{i_1}^2}{K^2}
  \cdot 4\beta^2 K^2\\
&\overset{(b)}{\le}& \frac{N_h}{U_h} \cdot \frac{1}{\beta^2} \cdot d_{i_1} \cdot 4\beta^2\\
&\le& 4D \cdot \frac{N_h}{U_h}.
\end{IEEEeqnarray*}
The inequality labeled $(a)$ comes from Definition~\ref{def:m-feasible} and \eqref{eq:refined-partition}, and:
\[
\left( \frac{\beta}{d_{i_1}}+\frac1K \right) \sqrt{\frac{N_{i_1}}{U_{i_1}}}
< M
\le \frac1K \sqrt{\frac{N_h}{U_h}}.
\]
As for the inequality $(b)$, it is due to:
\[
\forall x\ge0,
\quad \frac{1+x}{(\beta+x)^2} = \frac{1}{\beta+x} \cdot \frac{1+x}{\beta+x}
\le \frac1\beta\cdot\frac1\beta,
\]
because $x\to\frac{1+x}{\beta+x}$ is decreasing for $x\ge0$.
Hence,
\begin{IEEEeqnarray*}{rCl}
A_h
&\ge& \min\left\{ 1 , \frac{1}{4D} \right\}
  \cdot ( s_ht-d_h+1 ) U_h\\
&\ge& \frac{1}{4D}
  \cdot \left( 2\beta K - d_h \right) U_h\\
&\ge& \frac{1}{4D}
  \cdot \left( 2\beta - \beta \right) KU_h\\
&=& \frac{\beta}{4D} \cdot KU_h.
\end{IEEEeqnarray*}

Now we consider $i\in I\setminus I_1$.
We have:
\begin{IEEEeqnarray*}{rCl}
bs_i^2
&\le& \frac{N_{i_1}}{d_{i_1}U_{i_1}}
  \left( 1 + \frac{d_{i_1}}{K} \right)
  \left(
    \gamma_1\beta \frac{\sqrt{N_i/U_i}}{\tilde M}
	+ \gamma_2 d_{i_1} \sqrt{\frac{U_{i_1}/N_{i_1}}{U_i/N_i}}
  \right)^2\\
&\le& \frac{N_{i_1}}{d_{i_1}U_{i_1}}
  \left( 1 + \beta \right)
  \left(
    \gamma_1beta \frac{\sqrt{N_{i_1}/U_{i_1}}}{\tilde M}
	+ \gamma_2 d_{i_1}
  \right)^2
  \cdot \frac{U_{i_1}/N_{i_1}}{U_i/N_i}\\
&\le& \frac{N_i}{d_{i_1}U_i} (1+\beta)
  \left( \frac{\gamma_1\beta}{\beta/d_{i_1}+1/K}
    + \gamma_2 d_{i_1} \right)^2\\
&\le& \frac{N_i}{d_{i_1}U_i} (1+\beta)
  \left( \frac{\gamma_1\beta}{\beta}d_{i_1}
    + \gamma_2 d_{i_1} \right)^2\\
&=& \frac{N_i}{d_{i_1}U_i} (1+\beta) d_{i_1}^2
  \left( \gamma_1
    + \gamma_2 \right)^2\\
&\le& \frac{N_i}{U_i} \cdot 12 (1+\beta) D.
\end{IEEEeqnarray*}
Therefore, if we define $c=\frac{1}{12(1+\beta)D}$, then:
\begin{IEEEeqnarray*}{rCl}
A_i
&\ge& c
  \cdot (s_it-d_i+1)U_i\\
&\ge& c
  \cdot \left(
    \gamma_1\beta\frac{\sqrt{N_i/U_i}}{\tilde M}
	+ \gamma_2 d_{i_1} \sqrt{\frac{U_{i_1}/N_{i_1}}{U_i/N_i}}
	- d_i
  \right) U_i\\
&=& c \left[
  \left( \gamma_1\beta\frac{\sqrt{N_i/U_i}}{\tilde M}
    - d_i \right) U_i
  + \gamma_2 d_{i_1}\sqrt{\frac{U_{i_1}/N_{i_1}}{U_i/N_i}} \cdot U_i
\right]\\
&=& c \left[
  \left( 
    \gamma_1\beta
	- d_i\frac{\tilde M}{\sqrt{N_i/U_i}}
  \right) \frac{\sqrt{N_iU_i}}{\tilde M}
  + \gamma_2 d_{i_1}U_{i_1}
    \frac{\sqrt{N_iU_i}}{\sqrt{N_{i_1}U_{i_1}}}
\right]\\
&\ge& c \left[
  \left( 
    \gamma_1\beta
	- 2\beta
  \right) \frac{\sqrt{N_iU_i}}{\tilde M}
  + \gamma_2 d_{i_1}U_{i_1}
    \frac{\sqrt{N_iU_i}}{\sqrt{N_{i_1}U_{i_1}}}
\right]\\
&=& c \left[
  (\gamma_1-2) \beta \frac{\sqrt{N_iU_i}}{\tilde M}
  + \gamma_2 d_{i_1}U_{i_1}
    \frac{\sqrt{N_iU_i}}{\sqrt{N_{i_1}U_{i_1}}}
\right]\\
&\ge& c\cdot\min\left\{ \frac12\gamma_1-1,\gamma_2 \right\}\cdot\beta\\
&&{}
  \cdot \left[
  \frac{2\sqrt{N_iU_i}}{\tilde M}
  + \frac1\beta d_{i_1}U_{i_1}
    \frac{\sqrt{N_iU_i}}{\sqrt{N_{i_1}U_{i_1}}}
\right]\\
&=& 0.482 c\beta
  \cdot \left[
  \frac{2\sqrt{N_iU_i}}{\tilde M}
  + \frac1\beta d_{i_1}U_{i_1}
    \frac{\sqrt{N_iU_i}}{\sqrt{N_{i_1}U_{i_1}}}
\right].
\end{IEEEeqnarray*}

Finally, we have:
\[
bs_{i_1}(s_{i_1}t-d_{i_1}+1)
= bd_{i_1}
\ge \frac{N_{i_1}}{U_{i_1}},
\]
and, likewise for $j\in J$:
\[
bs_j(s_jt-d_j+1)
= bd_j
\ge \frac{d_j}{d_{i_1}} \frac{N_{i_1}}{U_{i_1}}
\ge \frac{N_j}{U_j},
\]
because of the regularity condition \eqref{eq:mu-reg2}.
Thus:
\[
A_{i_1} = \frac{N_{i_1}/d_{i_1}}{b};
\quad
A_j = \frac{N_j/d_j}{b}.
\]

Combining the $A_i$ terms, and keeping in mind the $\lambda_i$ terms, we get:
\begin{IEEEeqnarray*}{rCl}
R^\ast(M)
&\ge&
  \sum_{h\in H}
  	\frac12 \cdot \frac{\beta}{4D} \cdot KU_h\\
&&{} +
  \sum_{i\in I\setminus I_1}
    \frac12 \cdot \frac{0.482\beta}{12(1+\beta)D}
	\cdot \left[ 
	  \frac{2\sqrt{N_iU_i}}{\tilde M}
	  + \frac1\beta d_{i_1}U_{i_1} \frac{\sqrt{N_iU_i}}{N_{i_1}U_{i_1}}
	\right]\\
&&{} +
  \frac{ N_{i_1}/d_{i_1} + \sum_{j\in J} N_j/d_j - M }{b}\\
&\ge&
  \sum_{h\in H}
  	\frac{\beta}{8D} \cdot KU_h\\
&&{} +
    \frac{0.241\beta}{12(1+\beta)D}
	\left[ 
	  \frac{2S_IS_{I\setminus I_1}}{M-T_J+V_I}
	  + \frac1\beta d_{i_1}U_{i_1} \frac{S_{I\setminus I_1}}{N_{i_1}U_{i_1}}
	\right]\\
&&{} +
  \frac{ N_{i_1}/d_{i_1} + \sum_{j\in J} N_j/d_j - M }
       { 2N_{i_1}/(d_{i_1}U_{i_1}) }\\
&=&
  \sum_{h\in H}
  	\frac{\beta}{8D} \cdot KU_h\\
&&{} +
    \frac{0.241\beta}{12(1+\beta)D}
	\left[ 
	  \frac{2S_IS_{I\setminus I_1}}{M-T_J+V_I}
	  + \frac1\beta d_{i_1}U_{i_1} \frac{S_{I\setminus I_1}}{N_{i_1}U_{i_1}}
	\right]\\
&&{} +
  \frac12 U_{i_1}
    \left( 1 - \frac{M-T_J}{N_{i_1}/d_{i_1}} \right).
\IEEEyesnumber\label{eq:mu2-conv}
\end{IEEEeqnarray*}

Combining \eqref{eq:mu2-conv} with \eqref{eq:mu2-ach}, we get:
\begin{equation}
\label{eq:mu2-gap}
\frac{R(M)}{R^\ast(M)}
\le \max\left\{
  \frac{8D}{\beta},
  \frac{12(1+\beta)D}{0.241\beta},
  \frac{2D}{\beta}
\right\}
\le 9909 D.
\end{equation}

\subsubsection{Final gap}

By combining the inequalities from all the cases above, i.e., inequalities \eqref{eq:mu0-gap}, \eqref{eq:mu1a-gap}, \eqref{eq:mu1b-gap}, and \eqref{eq:mu2-gap}, we get the following final result:
\[
\frac{R(M)}{R^\ast(M)}
\le 9909 \cdot D,
\]
which concludes the proof of Theorem~\ref{thm:multi-user-gap}.

\subsection{Proofs of some useful results}
\label{app:multi-user-extras}

%%fakesection{Achievability rates}

\begin{proof}[Proof of Lemma~\ref{lemma:multi-user-achievability}]

%% Long equation for \Delta
\begin{figure*}[!t]
% ensure that we have normalsize text
\normalsize
% Store the current equation number.
\setcounter{MYtempeqncnt}{\value{equation}}
% TODO
% Set the equation number to one less than the one
% desired for the first equation here.
% The value here will have to be changed if equations
% are added or removed prior to the place these
% equations are referenced in the main text.
\setcounter{equation}{30}
\begin{IEEEeqnarray*}{rCl}
\Delta
&=& \left( 1-\frac{\alpha_iM}{N_i/d_i} \right)
- \left( 1 - \frac{M-T_J}{N_i/d_i} \right)
= \left( 1-\frac{ \frac{\sqrt{N_iU_i}}{S_I}(M-T_J+V_I) - \frac{N_i}{K} }{N_i/d_i} \right)
- \left( 1 - \frac{M-T_J}{N_i/d_i} \right)\\
&=& \frac{M-T_J}{N_i/d_i}
- \frac{\sqrt{N_iU_i}(M-T_J)}{S_IN_i/d_i}
- \frac{\sqrt{N_iU_i}V_I}{S_IN_i/d_i}
+ \frac{d_i}{K}
= \frac{M-T_J}{N_i/d_i} \left( 1 - \frac{\sqrt{N_iU_i}}{S_I} \right)
- \frac{\sqrt{N_iU_i}V_{I\setminus I_1}}{S_IN_i/d_i}
- \frac{\sqrt{N_iU_i}\cdot\frac{N_i}{K}}{S_IN_i/d_i}
+ \frac{d_i}{K}\\
&=& \frac{M-T_J}{N_i/d_i} \cdot \frac{S_{I\setminus I_1}}{S_I}
- \frac{d_iV_{I\setminus I_1}}{\sqrt{N_i/U_i}S_I}
+ \frac{d_i}{K} \left( 1 - \frac{\sqrt{N_iU_i}}{S_I} \right)
= \frac{d_i}{N_i}(M-T_J)\frac{S_{I\setminus I_1}}{S_I}
- \frac{d_iV_{I\setminus I_1}}{\sqrt{N_i/U_i}S_I}
+ \frac{d_i}{K} \cdot \frac{S_{I\setminus I_1}}{S_I}\\
&\overset{(a)}{\le}&
\frac{d_i}{N_i} \left[ \left( \frac1K+\frac{1}{d_i} \right) \sqrt{\frac{N_i}{U_i}}S_I - V_I \right] \frac{S_{I\setminus I_1}}{S_I}
- \frac{d_iV_{I\setminus I_1}}{\sqrt{N_i/U_i}S_I}
+ \frac{d_i}{K} \cdot \frac{S_{I\setminus I_1}}{S_I}\\
&=& \left( \frac{d_i}{K}+1 \right) \frac{S_{I\setminus I_1}}{\sqrt{N_iU_i}}
- \frac{d_iV_IS_{I\setminus I_1}}{N_iS_I}
- \frac{d_iV_{I\setminus I_1}\sqrt{N_iU_i}}{N_iS_I}
+ \frac{d_iS_{I\setminus I_1}}{KS_I}\\
&=& \left( \frac{d_i}{K}+1 \right) \frac{S_{I\setminus I_1}}{\sqrt{N_iU_i}}
- \frac{d_i \left( V_{I\setminus I_1}+\frac{N_i}{K} \right) S_{I\setminus I_1}}{N_iS_I}
- \frac{d_iV_{I\setminus I_1}\sqrt{N_iU_i}}{N_iS_I}
+ \frac{d_iS_{I\setminus I_1}}{KS_I}\\
&=& \left( \frac{d_i}{K}+1 \right) \frac{S_{I\setminus I_1}}{\sqrt{N_iU_i}}
- \frac{d_i V_{I\setminus I_1} S_I}{N_iS_I}
- \frac{d_iS_{I\setminus I_1}}{KS_I}
+ \frac{d_iS_{I\setminus I_1}}{KS_I}
= \frac{d_i}{K} \cdot \frac{S_{I\setminus I_1}}{\sqrt{N_iU_i}}
- \frac{d_iV_{I\setminus I_1}}{N_i}
+ \frac{S_{I\setminus I_1}}{\sqrt{N_iU_i}}\\
&=& \frac{d_i}{K} \sum_{i\in I\setminus I_1} \left( 
\frac{\sqrt{N_iU_i}}{\sqrt{N_iU_i}} - \frac{N_i}{N_i}
\right)
+ \frac{S_{I\setminus I_1}}{\sqrt{N_iU_i}}
= \frac{d_i}{K} \sum_{i\in I\setminus I_1}
\frac{ \sqrt{N_iU_i} + N_i\sqrt{U_i/N_i} }{ \sqrt{ N_iU_i } }
+ \frac{S_{I\setminus I_1}}{\sqrt{N_iU_i}}\\
&=& \frac{d_i}{K} \sum_{i\in I\setminus I_1}
\frac{ N_i\left( \sqrt{U_i/N_i} - \sqrt{N_i/U_i} \right) }{\sqrt{N_iU_i}}
+ \frac{S_{I\setminus I_1}}{\sqrt{N_iU_i}}\\
&\overset{(b)}{\le}& \frac{S_{I\setminus I_1}}{\sqrt{N_iU_i}}.
\IEEEyesnumber\label{eq:delta}
\end{IEEEeqnarray*}
% Restore the current equation number.
\setcounter{equation}{\value{MYtempeqncnt}}
% IEEE uses as a separator
\hrulefill
% The spacer can be tweaked to stop underfull vboxes.
\vspace*{4pt}
\end{figure*}
\addtocounter{equation}{1}
%% End of \Delta equation

We will prove the lemma for each of the five sets in the refined $M$-feasible partition.
First, recall that the achievable rate $R^\text{SL}(\cdot)$ for the single-level setup is defined in Theorem~\ref{thm:single-level}.

For $h\in H$:
\[
R_h(M) = R^\text{SL}(0,K,N_h,U_h,d_h) = KU_h.
\]

For $i\in I_0$:
\begin{IEEEeqnarray*}{rCl}
R_i(M)
&=& R^\text{SL}(\alpha_iM,K,N_i,U_i,d_i)\\
&\le& R^\text{SL}(0,K,N_i,U_i,d_i)\\
&=& KU_i\\
&\le& \frac{2}{\tilde M}\sqrt{\frac{N_i}{U_i}}\cdot U_i\\
&=& \frac{2S_I\sqrt{N_iU_i}}{M-T_J+V_I}.
\end{IEEEeqnarray*}

For $i\in I'$:
\begin{IEEEeqnarray*}{rCl}
R_i(M)
&=& R^\text{SL}(\alpha_iM,K,N_i,U_i,d_i)\\
&\le& U_i\cdot \frac{N_i}{\alpha_iM} \cdot \left( 1 - \frac{\alpha_iM}{N_i/d_i} \right)\\
&\le& \frac{N_iU_i}{\alpha_iM}\\
&=& \frac{N_iU_i}{\sqrt{N_iU_i}\cdot\tilde M - N_i/K}\\
&\le& \frac{N_iU_i}{\sqrt{N_iU_i}\cdot\tilde M - \frac12\sqrt{N_iU_i}\cdot\tilde M}\\
&=& \frac{2S_I\sqrt{N_iU_i}}{M-T_J+V_I}.
\end{IEEEeqnarray*}

For $i\in I_1$:
\[
R_i(M) \le U_i\cdot\frac{N_i}{\alpha_iM}\cdot\left( 1-\frac{\alpha_iM}{N_i/d_i} \right)
\le \frac1\beta d_iU_i \left( 1 - \frac{\alpha_iM}{N_i/d_i} \right),
\]
since $i\in I_1\implies \alpha_iM\ge\beta N_i/d_i$ (follows from \eqref{eq:refined-partition}).

Consider the following difference:
\[
\Delta = \left( 1-\frac{\alpha_iM}{N_i/d_i} \right)
- \left( 1 - \frac{M-T_J}{N_i/d_i} \right).
\]
We can show that:
\[
\Delta \le \frac{S_{I\setminus I_1}}{\sqrt{N_iU_i}}.
\]
The full derivation is given in \eqref{eq:delta}.
In the calculations in \eqref{eq:delta}, $(a)$ is due to the definition of $I_1$, while $(b)$ is due to the fact that the (unique) level in $I_1$ is more popular than any other level in $i'\in I$, and thus:
\[
\sqrt{\frac{U_{i'}}{N_{i'}}} \le \sqrt{\frac{U_i}{N_i}}.
\]

Using the bound on $\Delta$, we can rewrite $R_i(M)$ as follows:
\begin{IEEEeqnarray*}{rCl}
R_i(M)
&\le& \frac1\beta d_iU_i \left( 1 - \frac{\alpha_iM}{N_i/d_i} \right)\\
&=& \frac1\beta d_iU_i \left( 1 - \frac{M-T_J}{N_i/d_i} + \Delta \right)\\
&\le& \frac1\beta d_iU_i \left( 1 - \frac{M-T_J}{N_i/d_i} \right)
+ \frac1\beta d_iU_i \cdot \frac{S_{I_0}+S_{I'}}{\sqrt{N_iU_i}}.
\end{IEEEeqnarray*}

Finally, for $j\in J$:
\[
R_j(M) = R^\text{SL}(N_j/d_j,K,N_j,U_j,d_j) = 0.
\]

This completes the proof of the lemma.
\end{proof}

%%fakesection{M-feasble construction}

\begin{proof}[Proof of Lemma~\ref{lemma:m-feasible-existence}]

We prove the existence of an $M$-feasible partition by construction.
In fact, we will use Algorithm~\ref{alg:m-feasible} presented in Section~\ref{sec:multi-user-achievability} as the proof.
In what follows, we prove the correctness of the algorithm, and then conclude the proof of Lemma~\ref{lemma:m-feasible-existence}.

The idea behind the algorithm is as follows.
When the memory is $0$, all levels are in the set $H$.%
\footnote{Technically, there will be exactly one level in the set $I$, but it will still get zero memory so this distinction is irrelevant.}
As the memory is increased, levels will start moving from $H$ to $I$ as they gain more memory, and then finally from $I$ to $J$ as they get completely stored in the caches.
Thus, the range of memory values $[0,\infty)$ can be divided into $2L+1$ intervals; in each interval, the partition $(H,I,J)$ will be the same.
The boundaries of these intervals are the memory values at which levels switch from one set to the next.
\figurename~\ref{fig:algorithm-example} shows two examples of how the $(H,I,J)$ partition might evolve with two levels.

\begin{figure}
\centering
\includegraphics[width=\myfigswidth]{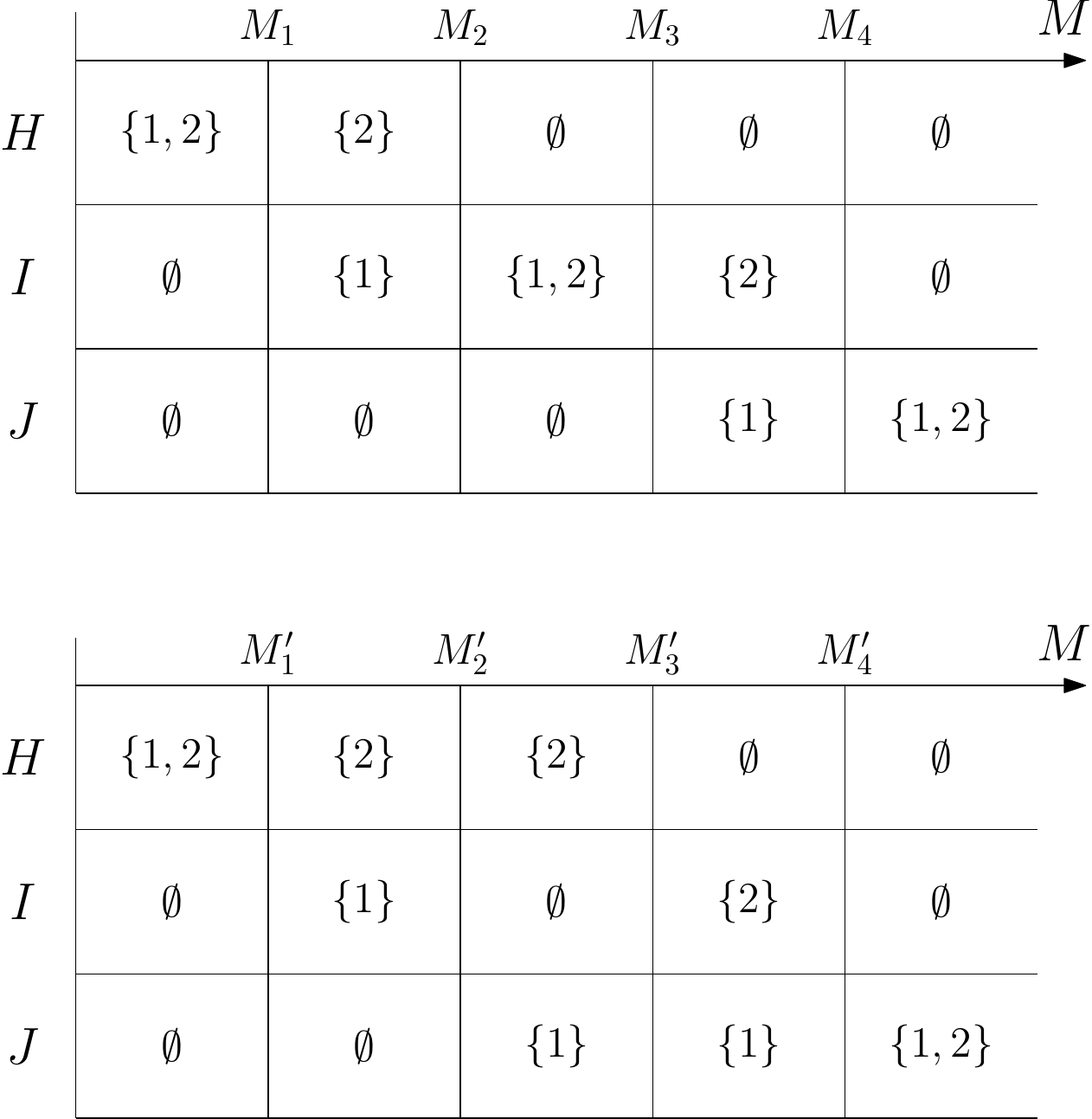}
\caption{Two examples of how an $M$-feasible partition $(H,I,J)$ could evolve when there are two levels $\{1,2\}$.
In the second case, it will turn out that $M_2'=M_3'$.
The reasoning behind this is that between $M_2'$ and $M_3'$, all levels have a fixed memory: level $2$ has memory $0$ while level $1$ has memory $N_1/d_1$.
Thus an increase in the overall memory between $M_2'$ and $M_3'$ would be wasted.
Moreover, $M_1=M_1'=0$ for similar reasons, and $M_4=M_4'=N_1/d_1+N_2/d_2$ because that is the point at which both levels can be completely stored.}
\label{fig:algorithm-example}
\end{figure}

The algorithm operates in three main steps.
In the first step, the sequence of $(H,I,J)$ partitions is determined.
For instance, we determine which of the two cases illustrated in \figurename~\ref{fig:algorithm-example} holds.
In the second step, the values of the boundaries between the different intervals are calculated; in \figurename~\ref{fig:algorithm-example}, those would be $\{M_1,M_2,M_3,M_4\}$ or $\{M_1',M_2',M_3',M_4'\}$.
Finally, the third step consists in determining the $M$-feasible partition $(H,I,J)$ based on steps 1 and 2, and based on the value of any given memory $M$.
For example, if we are in the first case in \figurename~\ref{fig:algorithm-example} and we are given a memory value between $M_3$ and $M_4$, then we know that $(H,I,J)=(\emptyset,\{2\},\{1\})$.

Recall from Definition~\ref{def:m-feasible} that the $M$-feasible partition is defined by the relation of $(M-T_J+V_I)/S_I$ to the two values $(1/K)\sqrt{N_i/U_i}$ and $(1/d_i+1/K)\sqrt{N_i/U_i}$, for every level $i$.
As long as the same relations are maintained, the same partition is chosen as the $M$-feasible partition.
Thus, the values of the boundary ranges must be of one of two forms:
\begin{IEEEeqnarray*}{rCl}
m_i^{I,J} &=& \frac1K\sqrt{\frac{N_i}{U_i}} \cdot S_I + T_J - V_I;\\
M_i^{I,J} &=& \left(\frac{1}{d_i}+\frac1K\right)\sqrt{\frac{N_i}{U_i}} \cdot S_I + T_J - V_I,
\end{IEEEeqnarray*}
for the apropriate partition $(H,I,J)$.
In fact, this partition has to be the one defined by either of the intervals that this boundary divides.
For example, $M_2$ in \figurename~\ref{fig:algorithm-example} is the boundary where level $2$ moves from $H$ to $I$.
Thus, it is of the form $m_2^{I,J}$, where $(H,I,J)=(\emptyset,\{1,2\},\emptyset)$, or $(H,I,J)=(\{2\},\{1\},\emptyset)$.
As it turns out, both choices give the same result.

Let us suppose we are looking at a particular interval with some $(H,I,J)$ partition, and we wish to know its upper boundary.
We know that it has to be either $m_h^{I,J}$, $h\in H$, or $M_i^{I,J}$, $i\in I$.
Specifically, it has to be the smallest of these values.
Since $I$ and $J$ are fixed, this is equivalent to determining the smaller of:
\begin{IEEEeqnarray*}{rCl}
\tilde m_h &=& \frac1K\sqrt{\frac{N_h}{U_h}};\\
\tilde M_i &=& \left( \frac{1}{d_i}+\frac1K \right)\sqrt{\frac{N_i}{U_i}}.
\end{IEEEeqnarray*}
But these values do not depend on the chosen partition $(H,I,J)$!
Thus, the sequence of $(H,I,J)$ can be uniquely determined solely by these $\tilde m_i$ and $\tilde M_i$ values.

For what follows, we refer the reader to Algorithm~\ref{alg:m-feasible} for the various symbols ($x_t$, $Y_t$).
It is worth mentioning the following two points.
First, the ordering between $x_t$ and $Y_t$ is preserved.
That is, $x_t<x_{t'} \implies Y_t<Y_{t'}$.
This is ensured by the observation made above that the ordering does not depend on the partition $(H,I,J)$.
Second, as mentioned earlier, when applying the function $x\mapsto x\cdot S_I+T_J-V_I$ on some $x=\tilde m_i$ or $x=\tilde M_i$, it does not matter whether we take the $(H,I,J)$ partition whose interval $x$ lower-bounds or upper-bounds.

Any interval in which $I$ is set to be empty automatically becomes an empty interval.
To prove this, we will show that $I=\emptyset$ implies that its upper and lower boundaries match.
This is done as follows.
Consider the aforementioned interval, and notice that it can only occur if some level moved from $I$ to $J$ at the lower boundary, and another level moved from $H$ to $I$ at the upper boundary.
See the bottom case of \figurename~\ref{fig:algorithm-example} for an example, interval $(M_2',M_3')$.
Thus, the sequence of $(H,I,J)$ partitions that occurs is of the form shown in \figurename~\ref{fig:algorithm-emptyI}.

\begin{figure}
\centering
\includegraphics[width=.3\textwidth]{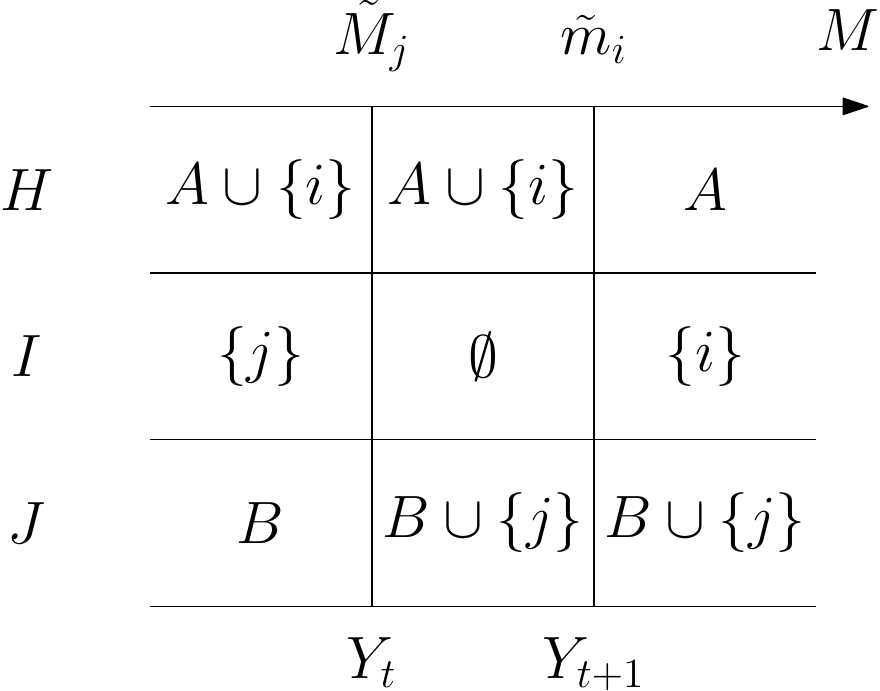}
\caption{The only situation where Algorithm~\ref{alg:m-feasible} results in an interval with $I=\emptyset$.
Here, the sets $A$, $B$, $\{i\}$, and $\{j\}$ are disjoint and together form the entire set of levels.}
\label{fig:algorithm-emptyI}
\end{figure}

Let us compute the interval boundaries $Y_t$ and $Y_{t+1}$.
Because of our previous observations, we can use $(H,I,J)=(A\cup\{i\},\emptyset,B\cup\{j\})$ for both.
Then,
\[
Y_t
= \tilde M_j \cdot S_I + T_J - V_I
= T_J
= T_B + N_j/d_j.
\]
Similarly,
\[
Y_{t+1}
= \tilde m_i \cdot S_I + T_J - V_I
= T_J
= T_B + N_j/d_j.
\]
Thus, $Y_t=Y_{t+1}$ and the interval is empty.

This concludes the proof of the lemma.
\end{proof}

\section{Proof of approximate optimality for the single-user setup (Theorem~\ref{thm:single-user-gap})}
\label{app:single-user}
\label{app:single-user-gap}

As discussed in Section~\ref{sec:single-user-gap}, we prove Theorem~\ref{thm:single-user-gap} by lower-bounding the optimal rate with expressions of the form shown in \eqref{eq:single-user-lower-bounds}.
Two cases must be considered.

\subsubsection{Case $M<1/6$}

When $M$ is this small, we choose $b=1$ broadcast message.
Recall that, because of regularity condition \eqref{eq:su-reg}, we have $N_i/K_i\ge1>1/6>M$.
The achievable rate in this case can be upper-bounded by the expression in \eqref{eq:su-ub-0}.

Consider now any level $i$.
Let $s_i=K_i$.
Then,
\begin{IEEEeqnarray*}{rCl}
v_i
&=& s_i\left( \min\left\{ 1 , \frac{N_i}{s_ib} \right\} - \frac{M}{b} \right)\\
&=& K_i \left( \min\left\{ 1 , \frac{N_i}{K_i} \right\} - M \right)\\
&\ge& (5/6)K_i. \IEEEyesnumber \label{eq:su-lb-gh}
\end{IEEEeqnarray*}

We can combine \eqref{eq:su-ub-0} with \eqref{eq:su-lb-gh} and \eqref{eq:single-user-lower-bounds} to get:
\begin{equation}
\label{eq:su-gap-0}
\frac{R(M)}{R^\ast(M)} \le \frac65.
\end{equation}

\subsubsection{Case $M\ge1/6$}

We will now choose $b=\ceil{6M}\ge1$.

\paragraph{Bound for $g\in G$}
Consider $s_g=1$.
Then,
\begin{IEEEeqnarray*}{rCl}
v_g
&=& \min\left\{ 1 , \frac{N_i}{\ceil{6M}} \right\} - \frac{M}{\ceil{6M}}\\
&\ge& \frac12 - \frac{M}{6M}\\
&=& \frac13, \IEEEyesnumber \label{eq:su-lb-g}
\end{IEEEeqnarray*}
because $s_gb=\ceil{6M}\le2\cdot6M\le2N_g$.

\paragraph{Bound for $h\in H$}
Consider $s_h=\ceil{K_h/6}\ge1$.
Then,
\begin{IEEEeqnarray*}{rCl}
v_h
&\ge& \frac{K_h}{6} \left( \min\left\{ 1 , \frac{N_h}{4K_hM} \right\} - \frac{M}{6M} \right)
= \frac{1}{72}\cdot K_h, \IEEEyesnumber \label{eq:su-lb-h}
\end{IEEEeqnarray*}
because $s_hb=\ceil{K_h/6}\cdot\ceil{6M} \le 4K_hM \le 4N_h$.

\paragraph{Bound for $i\in I$}
Consider $s_i=\ceil{N_i/6M}$.
Then,
\begin{IEEEeqnarray*}{rCl}
v_i
&\ge& \frac{N_i}{6M} \left( \min\left\{ 1 , \frac{N_i}{4N_i} \right\} - \frac{M}{6M} \right)
= \frac{1}{72}\cdot\frac{N_i}{M}, \IEEEyesnumber \label{eq:su-lb-i}
\end{IEEEeqnarray*}
because $s_ib=\ceil{N_i/6M}\cdot\ceil{6M} \le 4N_i$.

\paragraph{Bound for $J$}
First, if $M\ge N_J$, then the set $J$ contributes nothing to the upper bound on the rate in \eqref{eq:single-user-achievable-rate}; see \eqref{eq:su-ub-J}.
Thus we can ignore it, i.e., choose $s_J=0$ and thus $n_J=0$ and $v_J=0$.

So the interesting case is $M<N_J$.
Here, we must decode files from multiple levels collectively.
Consider $s_J=\ceil{N_J/6M}$, where $N_J=\sum_{j\in J}N_j$.
Notice that there are enough users and broadcasts to decode all files, because:
\[
s_Jb \ge \frac{N_J}{6M}\cdot6M = N_J.
\]
However, we must take care that no broadcast considers more than $K_j$ users at a time for any $j\in J$.
This can be ensured: since there are $b=\ceil{6M}$ broadcasts, and $b\ge N_j$ for all $j\in J$, then every broadcast need only consider at most one user per level.
Hence, all of the $N_J$ files can be decoded, and $n_J=N_J$.

If $M<N_J/6$, we have:
\begin{IEEEeqnarray*}{rCl}
v_J
&=& \frac{N_J - s_JM}{b}\\
&\ge& \frac{1}{12M}\left( N_J - \frac{N_J}{12M}\cdot M \right)\\
&\ge& \frac{144}{11}\cdot\frac{N_J}{M}.
\IEEEyesnumber \label{eq:su-lb-j1}
\end{IEEEeqnarray*}

If $N_J/6\le M<N_J$, then $s_J$ is actually equal to $1$, and:
\begin{IEEEeqnarray*}{rCl}
v_J
&=& \frac{N_J - M}{b}\\
&\ge& \frac{N_J - M}{12M}\\
&\ge& \frac{1}{12}\left( 1 - \frac{M}{N_J} \right).
\IEEEyesnumber\label{eq:su-lb-j2}
\end{IEEEeqnarray*}

\subsubsection{Multiplicative gap}

By combining \eqref{eq:su-lb-g}, \eqref{eq:su-lb-h}, \eqref{eq:su-lb-i}, \eqref{eq:su-lb-j1} and \eqref{eq:su-lb-j2} with \eqref{eq:single-user-achievable-rate} and \eqref{eq:su-ub-J}, and also taking into account \eqref{eq:su-gap-0}, we get:
\[
\frac{R(M)}{R^\ast(M)} \le 72,
\]
which concludes the proof of Theorem~\ref{thm:single-user-gap}.\qed

\section{Complete characterization for the small example (Theorem~\ref{thm:small-example})}
\label{app:small-example}

This appendix proves Theorem~\ref{thm:small-example} in two parts.
The first part presents the achievability scheme, and the second part gives the information-theoretic outer bounds.

The rate-memory curve shown in Figure~\ref{fig:small-ex-region}---which we are trying to prove is both achievable and optimal---can be described using the following equation:
\begin{equation}
\label{eq:small-example-rate}
R = \max\left\{ 3-2M \,,\, \frac52-M \,,\, 2-\frac12M \,,\, 1 - \frac{M-2}{N_2/2} \right\}.
\end{equation}

We would like to remind the reader that, in the example considered, the number of less popular files is $N_2\ge4$.

\subsection{Achievability scheme}

To prove the achievability of the piece-wise linear curve in Figure~\ref{fig:small-ex-region} and in \eqref{eq:small-example-rate}, we need only show the achievability of the corner points, as the rest can be achieved using a memory-sharing scheme between every pair of points.
These $(M,R)$ points are:
\[
\left(0,3\right);\quad \left(\frac12,2\right);\quad \left(1,\frac23\right);\quad \left(2,1\right);\quad \text{and } \left(2+\frac{N_2}{2},0\right).
\]

\subsubsection{Point $(M,R)=(0,3)$}
When $M=0$, we do not place any content in the caches, and instead broadcast all the requested files to the users.
Since these can be three different files, the peak rate is $R=3$.

\subsubsection{Point $(M,R)=(2+\frac{N_2}{2},0)$}
On the other extreme, when $M=2+\frac{N_2}{2}$, the caches are large enough to each store all the popular files, and complementary halves of the unpopular files.
Thus, a user can recover any popular file by solely accessing any cache, and can recover any unpopular file by accessing both caches.
Therefore, no broadcast is needed, and thus $R=0$ is achievable.

\subsubsection{Point $(M,R)=(2,1)$}
In this case, each cache can hold up to two files.
Our strategy is to dedicate this memory to the two popular files.
Thus, each cache contains all the popular files, but no information about the unpopular files.
Therefore, users 1 and 2 can recover their respective requested files without relying on any broadcast, whereas user 3 requires a full broadcast of his requested file.
As a result, the rate $R=1$ is achievable.

\subsubsection{Point $(M,R)=(1,\frac32)$}
Let us split each of the popular files into two parts of equal size: $W^1_1=(W^1_{1a},W^1_{1b})$ and $W^1_2=(W^1_{2a},W^1_{2b})$.
We place these in the caches as follows.
The first cache contains the first half of each file: $Z_1=(W^1_{1a},W^1_{2a})$, whereas the second cache contains the second half of each file: $Z_2=(W^1_{1b},W^1_{2b})$.
When the user requests are revealed, the BS sends a common message with two parts: $X^\mathbf{r}=(X^\mathbf{r}_1,X^\mathbf{r}_2)$.
The first part directly serves user 3 by giving him the full file he requested: $X^\mathbf{r}_1=W^2_{r_3}$.
The second part will be $X^\mathbf{r}_2=W^1_{r_1b}\oplus W^1_{r_2a}$, which, along with the AP cache content accessed by each of users 1 and 2, allows them to recover their respective file.
The BS had to transmit one full file, plus a linear combination of two half files, and, as a result, the total broadcast rate is $R=\frac32$.

\subsubsection{Point $(M,R)=(\frac12,2)$}
This case is slightly different in that it requires storing coded content in the caches.
We again split the popular files into two halves as before.
However, we this time store the following: $Z_1=W^1_{1a}\oplus W^1_{2a}$; and $Z_2=W^1_{1b}\oplus W^1_{2b}$.
The BS will transmit the following broadcast: $X^\mathbf{r}=\left(W^2_{r_3},W^1_{r_1b},W^1_{r_2a}\right)$.
This will serve all the requests of the users.
The broadcasts message consists of a full file and two half-files, and thus the total rate is $R=2$.

\subsection{Outer bounds}

We will now prove that the above scheme is optimal with respect to information-theoretic bounds.
We do that by showing that the rate is larger than each one of the expressions in the maximization in \eqref{eq:small-example-rate}.

In all of the following, inequalities marked by $(\ast)$ are due to Fano's inequality, and the $\epsilon_F$ term that arises along with these inequalities is a term that decays to zero as $F\to\infty$.

\subsubsection{First expression}
Let the request vector be $\mathbf{r}=(1,2,1)$, and consider both caches along with the broadcast $X^\mathbf{r}$.
Then, for all $F$,
\begin{IEEEeqnarray*}{rCl}
RF + 2MF
&\ge& H\left( Z_1, Z_2, X^\mathbf{r} \right)\\
&=& H\left( Z_1, Z_2, X^{(1,2,1)} \left| W^1_1,W^1_2,W^2_1 \right.\right)\\
&&{} + I\left( W^1_1,W^1_2,W^2_1 ; Z_1, Z_2, X^{(1,2,1)} \right)\\
&\overset{(\ast)}{\ge}& 3F(1-\epsilon_F)\\
R+2M &\ge& 3(1-\epsilon_F).
\end{IEEEeqnarray*}
By taking $F\to\infty$, we get:
\[
R \ge 3-2M.
\]

\subsubsection{Second expression}
Consider the two request vectors $\mathbf{r}_1=(1,2,1)$ and $\mathbf{r}_2=(2,1,2)$.
\begin{IEEEeqnarray*}{rCl}
2\left(RF+MF\right)
&\ge& H\left(Z_1,X^{\mathbf{r}_1}\right)
+ H\left(Z_2,X^{\mathbf{r}_2}\right)\\
&=& H\left(Z_1,X^{\mathbf{r}_1} \left| W^1_1 \right.\right)
+ I\left(W^1_1;Z_1,X^{\mathbf{r}_1}\right)\\
&&{} + H\left(Z_2,X^{\mathbf{r}_2} \left| W^1_1 \right.\right)
+ I\left(W^1_1;Z_2,X^{\mathbf{r}_2}\right)\\
&\overset{(\ast)}{\ge}& H\left(Z_1,Z_2,X^{\mathbf{r}_1},X^{\mathbf{r}_2}\left|W^1_1\right.\right)
+ 2F(1-\epsilon_F)\\
&=& H\left( Z_1,Z_2,X^{\mathbf{r}_1},X^{\mathbf{r}_2} \left| W^1_1,W^1_2,W^2_1,W^2_2 \right.\right)\\
&&{} + I\left( W^1_2,W^2_1,W^2_2 ; Z_1,Z_2,X^{\mathbf{r}_1},X^{\mathbf{r}_2} \left| W^1_1 \right.\right)\\
&&{} + 2F(1-\epsilon_F)\\
&\overset{(\ast)}{\ge}& 5F(1-\epsilon_F)\\
2R+2M &\ge& 5(1-\epsilon_F).
\end{IEEEeqnarray*}
By taking $F\to\infty$, we get:
\[
R \ge \frac52 - M.
\]

\subsubsection{Third expression}
Consider the following four request vectors: $\mathbf{r}_1=(1,2,1)$, $\mathbf{r}_2=(2,1,2)$, $\mathbf{r}_3=(1,2,3)$, and $\mathbf{r}_1=(2,1,4)$.
For clarity, define $\mathcal{X}=(X^{\mathbf{r}_1},X^{\mathbf{r}_2},X^{\mathbf{r}_3},X^{\mathbf{r}_4})$ and $\mathcal{W}^2=(W^2_1,W^2_2,W^2_3,W^2_4)$.
\begin{IEEEeqnarray*}{rCl}
2(2RF+MF)
&\ge& H\left( Z_1,X^{\mathbf{r}_1},X^{\mathbf{r}_2} \right)
+ H\left( Z_2,X^{\mathbf{r}_3},X^{\mathbf{r}_4} \right)\\
&=& H\left( Z_1,X^{\mathbf{r}_1},X^{\mathbf{r}_2} \left| W^1_1,W^1_2 \right.\right)\\
&&{} + I\left( W^1_1,W^1_2 ; Z_1,X^{\mathbf{r}_1},X^{\mathbf{r}_2} \right)\\
&&{} + H\left( Z_2,X^{\mathbf{r}_3},X^{\mathbf{r}_4} \left| W^1_1,W^1_2 \right.\right)\\
&&{} + I\left( W^1_1,W^1_2 ; Z_2,X^{\mathbf{r}_3},X^{\mathbf{r}_4} \right)\\
&\overset{(\ast)}{\ge}& H\left( Z_1,Z_2,\mathcal{X} \left| W^1_1,W^1_2 \right.\right)\\
&&{} + 2\cdot 2F(1-\epsilon_F)\\
&=& H\left( Z_1,Z_2,\mathcal{X} \middle| W^1_1,W^1_2,\mathcal{W}^2 \right)\\
&&{} + I\left( \mathcal{W}^2 ; Z_1,Z_2,\mathcal{X} \left| W^1_1,W^1_2 \right.\right)\\
&&{} + 4F(1-\epsilon_F)\\
&\overset{(\ast)}{\ge}& 8F(1-\epsilon_F)\\
4R+2M &\ge& 8(1-\epsilon_F).
\end{IEEEeqnarray*}
By taking $F\to\infty$, we get:
\[
R \ge 2 - \frac12 M.
\]

\subsubsection{Fourth expression}
Consider the $N_2$ request vectors $\mathbf{r}_1,\ldots,\mathbf{r}_{N_2}$, such that:
\begin{align*}
\mathbf{r}_k &= (1,2,k) \qquad \forall k\le \floor{\frac{N_2}{2}},\\
\mathbf{r}_k &= (2,1,k) \qquad \forall k\ge \floor{\frac{N_2}{2}}+1.\\
\end{align*}
Define $\mathcal{X}^e=\{X^{\mathbf{r}_k}:\text{ $k$ is even} \}$ and $\mathcal{X}^o=\{X^{\mathbf{r}_k}:\text{ $k$ is odd} \}$.
Thus, each of $\mathcal{X}^e$ and $\mathcal{X}^o$ contains at least one broadcast $X^{\mathbf{r}_k}$ such that $\mathbf{r}_k=(1,2,k)$, and one broadcast $X^{\mathbf{r}_l}$ such that $\mathbf{r}_l=(2,1,l)$.
Furthermore, define, for clarity, $\mathcal{W}^2=(W^2_1,\ldots,W^2_{N_2})$.
Then,
\begin{IEEEeqnarray*}{rCl}
N_2RF+2MF
&\ge& H\left(Z_1,\mathcal{X}^e\right)
+ H\left(Z_2,\mathcal{X}^o\right)\\
&=& H\left(Z_1,\mathcal{X}^e\left|W^1_1,W^1_2\right.\right)\\
&&{} + I\left(W^1_1,W^1_2;Z_1,\mathcal{X}^e\right)\\
&&{} + H\left( Z_2,\mathcal{X}^o \left| W^1_1,W^1_2 \right.\right)\\
&&{} + I\left(W^1_1,W^1_2;Z_2,\mathcal{X}^o\right)\\
&\overset{(\ast)}{\ge}& H\left(Z_1,Z_2,\mathcal{X}^e,\mathcal{X}^o \left| W^1_1,W^1_2 \right.\right)\\
&&{} + 4F(1-\epsilon_F)\\
&=& H\left(Z_1,Z_2,\mathcal{X}^e,\mathcal{X}^o \left| W^1_1,W^1_2,\mathcal{W}^2 \right.\right)\\
&&{} + I\left( \mathcal{W}^2 ; Z_1,Z_2,\mathcal{X}^e,\mathcal{X}^o \left| W^1_1,W^1_2 \right.\right)\\
&&{} + 4F(1-\epsilon_F)\\
&\overset{(\ast)}{\ge}& (4+N_2)F(1-\epsilon_F)\\
N_2R+2M &\ge& (4+N_2)(1-\epsilon_F).
\end{IEEEeqnarray*}
By taking $F\to\infty$, we get:
\[
R \ge \frac{4 + N_2 - 2M}{N_2} = 1 - \frac{M-2}{N_2/2}.
\]

\end{document}